\documentclass[a4paper,11pt,reqno]{amsart} 
\usepackage{amsmath}
\usepackage{amsfonts}
\usepackage[T1]{fontenc}  
\usepackage[utf8]{inputenc}
\usepackage[mathscr]{euscript}
\usepackage{verbatim}
\usepackage{amssymb,latexsym}
\usepackage{amsthm}
\usepackage{color}
\usepackage{graphicx}
\usepackage[hidelinks]{hyperref}
\usepackage[font=scriptsize]{caption}
\usepackage{bbm}
\usepackage{amsmath,accents}
\usepackage{subcaption}
\usepackage{float}

\newcommand{\Ab}{\mathbf{A}}
\newcommand{\Eb}{\mathbf E}

\usepackage[lmargin=3.5 cm,rmargin=3.5 cm,tmargin=3.5cm,bmargin=2.5cm,paper=a4paper]{geometry}
\DeclareMathOperator{\curl}{curl}\DeclareMathOperator{\Div}{div}
 \DeclareMathOperator{\dist}{dist} 
\DeclareMathOperator{\supp}{supp} \DeclareMathOperator{\dom}{\mathrm {Dom}}
\DeclareMathOperator{\bd}{\mathrm{bnd}} \DeclareMathOperator{\br}{\mathrm{bar}}\DeclareMathOperator{\T}{\mathrm{T}}\DeclareMathOperator{\blk}{\mathrm{int}}
\DeclareMathOperator{\re}{\mathrm{Re}}
\DeclareMathOperator{\csch}{\mathrm{csch}}
\DeclareMathOperator{\spc}{\mathrm{sp}}

\DeclareMathOperator{\ImP}{\mathrm{Im}\Phi}
\DeclareMathOperator{\im}{\mathrm{Im}}
\DeclareMathOperator{\p}{\mathrm{P}}
\DeclareMathOperator{\q}{\mathrm{Q}}
\newtheorem{thm}{Theorem}[section]

\newtheorem{theorem}[thm]{Theorem}
\newtheorem{assumption}[thm]{Assumption}
\newtheorem{notation}[thm]{Notation}

\newtheorem{lemma}[thm]{Lemma}
\newtheorem{proposition}[thm]{Proposition}

\theoremstyle{remark}
\newtheorem{rem}[thm]{Remark}

\newcommand{\nb}{\nabla}
\newcommand{\R}{\mathbb{R}}
\newcommand{\Fb}{\mathbf{F}}

\newcommand{\N}{\mathbb{N}}
\newcommand{\C}{\mathbb{C}}

\newcommand{\Hd}{H_{{\rm div}}^1(\Omega)}
\newcommand{\Om}{\Omega}
\newcommand{\kp}{\kappa}
\newcommand{\kn}{\nabla-i\kappa H{\bf A}}
\newcommand{\Es}{{\rm E}_{\rm g.st}(\kappa,  H)}
\newcommand{\GL}{\mathcal E_{\kappa,H}}

\def\sig#1{\vbox{\hsize=5.5cm
		\kern2cm\hrule\kern1ex
		\hbox to \hsize{\strut\hfil #1 \hfil}}}
\newcommand\signatures[4]{%
	\vspace{3cm}
	\hbox to \hsize{\hfil #1, \today\hfil}
	\vspace{3cm}
	\hbox to \hsize{\quad#2\hfil\hfil #3\quad}
	\vspace{3cm}
	\hbox to \hsize{\hfil#4\hfil}}
\numberwithin{equation}{section}

\title[Breakdown of superconductivity under magnetic steps]{The breakdown of superconductivity in the presence of magnetic steps}
\author[W. Assaad]{Wafaa Assaad}
\address{Lund University, Department of Mathematics, Lund, Sweden}
\email{wafaa.assaad@math.lth.se}
\date{\today}
\begin{document}
\maketitle
\begin{abstract}
	Many earlier works were devoted to the study of the breakdown of superconductivity in  type-II superconducting bounded planar domains, submitted to smooth magnetic fields. In the present contribution, we consider a new situation where the applied magnetic field is piecewise-constant, and the discontinuity jump occurs along a smooth curve meeting the boundary transversely. To handle this situation, we perform a detailed spectral analysis of a  new effective model. Consequently, we establish the monotonicity  of the transition from a superconducting to a normal state. Moreover, we determine the location of superconductivity in the sample just before it disappears completely. Interestingly, the study shows similarities with the case of corner domains subjected to constant fields.
\end{abstract}

\section{Introduction}\label{sec:int}

The breakdown of superconductivity in type-II superconductors submitted to a sufficiently strong magnetic field is a celebrated phenomenon in physics~\cite{saint1963onset,lu2000gauge,helffer2001magnetic,helffer2003upper}. A theorem of Giorgi and Phillips~\cite{giorgi2002breakdown} asserts that a 
 superconducting sample with Ginzburg--Landau parameter $\kappa$,  submitted to  a constant magnetic field of strength $H$, passes permanently to the normal state when $H$ exceeds some critical value. An important question in the literature has been to establish that the transition from  the superconducting to the normal state is \emph{monotone}, i.e. to prove that the sample is superconducting for all $H$ less than the aforementioned critical value. 

Such a  monotonicity has been established in several geometric situations both in 2- and 3-dimensional settings in the case where the Ginzburg--Landau parameter is big,  and for large classes of smooth magnetic fields~\cite{fournais2006third,fournais2007strong,fournais2009ginzburg,raymond2009sharp,fournais2011strong,dombrowski2013semiclassical}. 
In particular, the analysis of 2-dimensional  domains with smooth boundary, submitted to uniform fields shows that the problem is related to a purely linear eigenvalue problem~\cite{fournais2006third,fournais2007strong}. 
The case of corner domains  was treated in~\cite{bonnaillie2005fundamental,bonnaillie2006asymptotics,bonnaillie2007superconductivity}.

However, a monotone transition is not guaranteed in general, and an \emph{oscillatory} behavior occurs in certain geometric settings. One famous example is the Little--Parks effect for 2D annuli~\cite{little1962observation,erdHos1997dia,fournais2015lack},  where the topology of the sample causes the lack of monotonicity. Other examples of this oscillation effect were provided in~\cite{fournais2015lack}, in a case of a disc-shaped sample placed in a non-uniform magnetic field.
 
 In the present paper, we  focus on the case of a smooth domain placed in a \emph{discontinuous} magnetic field. More precisely, we consider a long cylindrical superconducting domain with smooth cross section, submitted to a magnetic field with direction parallel to the axis of the cylinder and whose profile is a step function.  Such a case was not treated in the aforementioned literature. We  aim mainly at answering the following questions:
\begin{itemize}
	\item \emph{Question 1.} How does the discontinuity of the magnetic field affect the monotonicity of the transition from the superconducting to the normal state?
	\item \emph{Question 2.} Where is superconductivity localized right before it completely disappears from the sample?
\end{itemize}
 As shown later in this article, the answers to these questions generate an interesting comparison between the case that we handle and another known case of corner domains submitted to constant magnetic fields (see Section~\ref{sec:main}).
\subsection{The functional and our assumptions}

Consider an open,  bounded, and simply connected set $\Om$ of $\R^2$. Assume that $\Om$ is the horizontal cross section of a long wire subjected to a magnetic field, whose profile is the function $B_0\colon\Om \rightarrow [-1,1]$ and  whose intensity is $H > 0$. The Ginzburg--Landau (GL) free energy is given by the functional
\begin{equation}\label{eq:GL}
\GL (\psi,\Ab)= \int_\Om \Big( \big|(\nb-i\kp H {\mathbf
	A})\psi\big|^2-\kp^2|\psi|^2+\frac{\kappa^2}{2}|\psi|^4 \Big)\,dx
+\kp^2H^2\int_{\Om}\big|\curl\Ab-B_0\big|^2\,dx,
\end{equation}
with $\psi \in H^1(\Om;\C)$ and  $\Ab\in H^1(\Om;\R^2)$. In physics, $\kp>0$ is a characteristic scale of the sample called  the GL parameter,  
$\psi$ is the order parameter with $|\psi|^2$ being a measure of the density of  Cooper pairs, and
 $\Ab$ is the vector potential whose $\curl$ represents the induced magnetic field in the sample.

We carry out our analysis in the asymptotic regime $\kappa\to+\infty$, which corresponds in physics to extreme type-II superconductors. 
We work under the following assumptions on the domain $\Om$ and the magnetic field $B_0$ (see~Figure~\ref{fig1}):
\begin{assumption}\label{assump1}~
	\begin{enumerate}
		\item $\Omega_1$ and $\Omega_2$ are two disjoint open sets.
		\item $\Omega_1$ and $\Omega_2$ have a finite number of connected components.
		\item $\partial\Omega_1$ and $\partial\Omega_2$ are piecewise-smooth with  a finite number of corners.
		\item $\Gamma=\partial\Omega_1\cap\partial\Omega_2$ is the union of a finite number of disjoint simple smooth curves $\{\Gamma_k\}_{k \in \mathcal K}$\,; we will refer to  $\Gamma$ as the \emph{magnetic barrier}.
		\item $\Omega=(\Omega_1\cup\Omega_2\cup\Gamma)^\circ$ and  $\partial\Omega$  is smooth.
		\item For any $k \in \mathcal K$, $\Gamma_k$ intersects $\partial \Om$ at two distinct points. This intersection is transversal, i.e. $\mathrm{T}_{\partial \Omega} \times \mathrm{T}_{\Gamma_k} \neq 0$ at the intersection point, where $\mathrm{T}_{\partial \Omega}$ and  $\mathrm{T}_{\Gamma_k}$ are respectively unit tangent vectors of $\partial \Omega$ and $\Gamma_k$.
		\item $B_0={\mathbbm 1}_{\Omega_1}+a{\mathbbm 1}_{\Omega_2}$, where $a \in [-1,1)\setminus\{0\}$ is a given constant.
	\end{enumerate} 
\end{assumption}
\begin{figure}[H]
	\begin{subfigure}{0.45\linewidth}
		\centering
		\includegraphics[scale=1]{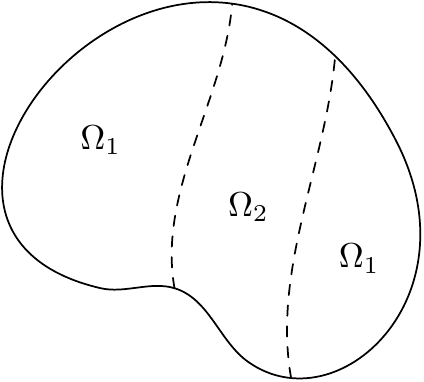}
	\end{subfigure}%
	\begin{subfigure}{0.45\linewidth}
		\centering
		\includegraphics[scale=1]{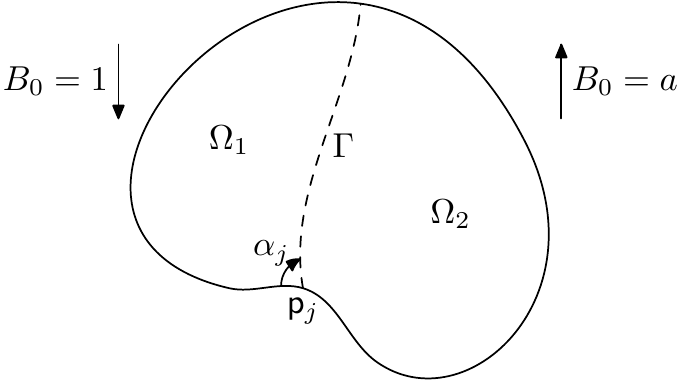}
	\end{subfigure}%
	\caption{Schematic representation of the set $\Om$ subjected to the piecewise-constant magnetic field $B_0$, with the magnetic barrier $\Gamma$.} 
	\label{fig1}
\end{figure}
\begin{notation}\label{not:alfa}
Since $\Gamma \cap \partial \Om$ is finite, we denote by
\[\Gamma \cap \partial \Om=\big\{ \mathsf p_j~:~j \in\{1,...,n\}\big\},\]
where $n =\mathrm{Card}(\Gamma \cap \partial \Om)$. For all $j \in \{1,...,n\}$, let   $\alpha_j \in(0,\pi)$ be the angle between $\Gamma$ and  $\partial \Om$ at the intersection point $\mathsf p_j$ (measured towards $\Om_1$). 
\end{notation}
Since the functional in~\eqref{eq:GL} is gauge invariant\footnote{The physically meaningful quantities $|\psi|^2$, $\curl \Ab$ and $|(\nabla-i\kappa H\Ab)\psi|^2$ are gauge invariant
	in the sense that they do not change under the  transformation $(\psi,\Ab)\mapsto (e^{i\varphi\kappa H}\psi,\Ab+\nb \varphi)$ for any  $\varphi \in H^2(\Om;\R)$.}, one may restrict  its minimization with respect to $(\psi,\Ab)$ (originally done in $H^1(\Om;\C)\times H^1(\Om;\R^2)$) to the space $H^1(\Om;\C)\times\Hd$, where 
\begin{equation}\label{eq:Hd}
\Hd=
\left\{
\Ab \in H^1(\Om;\R^2)~:~ \Div\Ab=0 \ \mathrm{in}\ \Om,\ \Ab\cdot\nu=0\ \mathrm{on}\ \partial \Om
\right\}
\end{equation}
and $\nu$ is a unit normal vector of $\partial \Omega$.
This restriction is beneficial due to the nice regularity properties of the space $\Hd$ (see~\cite[Appendix B]{Assaad}).
Hence, we introduce the following ground-state energy:
\begin{equation} \label{eq:gr_st}
\Es=\inf\big\{\GL(\psi,\Ab)~:~(\psi,\Ab) \in H^1(\Om;\C)\times\Hd\big\}.
\end{equation}

Critical points $(\psi, \Ab) \in H^1(\Om;\C)\times\Hd $ of $\GL$ are weak solutions of the following GL equations:
\begin{equation}\label{eq:Euler}
\begin{cases}
\big(\kn\big)^2\psi=\kp^2(|\psi|^2-1)\psi &\mathrm{in}\ \Om,\\
-\nb^{\perp}  \big(\curl\Ab-B_0\big)= \frac{1}{\kp H}\im\big(\overline{\psi}(\nb-i\kp H \Ab)\psi\big) & \mathrm{in}\ \Om,\\
\nu\cdot(\kn)\psi=0 & \mathrm{on}\ \partial \Om ,\\
\curl\Ab=B_0 & \mathrm{on}~ \partial \Om .
\end{cases}
\end{equation}
Here, $\nb^\perp= (\partial_{x_2},-\partial_{x_1})$.
\subsection{Critical fields}\label{sec:critical}
Let $\Fb \in \Hd$ be the unique vector potential generating the \emph{step} magnetic field $B_0$ (see~\eqref{A_1}). For large $\kappa$, a  result \`a la Giorgi--Phillips (Section~\ref{sec:giorgi}) asserts that for sufficiently strong magnetic fields, $H$, the only solution of~\eqref{eq:Euler} is the normal state $(0,\Fb)$. We want  to prove the existence of a \emph{unique} field where the transition to the normal state happens. To be consistent with the literature, we call this field \emph{the third critical field} and denote it by $H_{C_3}(\kappa)$.

 As mentioned, such a uniqueness result  has been proved  in many generic situations~\cite{fournais2006third,fournais2007strong,fournais2009ginzburg,raymond2009sharp,fournais2011strong,dombrowski2013semiclassical}. 
In their analysis of constant magnetic fields, Fournais and Helffer~\cite{fournais2006third,fournais2010spectral} introduced several natural critical fields, called \emph{global} and \emph{local} fields: a monotone transition requires the global fields to coincide.  To prove the equality of these fields (for large $\kappa$), Fournais and Helffer linked these global fields to  local fields involving spectral data of a linear problem. 

We adapt the definitions of the critical fields in~\cite{fournais2010spectral} to our situation of a step magnetic field. For large $\kappa$, we consider the global fields:
\begin{equation}\label{eq:Hc3_over}
\overline{H}_{C_3}(\kappa)=\inf\big\{H>0~:~\mbox{for all}\ H'>H, (0,\Fb)\ \mbox{is the only minimizer of}\ \mathcal E_{\kappa,H'}\big\},
\end{equation}
\begin{equation}\label{eq:Hc3_under}
\underline{H}_{C_3}(\kappa)=\inf\big\{H>0~:~(0,\Fb)\ \mbox{is the only minimizer of}\ \mathcal E_{\kappa,H}\big\}.
\end{equation}
The latter field was first introduced by Lu and Pan~\cite{lu1999estimates}.
We consider also the local fields:
\begin{equation}\label{eq:Hc3_loc_over}
\overline{H}^\mathrm{loc}_{C_3}(\kappa)=\inf\big\{H>0~:~\mbox{for all}\ H'>H,\ \lambda(\kappa H')\geq\kappa^2\big\},
\end{equation}
\begin{equation}\label{eq:Hc3_loc_under}
\underline{H}^\mathrm{loc}_{C_3}(\kappa)=\inf\big\{H>0~:~\lambda(\kappa H)\geq\kappa^2\big\},
\end{equation}
 where $\lambda(\kappa H)$ stands for the ground-state energy of a Schr\"odinger operator with a step magnetic field, defined in Section~\ref{sec:intro2}. The equality between $\overline{H}^\mathrm{loc}_{C_3}(\kappa)$ and $\underline{H}^\mathrm{loc}_{C_3}(\kappa)$---and consequently between $\overline{H}_{C_3}(\kappa)$ and $\underline{H}_{C_3}(\kappa)$---depends on whether  the function $\mathfrak b\mapsto \lambda(\mathfrak b)$ is monotone increasing for large $\mathfrak b$, a property that has been called 'strong diamagnetism'.  In the settings of this paper, we prove this property in Section~\ref{sec:monot}.
\subsection{Main results}\label{sec:main}
We present now our main results: Theorem~\ref{thm:Hc3} answers Question 1 in the introduction by establishing the existence and the uniqueness of the third critical field, for large $\kappa$, and providing  asymptotics of this field.  Question 2 is answered in Theorem~\ref{thm:agmon}, where we establish certain Agmon-type estimates that make precise the zone of nucleation of superconductivity before disappearing from the sample, and show that the size of this zone is of order $\kappa^{-2}$.

These results involve the following spectral quantities:
\begin{itemize}
	\item $\Theta_0\approx 0.59$ is the so-called de Gennes constant, introduced in Section~\ref{sec:corner} as the ground-state energy of the Neumann realization of the operator $\p_{1,U_\pi}$ in the half-space.
	\item $\mu(\alpha,a)$ is the ground-state energy of the Neumann realization of a Schr\"odinger operator with a step magnetic field in $\R^2_+$, introduced in Section~\ref{sec:new_model}. 
\end{itemize}

The main theorems, namely Theorems~\ref{thm:Hc3} and~\ref{thm:agmon}, are established under the following additional assumption:
\begin{assumption}\label{assump3}~
	Suppose that Assumption~\ref{assump1} holds. For $j\in\{1,...,n\}$, let $\alpha_j$ be the angle in Notation~\ref{not:alfa}. We assume that $\mu(\alpha_j,a)<|a|\Theta_0$.
\end{assumption}
We will discuss the conditions in this  assumption later in the paper (see Section~\ref{sec:heur}).
\begin{rem}
	 In Section~\ref{sec:bound_state}, we provide particular examples of pairs $(\alpha_j,a)$ for which this assumption is satisfied.
\end{rem}
\begin{theorem}\label{thm:Hc3}
	 There exists $\kappa_0>0$ such that if $\kappa\geq \kappa_0$ and $\lambda(\cdot)$ is as in~\eqref{eq:lmda_a}, then the equation
	\[\lambda(\kappa H)=\kappa^2\]
	admits a unique solution $H=H_{C_3}(\kappa)$ which can be estimated as follows:
	\begin{equation}\label{eq:Hc3}
	H_{C_3}(\kappa)	=\frac \kappa {\min\limits_{j \in\{1,...,n\}}\mu(\alpha_j,a)}+ \mathcal O(\kappa^\frac 12),\quad \mathrm{as}\ \kappa \rightarrow +\infty.
	\end{equation}
	Furthermore for $\kappa\geq \kappa_0$, the critical fields defined in~\eqref{eq:Hc3_over} and~\eqref{eq:Hc3_under} coincide and satisfy
	\begin{equation*}
	\overline{H}_{C_3}(\kappa)=\underline{H}_{C_3}(\kappa)=	 H_{C_3}(\kappa).
	\end{equation*}
\end{theorem}
It is worth  comparing the asymptotics of the third critical field in Theorem~\ref{thm:Hc3} with these established in the literature, for regular domains or corner domains submitted to  uniform magnetic fields. In bounded planar domains with smooth boundary, the third critical field has the following asymptotics as $\kappa$ tends to $+\infty$~\cite{lu1999estimates,helffer2001magnetic,helffer2003upper,fournais2006third, fournais2007strong}:
\[H^\mathrm{unif}_{C_3}(\kappa)=\frac \kappa {B\Theta_0}+o(\kappa),\]
when the applied field has  a constant (positive) value $B$. In corner domains, a richer physics is produced for stronger applied magnetic fields, since the corners allow superconductivity to survive longer in the regime $\kappa /(B\Theta_0)\leq H<H^\mathrm{cor}_{C_3}(\kappa)$, where $B$ is the constant field  and $H^\mathrm{cor}_{C_3}(\kappa)$ is the third critical field in the  corner situation. More precisely, the following asymptotics were established in certain geometric settings ~\cite{bonnaillie2005fundamental,bonnaillie2006asymptotics,bonnaillie2007superconductivity}:
\begin{equation}\label{eq:Hcor}
H^\mathrm{cor}_{C_3}(\kappa)=\frac \kappa {B\Lambda}+o(\kappa),\end{equation}
where $\Lambda$ is the ground-state energy of the infinite sector operator with opening angle $\alpha$, introduced in Section~\ref{sec:corner}, and $\alpha$ is the angle corresponding to the corners with the smallest such a ground-state energy. The result has been established under the assumption that $\alpha$ fulfils $\Lambda <\Theta_0$, which is known to be true for the opening angles $\alpha\in(0,\alpha_0)$, $\alpha_0\approx 0.595\pi$ (see Section~\ref{sec:corner}).

Theorem~\ref{thm:Hc3} shows a similarity between the situation in the present paper and that in the corner domains  submitted to uniform  fields.  In the former situation, the magnetic field, having a jump discontinuity along a curve that \emph{cuts the boundary}, has enlarged the scope of the field's strengths where superconductivity still survive in the sample, exactly as the corners do in the latter situation.
 Indeed,  we see that $H_{C_3}(\kappa)$ is of the same order but \emph{strictly larger} than $H^\mathrm{unif}_{C_3}(\kappa)$, where  $H^\mathrm{unif}_{C_3}(\kappa)$ corresponds to the constant field $B=|a|$.

Our next result makes even more clear the similarity between the two aforementioned situations.  
It is known that the corners attract the Cooper pairs (see for instance~\cite{bonnaillie2007superconductivity,helffer2018density}). Indeed, under certain geometric/spectral conditions~\cite[Assumption 1.3]{bonnaillie2007superconductivity}, some asymptotics  of the global energy  established in~\cite{bonnaillie2007superconductivity} suggest the existence of  intermediate phases, between the surface phase  and the normal phase, in which superconductivity can be confined to the  corners satisfying particular spectral conditions---the energetically favourable corners.
Moreover, \cite{bonnaillie2007superconductivity} asserts the nucleation of superconductivity at least at a corner of the domain having the smallest opening angle, before its breakdown.

Recently, the results of~\cite{bonnaillie2007superconductivity} have been sharpened in~\cite{helffer2018density} where some  asymptotics of the \emph{local} energy affirm the confinement of superconductivity to the energetically favourable corners. 

In our case, the Cooper pairs can be attracted by the intersection points of the magnetic barrier $\Gamma$ and the boundary $\partial \Om$. Indeed, working under the spectral conditions in Assumption~\ref{assump3}, Theorem~\ref{thm:agmon} suggests the following: when $\kappa/(|a|\Theta_0)\leq H<H_{C_3}(\kappa)$, superconductivity can successively nucleate near the intersection points of $\Gamma$ and $\partial \Om$,  $\{\mathsf p_j\}_j$, according to the ordering of their spectral parameters $\{\mu(\alpha_j,a)\}_j$. Furthermore, this theorem asserts that superconductivity is eventually localized near at least one of  the points ${\mathsf p_k}$ admitting the smallest parameter $\mu(\alpha_k,a)$, before vanishing in the entire sample. 
\begin{theorem}\label{thm:agmon}
 Take $\mu>0$ satisfying
\[\min_{j\in\{1,...,n\}}\mu(\alpha_j,a)\leq\mu<|a|\Theta_0.\]
We define
\[S=\left\{\mathsf p_j\in\Gamma\cap\partial\Om~:~\mu(\alpha_j,a)\leq\mu\right\}.\]
There exist positive constants $R_0$,  $\kappa_0$, $C$ and $\beta$ such that for all $\kappa\geq\kappa_0$, if
\[ H\geq \frac \kappa\mu,\]
and $(\psi,\Ab) \in H^1(\Om;\C)\times \Hd$ is a solution of~\eqref{eq:Euler}, then
\begin{equation}\label{eq:agmon}
\int_{\Om}e^{\beta \sqrt{\kappa H}\dist(x,S)}\Big(|\psi|^2+\frac 1{\kappa H}|(\nabla-i\kappa H\Ab)\psi|^2\Big)\,dx\leq C\int_{\{\sqrt{\kappa H}\dist(x,S)< R_0\}} |\psi|^2\,dx.\end{equation}
\end{theorem}

	This paper is an integral part of a research that started in~\cite{Assaad,Assaad2019}. Throughout these papers, we  present tools for studying the distribution of superconductivity in a smooth domain submitted to a step magnetic field satisfying  Assumption~\ref{assump1} (the SDSF case), when  $\kappa$ is large, considering various regimes of the intensity of this magnetic field. We particularly aim at detecting any behavior of the sample that is distinct from the well-known behavior of a smooth  domain submitted to a uniform magnetic field (the SDUF case) or  a corner domain submitted to a uniform field (the CDUF case). Such a distinction is not exhibited in the intensity-regime of~\cite{Assaad}. However, in the intensity-regime of~\cite{Assaad2019}, the sample's behavior in the SDSF case is remarkable. It can be dramatically different from the behavior in \emph{both} the SDUF and CDUF cases. The present paper records another interesting magnetic conduct. In the intensity-regime of this paper,  the SDSF case shows  analogy to the CDUF case. This analogy is noteworthy, especially when contrasted to  the discrepancy between these  two cases, observed in~\cite{Assaad2019}. 
	
	In what follows, we summarize our results under three intensity-regime scenarios:

	\begin{itemize}
		\item In the intensity-regime $H<\kappa/|a|$:~\cite{Assaad} establishes the existence of superconductivity in the whole bulk of $\Om$, and the  results of our SDSF case are similar to those of the SDUF and CDUF cases (see e.g.~\cite{sandier2003decrease}).
		\item In the intensity-regime $\kappa/|a|<H\leq\kappa/(|a|\Theta_0)$:~\cite{Assaad} shows the disappearance of superconductivity from the bulk of $\Om_1$ and $\Om_2$. In~\cite{Assaad2019} we affirm the nucleation of superconductivity near $\partial \Om \cup \Gamma$. This nucleation can be \emph{global} (along the entire $\partial \Om \cup\Gamma$) or \emph{partial} (along certain parts of $\partial \Om \cup \Gamma$), according to the values of $H$ and $a \in [-1,1)\setminus\{0\}$ (see~\cite[Section~1.5]{Assaad2019}). This  differs from what occurs in a smooth/corners domain, submitted to the uniform magnetic field\footnote{We choose the value $|a|$ for the uniform magnetic field just to facilitate the comparison between our SDSF case and the SDUF/CDUF case. Choosing a different value for this field will not qualitatively affect the comparison.} $B=|a|$ and considered in the same intensity-regime. Indeed, in the latter case,  if the boundary is smooth then superconductivity is localized \emph{exclusively} and \emph{uniformly} along this boundary~\cite{pan2002surface, almog2007distribution, helffer2011superconductivity, Correggi}. Recently,~\cite{correggi2017surface} proved that this uniform distribution is not affected (to leading order) by the presence of corners.
		\item In the intensity-regime $H>\kappa/(|a|\Theta_0)$: the discussion is done under Assumption~\ref{assump3}. Here, the distribution of superconductivity is  dictated  by the existence of intersection of the discontinuity curve $\Gamma$ and the boundary of the sample. Before its breakdown, superconductivity is shown to be confined to the points of $\partial \Om \cap\Gamma$. As explained in the discussion after Theorems~\ref{thm:Hc3} and~\ref{thm:agmon},  the sample's behavior differs in some aspects  from that in the SDUF case but  shows similarities with that in the CDUF case,  when the uniform field is $B=|a|$. 
		\end{itemize}
		
			Based on the above observations, the combined results of our three papers highlight the peculiarity  of the discontinuous case that we handle: according to the intensity-regime, the SDSF case may resemble to (or deviate from) one or both of the SDUF and CDUF cases. Particularly, the two schematic phase-diagrams in Figure~\ref{fig2} graphically illustrate the comparison between the SDSF case, with the step magnetic field $B_0$, and the CDUF case, with the uniform field $B=|a|$. These diagrams show the distribution of superconductivity in the sample according to the intensity of the applied magnetic field.
			In each case, we plot some critical lines in the $(\kappa,H)$-plane (for large $\kappa$) representing the following:
			\begin{equation*}H_{C_2}(\kappa)=\frac \kappa {|a|}\,,\ H_C^\mathrm {int}(\kappa)=\frac \kappa{|a|\Theta_0}\,, H^\mathrm {step}_{C_3}(\kappa)=H_{C_3}(\kappa)\ \mathrm{in}~\eqref{eq:Hc3}\,,\ \mathrm {and}\ H^\mathrm {cor}_{C_3}(\kappa) \ \mbox{as in}~\eqref{eq:Hcor}.\end{equation*}
			In the SDSF diagram, the configurations of the sample between $H_{C_2}(\kappa)$ and $H_C^\mathrm {int}(\kappa)$ illustrate different  instances of the sample's behavior, occurring according to the values of  $H$ and $a$ (see~\cite[Section 1.5]{Assaad2019}).
\begin{figure}[H]
	\begin{subfigure}{.45\linewidth}
		\centering
		\includegraphics[scale=0.8]{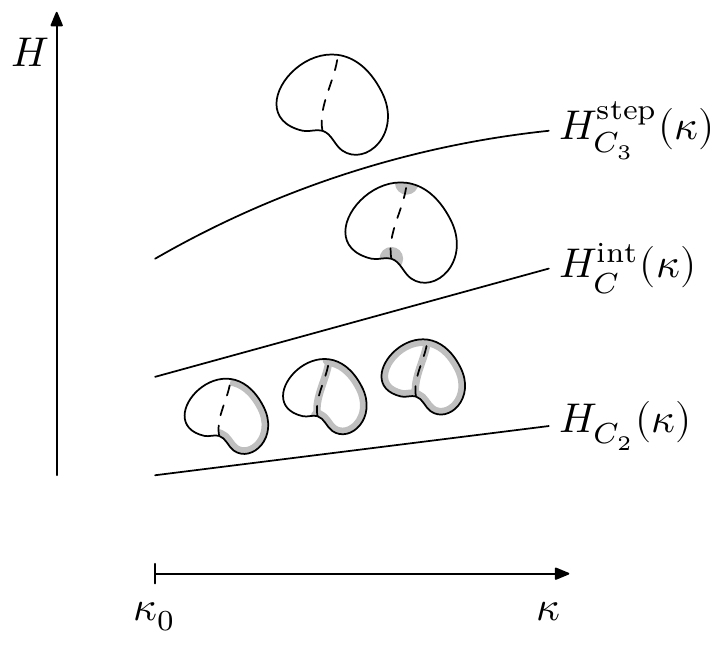}
	\end{subfigure}%
	\begin{subfigure}{.45\linewidth}
		\centering
		\includegraphics[scale=0.8]{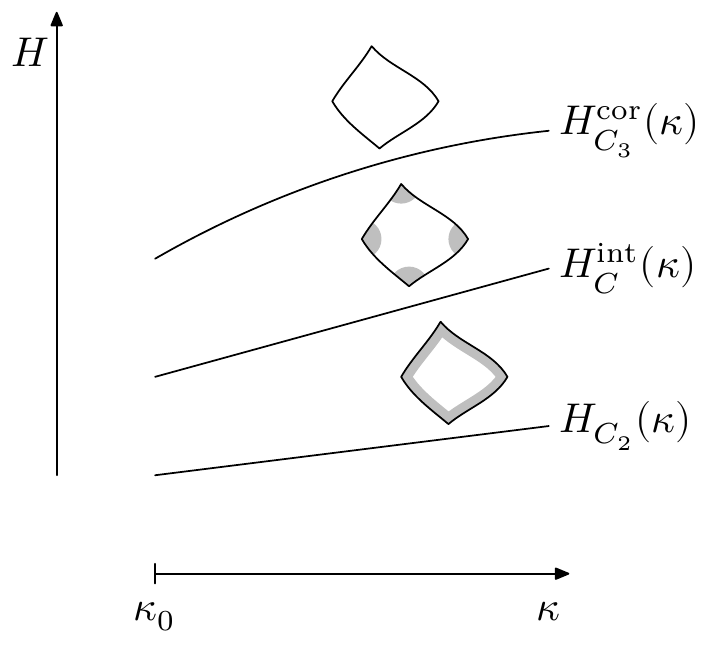}
	\end{subfigure}
	\caption{Schematic phase-diagrams: the SDSF case to the left and the CDUF case to the right. Only the grey regions carry superconductivity.} 
	\label{fig2}
\end{figure}
%
\subsection{Heuristic considerations and outline of the approach}\label{sec:heur}
The discussion in this section is quite informal and is done under the assumptions stated in the introduction (mainly Assumptions~\ref{assump1} and~\ref{assump3}, and that $\kappa$ is large). It aims at presenting the workflow in a simple way. Recall that the two principal results are Theorems~\ref{thm:Hc3} and~\ref{thm:agmon}. 

A  sort of Giorgi--Phillips result  established in Section~\ref{sec:giorgi} asserts that our sample stops superconducting when submitted to large  magnetic fields.  We aim at proving the following: with increasing values of the applied field, there is a one-way phase-transition between superconducting and normal states; once a superconducting sample passes to the normal state it remains in this state. This goal can be achieved by proving that the global critical fields, $\overline{H}_{C_3}(\kappa)$ and $\underline{H}_{C_3}(\kappa)$, defined in Section~\ref{sec:critical}, coincide.

As it is usually the case in the study of  breakdown of superconductivity, the equality of the global fields is not directly established. Instead, the analysis is more manageable when these fields are linked to local ones, $\overline{H}^\mathrm{loc}_{C_3}(\kappa)$ and $\underline{H}^\mathrm{loc}_{C_3}(\kappa)$, also introduced in Section~\ref{sec:critical}. These local fields involve the ground-state energy $\lambda(\mathfrak b)$ of the linear Schr\"odinger operator
\begin{equation*}
\mathcal P_{\mathfrak b,\Fb}=-(\nabla-i\mathfrak b\Fb)^2
\end{equation*}
 defined on $\Om$ with magnetic Neumann boundary conditions  (Section~\ref{sec:intro2}). Here $\mathfrak b$ is a positive parameter, and $\Fb \in \Hd$  is  the vector potential satisfying $\curl\Fb= {\mathbbm 1}_{\Om_1}+a{\mathbbm 1}_{\Om_2}$ ($a\in[-1,1)\setminus\{0\}$). The spirit behind linking the four aforementioned critical fields is that, close to the phase of transition from superconducting to normal state,  the problem can be viewed as linear. Indeed, when $\psi\approx0$  and $\Ab\approx \Fb$, the first equation in~\eqref{eq:Euler} can be approximated by
\[-(\nabla-i\mathfrak b\Fb)^2\psi=\kp^2\psi \quad \mathrm{in}\ \Om.\]
  This approximation of the problem by a linear one is the implicit reason behind establishing  that $\overline{H}_{C_3}(\kappa)=\overline{H}^\mathrm{loc}_{C_3}(\kappa)$ and  $\underline{H}_{C_3}(\kappa)\geq\underline{H}^\mathrm{loc}_{C_3}(\kappa)$ (Section~\ref{sec:main thm}).
Since $\overline{H}_{C_3}(\kappa)\geq \underline{H}_{C_3}(\kappa)$, the equality of the global fields is now equivalent to that of the local ones, which in its turn can be concluded from the fact that  the function $\mathfrak b\mapsto \lambda(\mathfrak b)$ is strictly increasing for large values of $\mathfrak b$. This monotonicity result is proved in Section~\ref{sec:monot} (see Proposition~\ref{prop:monot}), but 
its main ingredients are prepared in Section~\ref{sec:lin_prob}.

 In the aforementioned sections, we generally follow the highways in~\cite{helffer2001magnetic,bonnaillie2003analyse,bonnaillie2005fundamental,fournais2007strong,fournais2010spectral} where similar problems are handled in the case of smooth magnetic fields. However, the particularity of the step magnetic field case that we handle causes deviations at several stages of the analysis. Indeed, our discontinuous situation involves particular models while reducing the problem to other effective ones. Furthermore, a careful analysis and additional techniques are required when working in an environment with a low level of regularity compared to the smooth field case. Some examples showing such a particularity will be presented  while continuing this discussion below.

The asymptotic bounds of the ground-state energy $\lambda(\mathfrak b)$ in Theorem~\ref{thm:lambda_lin} are  key-elements in the monotonicity argument. In the lower bound proof (Section~\ref{sec:low}), a partition of unity allows the local examination of the energy in four main regions of $\Om$: the interior of $\Om$ away from $\partial \Om$, the neighbourhood of $\partial \Om$ away from $\Gamma$, the neighbourhood of $\Gamma$ away from $\partial \Om$, and the vicinity of the intersection points, $\mathsf p_j$, of $\Gamma$ and  $\partial \Om$. 

The study in the first two regions is the same as that in the uniform field case, since the field $\curl \Fb$ is constant in each of the sets $\Om_1$ and $\Om_2$. Hence, the results are borrowed  from the existing literature (e.g.~\cite{fournais2010spectral}). In these two regions, the energy admits  lower bounds of order  $|a|\mathfrak b$ and $|a|\Theta_0\mathfrak b$ respectively.

By suitable change of variables (Sections~\ref{sec:Psi} and~\ref{sec:bc}), gauge transformations and rescaling arguments, we link the study in the two remaining regions to the effective operators with step magnetic fields, $\mathcal L_a$ and $\mathcal H_{\alpha,a}$,  defined on $\R^2$ and $\R^2_+$ respectively. The operator $\mathcal L_a$ is introduced in Section~\ref{sec:La}. It has been studied earlier  in~\cite{hislop2016band,Assaad2019} (and the references therein), and the following bounds of the corresponding ground-state energy, $\beta_a$,  were established:
$|a|\Theta_0\leq \beta_a \leq |a|$. The analysis of the operator $\mathcal H_{\alpha,a}$ in Section~\ref{sec:new_model} is new.  A further comment about this operator is given later in the current section. At the moment we are mainly interested in the upper bound, $\mu(\alpha,a)\leq |a|\Theta_0$, of the ground-state energy 
of $\mathcal H_{\alpha,a}$ (see Remark~\ref{rem:v_0}).
 Consequently, we get the following spectral ordering
\[\mu(\alpha,a)\leq|a|\Theta_0\leq \beta_a \leq |a|,\]
which yields a lower bound  of $\lambda(\mathfrak b)$ with leading order $\min_{j\in\{1,...,n\}}\mu(\alpha_j,a)\mathfrak b$. 

We note that the fulfilment of Assumption~\ref{assump3} is not required while establishing the lower bound result. It is while deriving a matching upper bound of the energy that this assumption is useful (see Section~\ref{sec:up}). Indeed, under Assumption~\ref{assump3} the energies $\{\mu(\alpha_j,a)\}_j$ are eigenvalues (Remark~\ref{rem:v_0}). In particular, the minimal energy $\min_{j}\mu(\alpha_j,a)$ is an eigenvalue. This validates the construction of the trial function involving an eigenfunction corresponding to this minimal energy, in the proof of Proposition~\ref{prop:Upper3}. In the rest of the paper, we work under Assumption~\ref{assump3} each time the argument requires the upper bound of $\lambda(\mathfrak b)$.

In addition to the bounds in Section~\ref{sec:bound}, certain linear Agmon estimates established in 
Theorem~\ref{thm:lin_bound} are used to get the monotonicity result in Proposition~\ref{prop:monot}.
The proof of this proposition is an adaptation of that in~\cite[Theorem~1.1]{fournais2007strong} to our step field situation. It employs the leading order term of  $\lambda(\mathfrak b)$, sparing us the complexity of using  higher order expansions of this energy  as in e.g.~\cite{fournais2006third,bonnaillie2007superconductivity}.  However, the discontinuity of our field as well as the way the magnetic field meets the boundary impose more complicated techniques on the argument (see the discussion below Proposition~\ref{prop:monot}). Moreover, the proof contains a perturbation argument using the independence of the linear operator domain from the parameter $\mathfrak b$ (see~\eqref{eq:domP_F}). Whereas establishing such an independence is standard in the case of smooth fields, our case requires a particular argument given in Appendix~\ref{sec:regularity}.

Consequently, we conclude that the value of the equal global and local fields---the third critical field $H_{C_3}(\kappa)$---is the unique solution of the equation $\lambda(\kappa H)=\kappa^2$ (Proposition~\ref{prop:H_unique}). Asymptotic estimates of this field are given in Proposition~\ref{prop:Hc3}. The aforementioned results (in Sections~\ref{sec:monot} and~\ref{sec:main thm}) constitute the proof of Theorem~\ref{thm:Hc3}.
 
 The  second main result of this work, namely Theorem~\ref{thm:agmon}, is established  in Section~\ref{sec:agmon}. The proof  is given under Assumption~\ref{assump3} which implies the exclusive nucleation of superconductivity near the  points of $\Gamma \cap \partial \Om$ corresponding to the minimal  energy  $\min_{j}\mu(\alpha_j,a)$, right before  its breakdown. Lemma~\ref{lem:spectral} is essential in the proof. It mainly relies on the local energy estimates in Proposition~\ref{prop:Ub2}, together with a simple, yet important, link between the fields $\Ab$ and $\Fb$, done in small patches of the sample (see~\eqref{eq:AF_loc}).
 
 The discussion done so far  shows the main contribution of the operator $\mathcal H_{\alpha,a}$, defined in Section~\ref{sec:new_model}, to our problem. We conclude this outline with a brief spectral  description of this operator. $\mathcal H_{\alpha,a}$ is defined on the half-plane with magnetic Neumann boundary condition, and depends on the two parameters $\alpha \in (0,\pi)$ and $a\in[-1,1)\setminus\{0\}$. 
It is reminiscent of the operator $\mathcal L_a$ defined on the plane (Section~\ref{sec:La}), since each of them involves a  step magnetic field:
 \[\curl \Ab_{\alpha,a}(x)={\mathbbm 1}_{D^1_\alpha}(x)+a{\mathbbm 1}_{D^2_\alpha}(x)\,,\quad\ x\in\R_+^2\quad (\mathrm{for}\ \mathcal H_{\alpha,a}),\]
 \[\curl \sigma\Ab_0(x)={\mathbbm 1}_{\R_+}(x_1)+a{\mathbbm 1}_{\R_-}(x_1)\,,\quad\ x\in\R^2\quad (\mathrm{for}\ \mathcal L_a).\]
However, due to the  dependence of $\mathcal H_{\alpha,a}$ on the angle $\alpha$,  the study of this operator combines  spectral properties of both the operator $\mathcal L_a$ and the sector operator with a constant field defined in Section~\ref{sec:corner}. Actually, our analysis reveals more spectral similarities with the latter operator. Yet, as it will be shown in Section~\ref{sec:new_model}, the discontinuity of the magnetic field in our operator makes the study technically more challenging than that of the sector operator.
 
 By using Persson's lemma in Appendix~\ref{sec:persson}, we show that the bottom of the essential spectrum of $\mathcal H_{\alpha,a}$ is $|a|\Theta_0$. This implies that the ground-state energy satisfies
 \begin{equation}\label{eq:heu}
 \mu(\alpha,a)\leq |a|\Theta_0.\end{equation}
 As mentioned earlier in this section, we are interested in the pairs $(\alpha,a)$ for which the inequality in~\eqref{eq:heu} is strict and consequently the energy $\mu(\alpha,a)$ is an eigenvalue. The existence of such pairs validates Assumption~\ref{assump3} under which this work is done. 
 Let us call here such pairs \emph{admissible pairs}. The pair $(\pi/2,-1)$ is admissible and is directly derived by a symmetry argument  (Proposition~\ref{prop:inf_ev}).
 
Certainly, the continuity of $\mu(\alpha,a)$ with respect to the parameters $\alpha$ and $a$, once verified, would provide more admissible pairs living in a neighbourhood of $(\pi/2,-1)$ (more generally near any already found admissible pair). However, such a regularity result is hard to  establish in our case.  In fact, a continuity result of $\mu(\alpha,a)$ with respect to $a$ is reached after a lengthy proof in Section~\ref{sec:cont}. Still, we did not succeed to prove the continuity with respect to $\alpha$. The way the operator depends on $\alpha$ prevents making profit of the techniques used in earlier works (e.g.~\cite[Section 5.3]{bonnaillie2003analyse}) in similar situations while studying the sector operator (see discussion in Section~\ref{sec:cont}).

The continuity of the energy with respect to $a$ extends the admissibility result at $(\pi/2,-1)$  to other pairs $(\alpha,a)$ for which $\alpha=\pi/2$ (Proposition~\ref{prop:inf_ev}). A more complicated (rigorous) computation is done in Proposition~\ref{prop:bnd_stat_2} seeking  more admissible pairs, in particular pairs with $\alpha\neq \pi/2$. 
  The  proof of this proposition is inspired by the approach in~\cite{exner2018bound}. It starts with some techniques that facilitate the adoption of such an approach. Then it uses a variational argument with convenient test functions to establish a sufficient condition for a pair $(\alpha, a)$ to be admissible. 
   After this proposition, an illustration  using Mathematica is given to show a region of admissible pairs in the vicinity of $(\pi/2,-1)$ (see discussion below the proposition, and Figure~\ref{fig:bnd_stat_2}). 
 
 At this point, it is worth  comparing our results to those in~\cite{exner2018bound}, in order to highlight the challenges created by the  step magnetic field. The argument in~\cite{exner2018bound} shows the existence of bound states for any sector operator with opening angle $\alpha$, such that $\alpha\in(0,\alpha_0)$ and $\alpha_0\approx 0.595\pi$. Similar methods adapted to our situation yield the existence of bound states of the operator $\mathcal H_{\alpha,a}$ for distinct values of $\alpha$. Yet, these values are still near $\pi/2$ (Figure~\ref{fig:bnd_stat_2}).  Also note that the corresponding values of $a$ are negative (near $-1$), and no positive values of $a$ are provided by these methods.
 
 The spectral study of the operator $\mathcal H_{\alpha,a}$, that occupies Section~\ref{sec:new_model} (and Appendix~\ref{sec:persson}), is an essential contribution of the present article. 
 \subsection{Notation}
 \begin{itemize}
 	\item[]	
 	\item The letter $C$ denotes a positive constant whose value may change from one formula to another.
 	\item  Let  $\beta \in (0,1)$.   We use the following
 	H\"{o}lder space:
 	\[\mathcal C^{0,\beta}({\overline \Om})=\Big\{f \in \mathcal C({\overline \Om})\ | \sup_{x\neq y\in \Om}\frac {|f(x)-f(y)|}{|x-y|^\beta}<+\infty\Big\}.\]	
 \end{itemize}
\subsection{Organization of the paper} The rest of the paper is divided into seven sections. In Section~\ref{sec:models}, we summarize some useful  properties of certain known 2D model operators. The operator $\mathcal H_{\alpha,a}$ is analysed in Section~\ref{sec:new_model}. The spectral data of the model operators are used in Section~\ref{sec:lin_prob} while studying the linear eigenvalue problem.  The breakdown of superconductivity under strong magnetic fields is proved in Section~\ref{sec:giorgi}. In Section~\ref{sec:monot}, we establish the eigenvalue monotonicity result when $\kappa$ is large. Consequently, we deduce the equality of the local critical fields and provide certain asymptotics of them as $\kappa$ tends to $+\infty$. The non-linear Agmon estimates in Theorem~\ref{thm:agmon} are  established in Section~\ref{sec:agmon}. Finally, in Section~\ref{sec:main thm} we show the equality of the global and local critical fields for large $\kappa$ and conclude the result in Theorem~\ref{thm:Hc3}. The appendices gather technical estimates that we use here and there.

\section{Some model operators}\label{sec:models}
 We present self-adjoint realizations of some Schr\"odinger operators with magnetic fields in open sets of $\R^2$. A spectral study of these operators can be found in the literature (for instance see~\cite{jadallah2001onset,bonnaillie2003analyse,fournais2010spectral,Assaad2019}).
\subsection{Operators with a constant magnetic field}\label{sec:corner}
Let $U$ be an open and simply connected domain of $\R^2$. Let $\mathfrak b >0$, and $\Ab_0$ be the constant magnetic potential defined by
\begin{equation}\label{canon}
\Ab_0(x) = (0,x_1)\qquad \big(x = (x_1,x_2) \in \R ^2\big).
\end{equation}
If $U=\R^2$, we consider the self-adjoint operator
\begin{equation*}\label{eq:HB-R2}
\p_{\mathfrak b, \R^2}=-(\nabla-i\mathfrak b\Ab_0)^2,
\end{equation*} 
defined on the domain
\begin{equation*}\dom \p_{\mathfrak b, \R^2}=\big\{u\in L^2(\R^2)~:~(\nabla-i\mathfrak b\Ab_0)^j u  \in L^2(\R^2),\, \mathrm{for}\ j\in \{1,2\}\big\}.\end{equation*}
If $U \subsetneq \R^2$, we assume that $\partial U$ is piecewise-smooth with possibly a finite number of corners. In this case we consider the Neumann realization of the self-adjoint operator 
\begin{equation*}\label{eq:HB}
\p_{\mathfrak b, U}=-(\nabla-i\mathfrak b\Ab_0)^2,
\end{equation*}
defined on the domain 
\begin{multline*}\dom \p_{\mathfrak b, U}=\big\{u\in L^2(U)~:~(\nabla-i\mathfrak b\Ab_0)^j u \in L^2(U),\\ \mathrm{for}\ j\in \{1,2\},\,(\nabla-i\mathfrak b\Ab_0)\cdot\nu|_{\partial U}=0\big\},\end{multline*}
where $\nu$ is a unit normal vector of $\partial U$ (when it exists). Let
\[\q_{\mathfrak b,U}(u)=\int_{U}\big|(\nabla-i\mathfrak b\Ab_0)u\big|^2\,dx\]
be the associated quadratic form defined on  
\[\dom \q_{\mathfrak b,U}=\left\{u\in L^2(U)~:~ (\nabla-i\mathfrak b\Ab_0)u \in L^2(U)\right\}.\]
We denote the bottom of the spectrum of $\p_{\mathfrak b, U}$ by $\lambda_U(\mathfrak b)$. 

The case where $U$ is an angular sector in the plane  corresponds to an important \emph{sector} operator. For $0<\alpha\leq \pi$, we define the domain $U_\alpha$ in polar coordinates 
\[ U_\alpha=\big\{r(\cos\theta,\sin\theta)\in \R^2~:~r\in(0,\infty),\ 0<\theta<\alpha\big\}.\]
Using a simple scaling argument, one can prove the following relation between the spectra of the operators $\p_{\mathfrak b, U_\alpha}$ and $\p_{1, U_\alpha}$:
\[\spc \p_{\mathfrak b, U_\alpha}=\mathfrak b\spc \p_{1,U_\alpha}.\]
Therefore, we may restrict to the case $\mathfrak b=1$ and define 
\begin{equation}\label{eq:mu1}
\mu(\alpha)=\lambda_{U_\alpha}(1).
\end{equation}
The special case of $\alpha=\pi$ (the half-plane) has been intensively studied. In this case we denote
\begin{equation}\label{eq:theta01}
\Theta_0:=\mu(\pi).
\end{equation}
Numerical computation shows that $\Theta_0=0.5901....$ We note that $\mu(\pi)$ is not an eigenvalue of $\p_{1,U_\pi}$.

It was conjectured that  $\mu(\alpha)$ is an eigenvalue satisfying $\mu(\alpha)<\Theta_0$, for all $\alpha \in(0,\pi)$
(see e.g.~\cite[Remark~2.4]{bonnaillie2005fundamental}).
 This conjecture has been proved for  $\alpha\in (0,\alpha_0)$ where $\alpha_0\approx 0.595\pi$ ~\cite{jadallah2001onset,bonnaillie2005fundamental,exner2018bound}.  The validity of the conjecture for all $\alpha \in (0,\pi)$ is still not settled, although numerical evidence suggests it (see~\cite{bonnaillie2007computations}). 
When $\mu(\alpha)$ is an eigenvalue,  let $u_\alpha$ be a corresponding normalized eigenfunction. 
%
\subsection{An operator with a step magnetic field in the plane}\label{sec:La}

Let  $a \in [-1,1)\setminus \{0\}$.
For $x \in \R^2$, let  $\sigma$ be a step function defined as follows:
\begin{equation}\label{eq:sigma1}
\sigma(x)=\mathbbm{1}_{\R_+}(x_1)+a\mathbbm{1}_{\R_-}(x_1).
\end{equation}
We introduce the self-adjoint operator 
\begin{align}\label{eq:ham_operator}
\mathcal L_a&=-(\nabla-i\sigma\Ab_0)^2,\quad\mathrm{with}\\
\dom \mathcal L_a&=\big\{u\in L^2(\R^2)~:~(\nabla-i\sigma\Ab_0)^j u \in L^2(\R^2),\, \mathrm{for}\ j\in \{1,2\}\big\},\nonumber
\end{align}
and $\Ab_0$ is the magnetic potential in~\eqref{canon}.  We denote the ground-state energy of $\mathcal L_a$ by
\begin{equation}\label{eq:lamda}
\beta_a= \inf \spc\big(\mathcal L_a \big).
\end{equation}
A spectral analysis of the operator $\mathcal L_a$ has been done in~\cite{hislop2016band} and~\cite{Assaad2019} (see also~\cite{iwatsuka1985examples,hislop2015edge} and references therein), and $\beta_a$  is found to satisfy:
\begin{itemize}
	\item For $0<a<1$,  $\beta_a=a$.
	\item For $a=-1$, $\beta_a=\Theta_0$.
	\item For $-1<a<0$, $|a|\Theta_0<\beta_a<|a|$.
\end{itemize}
\section{A new operator with a step magnetic field in the half-plane}\label{sec:new_model}
In this section we introduce a Schr\"odinger operator with a step magnetic field in $\R_+^2$. To the best of our knowledge, the spectral analysis of this operator is considered for the first time in this contribution. The ground-state energy of this model operator is involved in the leading order of the third critical field $H_{C_3}(\kappa)$, for large values of $\kappa$ (see Theorem~\ref{thm:Hc3}), and it also appears when to determining the zone of concentration of superconductivity in the sample $\Om$, for large $\kappa$ and for sufficiently strong magnetic fields (see Theorem~\ref{thm:agmon}).

Let $a \in[-1,1)\setminus\{0\}$ and  $\alpha \in(0,\pi)$. We define the sets $D_\alpha^1$ and $D_\alpha^2$  in polar coordinates as follows:
\begin{align}
D_\alpha^1&=\{r(\cos\theta,\sin\theta)\in \R^2~:~r\in(0,\infty),\ 0<\theta<\alpha\}, \nonumber\\
D_\alpha^2&=\{r(\cos\theta,\sin\theta)\in \R^2~:~r\in(0,\infty),\ \alpha<\theta<\pi\}. \label{eq:A_alfa}
\end{align}
Consider in $\R^2_+$ the Neumann realization of the operator 
\begin{equation}\label{eq:P_alfa}
\mathcal H_{\alpha,a}=-\left(\nabla-i\Ab_{\alpha,a} \right)^2,
\end{equation}
where  $\Ab_{\alpha,a}=\big(0,A_{\alpha,a}\big)$ is  the magnetic potential\footnote{One may choose a simpler magnetic potential than  $\Ab_{\alpha,a}$, but the choice  in~\eqref{eq:Aa1},~\eqref{eq:Aa2} and~\eqref{eq:Aa3} will prove  useful in  Section~\ref{sec:lin_prob} (see Lemma~\ref{lem:gauge1}).} such that: 
\begin{equation}\label{eq:Aa1}
\mathrm{For}\ \alpha \in(0,\pi/2),\quad A_{\alpha,a}(x_1,x_2)= 
\begin{cases}
x_1+\frac {a-1}{\tan \alpha}x_2,&\mathrm{if}~(x_1,x_2)\in D^1_\alpha,\\
ax_1,&\mathrm{if}~(x_1,x_2)\in D^2_\alpha,
\end{cases}
\end{equation}
\begin{equation}\label{eq:Aa2}
\mathrm{for}\ \alpha \in(\pi/2,\pi),\quad A_{\alpha,a}(x_1,x_2)= 
\begin{cases}
x_1,&\mathrm{if}~(x_1,x_2)\in D^1_\alpha,\\
ax_1+\frac {1-a}{\tan \alpha}x_2,&\mathrm{if}~(x_1,x_2)\in D^2_\alpha,
\end{cases}
\end{equation}
\begin{equation}\label{eq:Aa3}
\mathrm{and}\quad A_{\frac \pi 2,a}(x_1,x_2)= 
\begin{cases}
x_1,&\mathrm{if}~(x_1,x_2)\in D^1_{\pi/2},\\
ax_1,&\mathrm{if}~(x_1,x_2)\in D^2_{\pi/2}.
\end{cases}
\end{equation}
 The potential $\Ab_{\alpha,a}$ is in $H^1(\R^2_+,\R^2)$ and satisfies $\curl\Ab_{\alpha,a}={\mathbbm 1}_{D_\alpha^1}+a{\mathbbm 1}_{D_\alpha^2}$.
The operator $\mathcal H_{\alpha,a}$ is defined on the domain
 \begin{multline}\label{eq:domH}
\dom \mathcal H_{\alpha,a}=\big\{u\in L^2(\R^2_+)~:~ (\nabla-i\Ab_{\alpha,a})^j u \in L^2(\R^2_+),\\ \mathrm{for}\ j\in \{1,2\}\,,(\nabla-i\Ab_{\alpha,a})\cdot (0,1)|_{\partial (\R^2_+)}=0\big\}.
\end{multline}
The associated quadratic form, $q_{\alpha,a}$, is defined as 

\begin{align}\label{eq:qa}
q_{\alpha,a}(u)&=\int_{\R^2_+}\big|(\nabla-i\Ab_{\alpha,a})u\big|^2\,dx, \quad \mathrm{with} \\
\dom q_{\alpha,a}&=\left\{u\in L^2(\R^2_+)~:~ (\nabla-i\Ab_{\alpha,a})u \in L^2(\R^2_+)\right\}. \nonumber
\end{align}
Let 
\begin{equation}\label{eq:mu_alfa}
\mu(\alpha,a)=\inf_{\substack{u\in\dom q_{\alpha,a}\\ u \neq 0 }}\frac{q_{\alpha,a}(u)}{\|u\|^2_{L^2(\R^2_+)}},
\end{equation}
be the bottom of the spectrum of $\mathcal H_{\alpha,a}$. 
 \subsection{Bottom of the essential spectrum}
\begin{theorem}\label{thm:ess_theta}
	Let $\alpha \in (0,\pi)$ and $a \in [-1,1)\setminus \{0\}$. Then $\inf \spc_\mathrm{ess}(\mathcal H_{\alpha,a})=|a|\Theta_0$.
\end{theorem}
We refer to Appendix~\ref{sec:persson} for the proof of Theorem~\ref{thm:ess_theta}. Our proof is an adaptation of the corresponding proof for sector operators~\cite[Section 3]{bonnaillie2003analyse}, which in turn is  a  generalization of Persson's lemma for unbounded domains in $\R^2$ and Neumann realizations, and is based on ideas in~\cite{persson1960bounds,helffer99,agmon2014lectures}. 

\begin{rem} \label{rem:v_0}
 From Theorem~\ref{thm:ess_theta}, it follows that $\mu(\alpha,a)\leq|a|\Theta_0$ for all $\alpha \in (0,\pi)$ and $a \in [-1,1)\setminus \{0\}$, and if  $\mu(\alpha,a)<|a|\Theta_0$, then $\mu(\alpha,a)$ is an eigenvalue of $\mathcal H_{\alpha,a}$.
\end{rem}
\subsection{A continuity result}\label{sec:cont}
The operator $\mathcal H_{\alpha,a}$ depends on the parameters $\alpha$ and~$a$. Some change of variable techniques have been previously used for other parameter dependent operators (see e.g.~\cite[Section 5.3]{bonnaillie2003analyse}) to  link the problem to an operator with a fixed domain, independent of the parameters. This allows the use of the perturbation theory \cite[Appendix C]{fournais2010spectral} to prove certain regularity properties of the ground-state energy. Unfortunately, such techniques may not be useful in our case. This causes  difficulties in establishing some smoothness results with respect to $\alpha$. The aim of this section is to prove the continuity of $\mu(\alpha,a)$ with respect to $a$.
\begin{proposition}\label{prop:mu-continuity}
	Let $\alpha \in (0,\pi)$. The function $a\mapsto\mu(\alpha,a)$ is continuous for  $a \in  [-1,1)\setminus \{0\}$.
\end{proposition}
The proof of Proposition~\ref{prop:mu-continuity}  mainly relies on establishing that  $\mu(\alpha,a)$ is the limit of another ground-state energy, $\mu(\alpha,a,r)$, of an operator with associated form domain that is independent of $\alpha$ and $a$. Then, the continuity of $a\mapsto\mu(\alpha,a)$ is deduced from that of $a \mapsto \mu(\alpha,a,r)$. This will be made more precise in what follows. Let $B_r=B(0,r)$ be the ball of radius $r>0$, and $B_r^+=B_r\cap \R^2_+$. Define 
\begin{equation}\label{eq:Dr} \mathcal D_r=\big\{u \in H^1(B_r^+)~:~u=0\ \mathrm{on}\ \partial B_r \cap \R^2_+\big\}.\end{equation}
Let \begin{equation}\label{eq:mu_r} 
\mu(\alpha,a,r)=\inf_{\substack{u \in \mathcal D_r\\ u \neq 0}}\frac {\|(\nabla-i\Ab_{\alpha,a})u\|_{L^2(B_r^+)}}{\|u\|_{L^2(B_r^+)}}.\end{equation}
\begin{lemma}\label{lem:cont2}
	The function $a \mapsto \mu(\alpha,a,r)$ is continuous.
\end{lemma}
The proof of the lemma above is standard, but  presented in Appendix~\ref{sec:persson} for the convenience of the reader.
\begin{rem}
	Note that the  form domain $\mathcal D_r$ is independent of the parameter $a$, and that for a fixed function $u\in \mathcal D_r$, $a\mapsto \|(\nabla-i \Ab_{\alpha,a})u\|^2_{L^2(B_r^+)}$ is holomorphic. Consequently, one can apply  the perturbation theory to prove more regularity of $a\mapsto \mu(\alpha,a,r)$ (see \cite[Appendix~C]{fournais2010spectral}). However,  we will be satisfied by the continuity result of Lemma~\ref{lem:cont2} to establish Proposition~\ref{prop:mu-continuity}.
\end{rem}
\begin{rem}	
	Unfortunately, the perturbation theory might not be helpful in proving the smoothness of $\alpha \mapsto \mu(\alpha,a,r)$, despite of the independence of the domain $\mathcal D_r$ from $\alpha$. Moreover,  the continuity of $\alpha \mapsto \mu(\alpha,a,r)$ is not obvious; a technical difficulty comes from the possibility that
	\[\liminf_{h\rightarrow 0}\sup_{\substack{u\in  \mathcal D_r\\ \|u\|_{L^2(B_r^+)}=1}}\int_{\alpha}^{\alpha+h}|u|^2\,dx>0\]
	which prevents us from comparing the eigenvalues $\mu(\alpha,a,r)$ and $\mu(\alpha+h,a,r)$ using the min-max principle.
\end{rem}
The next lemma gives an energy lower bound for the functions in the domain of $q_{\alpha,a}$, supported away from the origin. The proof is also provided in Appendix~\ref{sec:persson}. 

Consider the set
\begin{equation}\label{eq:M_r}
\mathcal M_r=\big\{u \in \dom q_{\alpha,a}~:~u=0\ \mathrm{in}\ B_r^+\big\}.\end{equation}
	\begin{lemma}\label{lem:out_B_r}
		Let $\alpha \in (0,\pi)$ and $a \in [-1,1)\setminus \{0\}$. There exists a constant $C>0$, independent of $a$ and dependent on $\alpha$, such that for all $r>0$ and any non-zero function $u \in\mathcal M_r$, it holds
		\[q_{\alpha,a}(u) \geq \Big(|a|\Theta_0-\frac C {r^2}\Big)\|u\|^2_{L^2(\R^2_+)}.\]
	\end{lemma}
\begin{proof}[Proof of Proposition~\ref{prop:mu-continuity}]	
First, we prove that 
\begin{equation}\label{eq:conv_mu}
\lim_{r \rightarrow +\infty} \mu(\alpha,a,r)=\mu(\alpha,a).
\end{equation}
 By a standard application of the min-max principle,  we see that $r\mapsto \mu(\alpha,a,r)$ is decreasing in $\R_+$. Indeed, for $r>0$, we may extend any $u\in \mathcal D_r$ by zero outside $B_r$ (the extension is still denoted by $u$ for simplicity), hence for $\rho>r$ and $u \in \mathcal D_r$, we view $u \in \mathcal D_{\rho}$. Consequently, $\lim_{r \rightarrow +\infty} \mu(\alpha,a,r)$ exists.

	Since $ D_r\subset\dom q_{\alpha,a}$, the estimate  $\mu(\alpha,a)\leq\lim_{r \rightarrow +\infty} \mu(\alpha,a,r)$ is straightforward. It remains to establish $\mu(\alpha,a)\geq\lim_{r \rightarrow +\infty} \mu(\alpha,a,r)$.
	Let $u \in \dom q_{\alpha,a}$. Consider a smooth cut-off function $f_r$ supported in $B_r$, such that
	\begin{equation}\label{eq:cut}
	0\leq f_r\leq1,\quad f_r=1\ \mathrm{in}\ B_{\frac r2},\quad \mathrm{and}\  |\nabla f_r|\leq \frac C{r},
	\end{equation}
	for some universal constant $C>0$.	We have
	\begin{multline}\label{eq:cut1}
	\|(\nabla-i\Ab_{\alpha,a})f_r u\|_{L^2(\R^2_+)}^2=\|f_r(\nabla-i\Ab_{\alpha,a})u\|_{L^2(\R^2_+)}^2 +\|u|\nabla f_r| \|_{L^2(\R^2_+)}^2\\+2 \re\big\langle u\nabla f_r,f_r (\nabla-i\Ab_{\alpha,a}) u\big\rangle.\end{multline}
	Then by~\eqref{eq:cut} and~\eqref{eq:cut1}, we bound $\|f_r(\nabla-i\Ab_{\alpha,a})u\|_{L^2(\R^2_+)}^2$ from below by \[\|(\nabla-i\Ab_{\alpha,a})f_r u\|_{L^2(\R^2_+)}^2-\frac C{r^2}\|u\|^2_{L^2(\R^2_+)}
	-\frac Cr\|u\|_{L^2(\R^2_+)}\|f_r(\nabla-i\Ab_{\alpha,a})u\|_{L^2(\R^2_+)}\]
	which, in turn, by the min-max principle can be bounded below by 
	\[\mu(\alpha,a,r)\|f_r u\|^2_{L^2(\R^2_+)}- \frac C{r^2}\|u\|^2_{L^2(\R^2_+)}
	- \frac Cr\|u\|_{L^2(\R^2_+)}\|f_r(\nabla-i\Ab_{\alpha,a})u\|_{L^2(\R^2_+)}.\]

	Hence, having $ q_{\alpha,a}(u)\geq \|f_r(\nabla-i\Ab_{\alpha,a})u\|_{L^2(\R^2_+)}^2$, we get
	\[\frac {q_{\alpha,a}(u)}{\|u\|^2_{L^2(\R^2_+)}}\geq \mu(\alpha,a,r)\frac  {\|f_r u\|^2_{L^2(\R^2_+)}}{\|u\|^2_{L^2(\R^2_+)}}-\frac C{r^2}-\frac Cr\frac {\|f_r(\nabla-i\Ab_{\alpha,a})u\|_{L^2(\R^2_+)}}{\|u\|_{L^2(\R^2_+)}}.\]
	Taking $r$ to $+\infty$ and using the dominated convergence theorem, we obtain
	\[\frac {q_{\alpha,a}(u)}{\|u\|^2_{L^2(\R^2_+)}}\geq \lim_{r \rightarrow +\infty}\mu(\alpha,a,r).\]
	Since $u \in \dom q_{\alpha,a}$ is arbitrary, we conclude that $\mu(\alpha,a)\geq\lim_{r \rightarrow +\infty} \mu(\alpha,a,r)$. 

Next, we establish a useful lower bound of $\mu(\alpha,a)$. Let $r>0$ and consider a partition of unity, $(\varphi_r,\chi_r)$, of $\R^2$ satisfying
	\begin{equation*}
	\supp \varphi_r \subset B_r\,,\quad \supp \chi_r \subset (B_{\frac r2})^\complement\,,\quad |\nabla \varphi_r|^2+|\nabla \chi_r|^2\leq \frac C{r^2},
	\end{equation*}
	for some universal constant $C>0$. Let $u \in \dom q_{\alpha,a}$ such that $\|u\|_{L^2(\R^2_+)}=1$. The IMS formula (\cite[Theorem 3.2]{cycon2009schrodinger}) ensures that 
	\begin{equation}\label{eq:cont1}
	q_{\alpha,a}(u)\geq q_{\alpha,a}(\varphi_r u)+q_{\alpha,a}(\chi_r u)-\frac C{r^2}.
	\end{equation}
	Note that $\varphi_r u \in \mathcal D_r$ and $\chi_r u\in \mathcal M_{\frac r2}$, where $\mathcal M_{\frac r2}$ is defined in~\eqref{eq:M_r}. Thus
	\begin{equation}
	\label{eq:cont2}
	q_{\alpha,a}(\varphi_r u) +q_{\alpha,a}(\chi_r u)\geq \mu(\alpha,a,r) \|\varphi_r u\|_{L^2(\R^2_+)}+|a|\Theta_0 \|\chi_r u\|_{L^2(\R^2_+)}-\frac C{r^2},
	\end{equation} 
	for some $C$ that is independent of $a$ and $r$. In the above inequality, we used~\eqref{eq:mu_r} and Lemma~\ref{lem:out_B_r}. Combining~\eqref{eq:cont1} and \eqref{eq:cont2} gives
	\begin{equation}\label{eq:cont4}
	\mu(\alpha,a)\geq \min\big(\mu(\alpha,a,r),|a|\Theta_0\big)-\frac C{r^2},\end{equation}
	for $C$ independent of $a$ and $r$.
	
Finally, we establish the continuity of $a\mapsto \mu(\alpha,a)$. Let $r>0$ and $h \in \R$ such that $a+h \in  [-1,1)\setminus \{0\}$. 
By~\eqref{eq:conv_mu} and the monotonicity of $r\mapsto \mu(\alpha,a,r)$ (see Step 1.), we have
	\[\mu(\alpha,a+h)\leq \mu(\alpha,a+h,r).\]
	Hence, using the continuity of $a \mapsto \mu(\alpha,a,r)$ (Lemma~\ref{lem:cont2}) gives
	\[\limsup_{h \rightarrow 0}\mu(\alpha,a+h) \leq \mu(\alpha,a,r).\]
	Let $r$ tend to $+\infty$ and use~\eqref{eq:conv_mu} to get
	$\limsup_{h \rightarrow 0}\mu(\alpha,a+h) \leq \mu(\alpha,a)$.
	Next, by~\eqref{eq:cont4} we have
	\begin{equation}\label{eq:cont10}
	\mu(\alpha,a+h)\geq \min\big(\mu(\alpha,a+h,r),|a+h|\Theta_0\big)-\frac C{r^2}.
	\end{equation}
	But Theorem~\ref{thm:ess_theta} asserts that
	\begin{equation}
	|a+h|\Theta_0 \geq |a|\Theta_0-|h| \Theta_0
	\geq \mu(\alpha,a)-|h|\Theta_0. \label{eq:cont11}
	\end{equation}
	We plug~\eqref{eq:cont11} in~\eqref{eq:cont10} and we insert $\liminf_{h \rightarrow 0}$ to obtain
\begin{equation*}\liminf_{h \rightarrow 0}\mu(\alpha,a+h)\geq \min\big(\mu(\alpha,a,r),\mu(\alpha,a)\big)-\frac C{r^2}
\geq \mu(\alpha,a)-\frac C{r^2}.
\end{equation*}
	In the above inequality, Lemma~\ref{lem:cont2} and the monotonicity of $r\mapsto \mu(\alpha,a,r)$ are used again. Take $r$ to $+\infty$ and use~\eqref{eq:conv_mu} to conclude that 
	$\liminf_{h \rightarrow 0}\mu(\alpha,a+h) \geq \mu(\alpha,a)$.\qedhere
\end{proof}
 %
\subsection{Bound states}\label{sec:bound_state}
In what follows, we provide particular values of $\alpha$ and $a$ where $\mu(\alpha,a)$ is an eigenvalue (see Propositions~\ref{prop:inf_ev} and~\ref{prop:bnd_stat_2}), then we conclude with establishing some decay result of the corresponding eigenfunction(s) (see Theorem~\ref{thm:v0-decay}).
\begin{proposition}\label{prop:inf_ev}
	There exists $\gamma_0 \in (0,1)$ such that, for all  $a \in [-1,-1+\gamma_0)$, 
	the bottom of the spectrum of $\mathcal H_{\pi/2,a}$,  $\mu\big(\pi/2, a\big)$, is an eigenvalue. 
\end{proposition}
\begin{proof}
	Let $u:=u_{\pi/2}$ be a normalized eigenfunction associated with the eigenvalue $\mu(\pi/2)$ introduced in Section~\ref{sec:corner}.
	Consider  a function $\hat u$ in $\R \times \R_+$ satisfying
	\begin{equation*}
	\hat u(x_1,x_2)=
	\left\{
	\begin{array}{ll}
	u(x_1,x_2)& \mathrm{if}\ x_1> 0,\\
	u(-x_1,x_2)& \mathrm{if}\ x_1<0.
	\end{array}
	\right.
	\end{equation*}
	 For $a=-1$,  a simple computation yields that $\hat u \in \dom q_{\pi/2,-1}$ and satisfies
	\begin{equation*} 
	\frac {q_{\frac \pi 2,-1}(\hat u)}{\|\hat u\|^2_{L^2(\R^2_+)}} =\mu\Big(\frac \pi 2\Big)
	<\Theta_0
	\end{equation*}
	(see Section~\ref{sec:corner}). Hence, the min-max principle ensures that
	$\mu\big(\pi/2, -1\big) <\Theta_0$,
	which establishes that this ground-state energy is an eigenvalue (see Remark~\ref{rem:v_0}).
	The rest of the proof follows from the continuity of $a\mapsto \mu(\pi/2,a)$ at $a=-1$ (see Proposition~\ref{prop:mu-continuity}).
\end{proof} 
Inspired by the construction in~\cite[Proof of Theorem 1.1]{exner2018bound} in the study of corner domains, we establish a sufficient condition on the angle $\alpha$ and the number $a$ under which $\mu(\alpha,a)$ is an eigenvalue.
\begin{proposition}\label{prop:bnd_stat_2}
	For $\alpha \in (0,\pi)$ and  $a \in [-1,1)\setminus \{0\}$, consider the function $P_{\alpha,a}:(0,+\infty) \rightarrow \R$ defined by 
	\begin{equation*}
	P_{\alpha,a}(x)=Ax^2-\frac \pi 2|a|\Theta_0x+\frac\pi 2,\end{equation*}
	with
	\begin{multline*}A=-\frac 1{64} \csch(\pi)\Big(-4a\pi+(3-2a+3a^2)\pi\cosh(\pi)+(-1+a)\pi\big((-1+a)\cosh(\pi-2\alpha)\\+4\cosh(\pi-\alpha)-4a\cosh(\alpha)\big)-8\big(\alpha+a^2(\pi-\alpha)\big)\sinh(\pi)\Big).\end{multline*}
	If there exists $x=x(\alpha,a)>0$ such that $P_{\alpha,a}(x)<0$, then $\mu(\alpha,a)$ is an eigenvalue of the operator $\mathcal H_{\alpha,a}$.
\end{proposition}
\begin{proof}
	Fix $a \in [-1,1)\setminus\{0\}$ and $\alpha \in(0,\pi)$.  Recall the notation in the introduction of Section~\ref{sec:new_model}. There exists  a function $\varphi \in H^1_\mathrm{loc}(\R_+^2)$ such that the vector potential  $\Ab_{\alpha,a}$ satisfies on $\R_+^2$ 
	\begin{equation*}
\Ab_{\alpha,a}=\breve \sigma\Ab_1+\nabla\varphi,
	\end{equation*} 
	where $\Ab_1(x)=1/2(-x_2,x_1)$ (with $x=(x_1,x_2)$) and
	\begin{equation*}
	\breve\sigma(x)=
	\left\{
	\begin{array}{ll}
	1& \mathrm{if}\ x\in D_\alpha^1,\\
	a& \mathrm{if}\ x\in D_\alpha^2
	\end{array}
	\right.
	\end{equation*}
~\cite[Lemma 1.1]{leinfelder1983gauge}. An explicit definition of this function is the following:  
\begin{equation*}
\mathrm{For}\ \alpha\in(0,\pi/2],\ \varphi(x)=
\left\{
\begin{array}{ll}
\frac 12x_1x_2+\frac {a-1}{2} \cot \alpha\, x_2^2 & \mathrm{if}\ x\in D_\alpha^1,\\
\frac a2 x_1x_2& \mathrm{if}\ x\in D_\alpha^2,
\end{array}
\right.
\end{equation*}
\begin{equation*}
\mathrm{For}\ \alpha\in(\pi/2,\pi),\ \varphi(x)=
\left\{
\begin{array}{ll}
\frac 12 x_1x_2 & \mathrm{if}\ x\in D_\alpha^1,\\
\frac a2x_1x_2+\frac {1-a}{2} \cot \alpha\, x_2^2& \mathrm{if}\ x\in D_\alpha^2.
\end{array}
\right.
\end{equation*}
Hence, considering the quadratic form
\begin{equation*}\label{eq:q1}
\breve q(v)=\int_{\R^2_+}\big|(\nabla-i\breve \sigma\Ab_1)v\big|^2\,dx,
\end{equation*}
with domain
\[\dom \breve q=\left\{v \in L^2(\R_+^2)~:~(\nabla-i\breve \sigma\Ab_1)v \in L^2(\R_+^2) \right\},\]
we get for all $v\in \dom \breve q$
\begin{equation}\label{eq:q0}
\breve q(v)=q_{\alpha,a}(e^{i\varphi}v).
\end{equation}
The quadratic form $\breve q$ is expressed in polar coordinates $(\rho,\theta) \in \breve D^{pol}:=(0,+\infty)\times(0,\pi)$ as follows
	\begin{equation*}
	\breve q^{pol}(v)=\int_0^\pi\int_0^{+\infty}\Big(|\partial_\rho v|^2+\frac 1 {\rho^2}\Big|\Big(\partial_{\theta}-i\breve\sigma^{pol}\frac {\rho^2}2 \Big)v\Big|^2\Big)\rho\,d\rho\,d{\theta}
	,
	\end{equation*}
	where $\breve{\sigma}^{pol}(\rho,\theta)=\breve{\sigma}(x_1,x_2)$ and
	\[\dom \breve q^{pol}=\left\{v \in L^2_\rho(\breve D^{pol})~:~\partial_\rho v \in L^2_\rho(\breve D^{pol})\,,\, \frac 1 \rho\Big(\partial_{\theta}-i\breve\sigma^{pol}\frac {\rho^2}2\Big)v \in L^2_\rho(\breve D^{pol})  \right\}.\]
	For any set $D \subset \R^2$, $L^2_\rho(D)$ denotes the weighted space of weight $\rho$. Just to easily follow the computation steps in~\cite{exner2018bound}, we consider further the quadratic form $\tilde{q}^{pol}$, defined on $\tilde{D}^{pol}:=(0,+\infty)\times(-\pi+\alpha,\alpha)$ by
	\begin{equation*}
	\tilde{q}^{pol}(u)=\int_{-\pi+\alpha}^\alpha\int_0^{+\infty}\Big(|\partial_\rho u|^2+\frac 1 {\rho^2}\Big|\Big(\partial_{\theta}+i\tilde{\sigma}^{pol}\frac {\rho^2}2 \Big)u\Big|^2\Big)\rho\,d\rho\,d\theta,
	\end{equation*}
	where
	\[\dom \tilde{q}^{pol}=\left\{u \in L^2_\rho(\tilde{D}^{pol})~:~\partial_\rho u \in L^2_\rho(\tilde{D}^{pol})\,,\, \frac 1 \rho\Big(\partial_{\theta}+i\tilde{\sigma}^{pol}\frac {\rho^2}2\Big)u \in L^2_\rho(\tilde{D}^{pol})  \right\},\]
	and \begin{equation}
	\tilde{\sigma}^{pol}(\rho,\theta)=
	\left\{
	\begin{array}{ll}
	a& \mathrm{if}\ (\rho,\theta)\in (0,+\infty)\times(-\pi+\alpha,0),\\
	1& \mathrm{if}\ (\rho,\theta)\in (0,+\infty)\times(0,\alpha).
	\end{array}
	\right.
	\end{equation} 
	Performing a suitable symmetry and  rotation of domain,
	we get for all $u\in \dom
	\tilde{q}^{pol}$ 
	\begin{equation}\label{eq:q2}
	\tilde{q}^{pol}(u)=\breve{q}^{pol}(v),
	\end{equation}
	where $v(\rho,\theta)=u(\rho,-\theta+\alpha)$. 
	
	In light of the above discussion (more precisely using~\eqref{eq:q0} and~\eqref{eq:q2}), a sufficient condition for $\mu(\alpha,a)$  to be an eigenvalue is to find a test function $u_* \in \dom \tilde{q}^{pol}$  satisfying
	\begin{equation}\label{eq:bnd_stat_22}
	\tilde{q}^{pol}(u_*)<|a|\Theta_0\|u_*\|^2_{L^2_\rho(\tilde{D}^{pol})}.\end{equation}
	This  follows from Remark~\ref{rem:v_0} and the min-max principle. To this end, we consider the  function 
	\[u_*(\rho,\theta)=e^{-\beta \frac {\rho^2}2}e^{-i\rho g(\theta)},\]
	where  $g\colon(-\pi+\alpha,\alpha) \rightarrow \R$ is a piecewise-differentiable function,  $\beta>0$, $g$ and $\beta$ to be suitably chosen later. We define the functional $\mathcal I$ on $\dom \tilde{q}^{pol}$ by
	\[u\mapsto \mathcal I[u]=\tilde{q}^{pol}(u)-|a|\Theta_0\|u\|^2_{L^2_\rho(\tilde{D}^{pol})}.\]
	Then establishing~\eqref{eq:bnd_stat_22}  is equivalent to showing that 
	\begin{equation}\label{eq:bnd_stat_21}
	\mathcal I[u_*]<0.
	\end{equation} 
    An elementary computation yields
	\begin{align*}
	\mathcal I[u_*]&=\int_0^{+\infty} \rho e^{-\beta \rho^2}\,d\rho \int_{-\pi+\alpha}^0 \Big(g^2+(\partial_\theta g)^2-|a|\Theta_0\Big)\,d\theta\\&\quad- \int_0^{+\infty} \rho^2 e^{-\beta \rho^2}\,d\rho \int_{-\pi+\alpha}^0 a\partial_\theta g\,d\theta\\
	&\quad+\int_0^{+\infty} \rho e^{-\beta \rho^2}\,d\rho \int_0^\alpha \Big( g^2+(\partial_\theta g)^2-|a|\Theta_0\Big)\,d\theta- \int_0^{+\infty} \rho^2 e^{-\beta \rho^2}\,d\rho \int_0^\alpha \partial_\theta g\,d\theta\\
	&\quad+\Big(\pi\beta^2+\frac 1 4\big(\alpha+a^2(\pi-\alpha)\big)\Big)\int_0^{+\infty}\rho^3 e^{-\beta \rho^2}\,d\rho.
	\end{align*}
	Let $\mathcal E_n=\int_0^{+\infty} \rho^n e^{-\beta \rho^2}\,d\rho$, for $n\geq0$. We use the equalities 
	$\mathcal E_1=1/(2\beta)$, $\mathcal E_2=\sqrt \pi/(4 \beta^{3/2})$, and $\mathcal E_3=1/(2\beta^2)$
	\cite[Equations 3.461]{gradshteyn2015table} to conclude that 
	\begin{align}
	\mathcal I[u_*]&=\frac 1{2\beta} \int_{-\pi+\alpha}^0 \Big(g^2+(\partial_\theta g)^2\Big)\,d\theta-a\frac {\sqrt \pi}{4 \beta^{\frac 32}} g(\theta)\Big|_{-\pi+\alpha}^0 \nonumber\\
	&\quad+\frac 1{2\beta} \int_0^\alpha \Big(g^2+(\partial_\theta g)^2\Big)\,d\theta-\frac {\sqrt \pi}{4 \beta^{\frac 32}} g(\theta)\Big|_0^\alpha +\frac \pi 2-\frac {|a|\Theta_0\pi}{2\beta}+\frac 1{8 \beta^2} \big(\alpha+a^2(\pi-\alpha)\big). \label{eq:bnd_stat_23}
	\end{align}
	We choose further
	\begin{equation*}
	g(\theta)=
	\left\{
	\begin{array}{ll}
	c_1 e^\theta+c_2 e^{-\theta}& \mathrm{if}\ -\pi+\alpha<\theta\leq 0,\\
	c_3 e^\theta+c_4 e^{-\theta}& \mathrm{if}\ 0<\theta<\alpha,
	\end{array}
	\right.
	\end{equation*} 
	where $c_1,c_2,c_3,c_4$ are real coefficients satisfying $c_1+c_2=c_3+c_4$.
	This condition on the coefficients is imposed to guarantee the continuity of the function $g$.
	The choice of $g$ is motivated by a similar one in~\cite[Section~2.1]{exner2018bound}, which was optimal within a certain class of test functions.  We plug this $g$ into~\eqref{eq:bnd_stat_23} and get
\begin{multline*}
\mathcal I[u_*]=\frac{(2-e^{-2\alpha}-e^{-2\pi+2\alpha})}{2\beta}c_1^2+\frac{(-e^{-2\alpha}+e^{2\pi-2\alpha})}{2\beta}c_2^2
+\frac{(-e^{-2\alpha}+e^{2\alpha})}{2\beta}c_3^2+\\
\frac{(1-e^{-2\alpha})}{\beta}c_1c_2+\frac{(-1+e^{-2\alpha})}{\beta}c_1c_3+\frac{(-1+e^{-2\alpha})}{\beta}c_2c_3+\\
\frac {(1-a-e^{-\alpha}+ae^{-\pi+\alpha})\sqrt{\pi}}{4\beta^{\frac 32}}c_1
+\frac{(1-a-e^{-\alpha}+ae^{\pi-\alpha})\sqrt{\pi}}{4\beta^{\frac 32}}c_2+\\
\frac{(e^{-\alpha}-e^{\alpha})\sqrt{\pi}}{4\beta^{\frac 32}}c_3+\frac{4\pi\beta^2-4\pi\beta |a|\Theta_0+ a^2(\pi-\alpha)+\alpha}{8\beta^2}.
\end{multline*}
$\mathcal I[u_*]$ is a quadratic expression in $c_1$, $c_2$ and $c_3$. Minimizing $\mathcal I[u_*]$ with respect to these coefficients yields a unique solution $(c_1,c_2,c_3)$, where
	\begin{align*}
	c_1=& \frac{e^{\pi-2\alpha}\big((-1+a)e^{\pi}+(-1+a)e^{\pi+2\alpha}+2e^\alpha(-a+e^\pi)\big)\sqrt{\pi}\big(-1+\coth(\pi)\big)}{16\sqrt{\beta}} \\
	c_2=&\frac{\big(-1+a+(-1+a)e^{2\alpha}-2(-1+ae^\pi) e^\alpha\big)\sqrt{\pi}\big(-1+\coth(\pi)\big)}{16\sqrt{\beta}}\\
	c_3=&\frac{e^{-\alpha}\big(-a+e^\pi+(-1+a)\cosh(\pi-\alpha)\big)\sqrt{\pi}\csch(\pi)}{8\sqrt{\beta}}.
	\end{align*}
	We compute the corresponding $\mathcal I[u_*]$, taking $x=1/\beta$. We get $\mathcal I[u_*]= P_{\alpha,a}(x)$.
	This result together with~\eqref{eq:bnd_stat_21} complete the proof.
\end{proof}
\paragraph{{\bf Computation}} Bonnaillie has established in~\cite{bonnaillie2012harmonic} a lower bound, $\Theta_0^\mathrm{low}$, of $\Theta_0$ equal to $0.590106125-10^{-9}$.  For each $\alpha \in (0,\pi)$,  $a \in [-1,1)\setminus\{0\}$ and $x>0$  we set
$	P_{\alpha,a,\Theta_0^\mathrm{low}}(x)=Ax^2-\pi/2|a|\Theta_0^\mathrm{low}x+\pi/2,$
	 for $A$ in Proposition~\ref{prop:bnd_stat_2}.
then $P_{\alpha,a}(x)\leq P_{\alpha,a,\Theta_0^\mathrm{low}}(x)$.  Our \emph{rigorous} computation  shows that, for all $\alpha \in (0,\pi)$ and $a \in [-1,1)\setminus\{0\}$,  $P_{\alpha,a,\Theta_0^\mathrm{low}}(x)$ admits a minimum with respect to $x$, attained at a positive value $x_0=x_0(\alpha,a)$. Then, we use Mathematica to  plot the region of  the pairs $(\alpha,a)$ where  $\min_{x>0} P_{\alpha,a,\Theta_0^\mathrm{low}}(x)=P_{\alpha,a,\Theta_0^\mathrm{low}}(x_0)<0$. The shaded region in Figure~\ref{fig:bnd_stat_2} represents these pairs. Consequently, the corresponding  $P_{\alpha,a}(x_0)$ is negative and the corresponding $\mu(\alpha,a)$ is an eigenvalue. 
\begin{figure}
		\centering
		\includegraphics[scale=0.4]{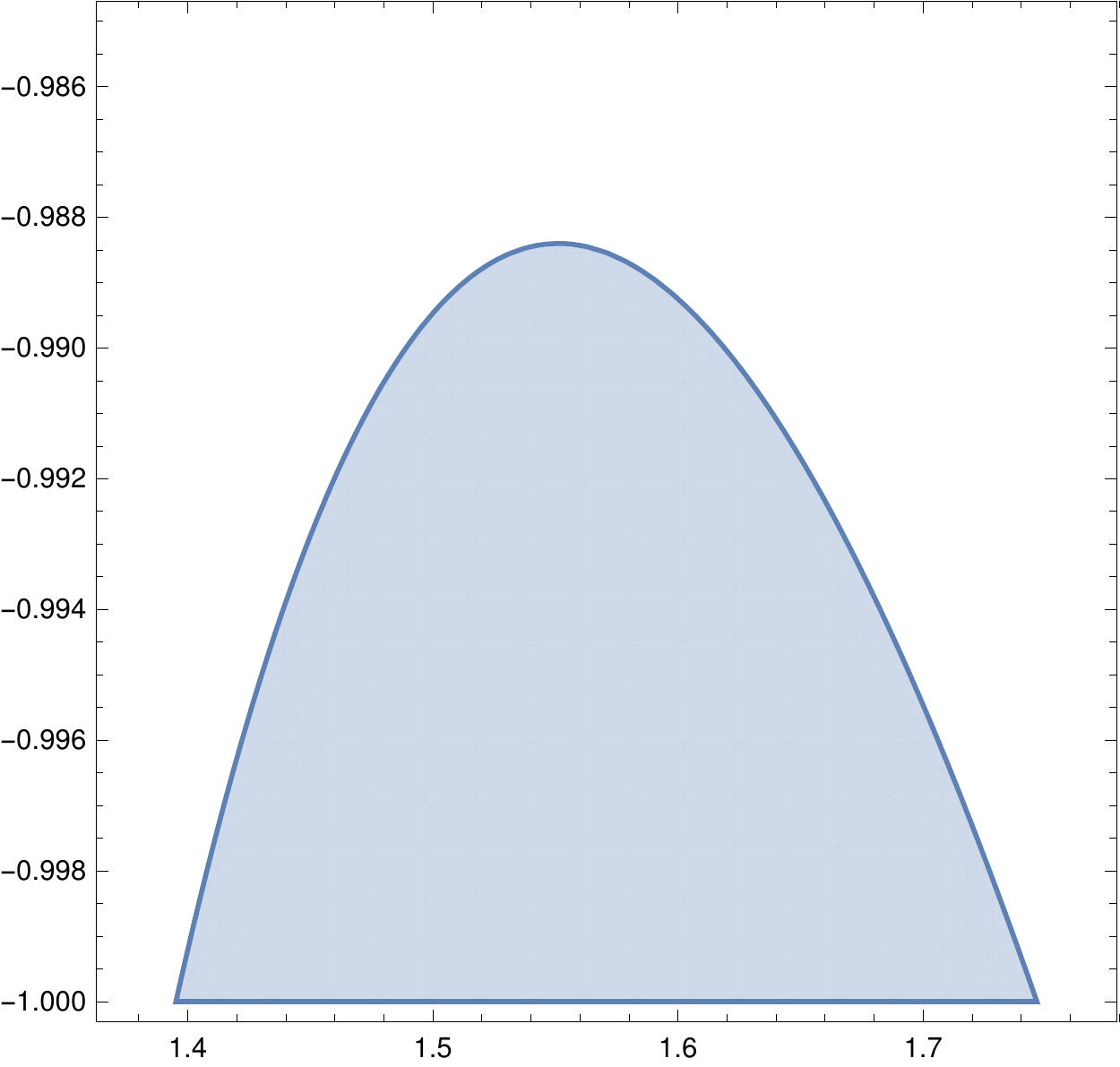}
	\caption{ The horizontal axis represents the angles $\alpha$ and the vertical axis represents the values of $a$. For $(\alpha,a)$ in the shaded region, $\mu(\alpha,a)$ is an eigenvalue.}
	\label{fig:bnd_stat_2}
\end{figure}
	
In the case where $\mu(\alpha,a)$ is the lowest eigenvalue of the operator $\mathcal H_{\alpha,a}$, let $v_{\alpha,a}$ be a corresponding normalized eigenfunction.
The following theorem reveals a decay of the eigenfunction $v_{\alpha,a}$, for large values of $|x|$. We omit the proof of Theorem~\ref{thm:v0-decay}, and we refer for details to the  similar proof in~\cite[Theorem~9.1]{bonnaillie2003analyse}.
\begin{theorem}\label{thm:v0-decay}
	Let $\alpha \in (0,\pi)$ and $a \in [-1,1)\setminus \{0\}$. Consider the case where $\mu(\alpha,a)$ is the lowest eigenvalue of the operator $\mathcal H_{\alpha,a}$ introduced in~\eqref{eq:P_alfa}, and let $v_{\alpha,a}$ be a corresponding normalized eigenfunction. For all $\delta$ such that $0<\delta<|a|\Theta_0-\mu(\alpha,a)$, there exists a constant $C_{\delta,\alpha}$  such that
	\[\|v_{\alpha,a}e^\phi\|_{L^2(\R^2_+)}+q_{\alpha,a}(v_{\alpha,a}e^\phi)\leq C_{\delta,\alpha},\]
	where  
	 $q_{\alpha,a}$ is the quadratic form in~\eqref{eq:qa}, and $\phi$ is a function defined in $\R^2_+$ as follows:	\[\phi(x)=\sqrt{|a|\Theta_0-\mu(\alpha,a)-\delta}\,|x|\,,\qquad \mbox{for all}\ x\in \R^2_+.\]
\end{theorem} 
\section{The linear problem}\label{sec:lin_prob}

\subsection{The linear operator}\label{sec:intro2}
The parameter dependent magnetic Shr\"odinger operator has been extensively studied in the literature in the case of regular/corner planar domains, submitted to \emph{smooth} magnetic fields (see e.g.~\cite{helffer1996semiclassical,bernoff1998onset,lu1999estimates,lu2000gauge,helffer2001magnetic,bonnaillie2005fundamental,fournais2010spectral}). 

Let $\mathfrak b>0$, $\Eb \in H^1(\Om,\R^2)$.  We consider the Neumann realization of the self-adjoint operator in the domain $\Om$ (satisfying Assumption~\ref{assump1}):
\begin{align}
\mathcal P_{\mathfrak b,\Eb}&=-(\nabla-i\mathfrak b\Eb)^2 \quad\mathrm{with} \label{eq:P0}\\
\dom \mathcal P_{\mathfrak b,\Eb}&=\big\{u\in L^2(\Om)~:~(\nabla-i\mathfrak b\Eb)^j u \in L^2(\Om),\,j\in \{1,2\},\,(\nabla-i\mathfrak b\Eb)\cdot\nu|_{\partial \Om}=0\big\}, \nonumber
\end{align}
where $\nu$ is a unit  normal vector of $\partial \Om$. The associated quadratic form is
\begin{align}\label{eq:Quad}
Q_{\mathfrak b,\Eb}(u)&=\int_{\Om}\big|(\nabla-i\mathfrak b\Eb)u\big|^2\,dx\quad \mathrm{with}\\
\dom Q_{\mathfrak b,\Eb}&=\left\{u\in L^2(\Om)~:~ (\nabla-i\mathfrak b\Eb)u \in L^2(\Om)\right\}\nonumber.
\end{align}

If $\Eb=\Fb$, where \emph{$\Fb  \in \Hd$ is the magnetic potential in~\eqref{A_1} satisfying $\curl\Fb= B_0={\mathbbm 1}_{\Om_1}+a{\mathbbm 1}_{\Om_2}$ for a fixed $a\in[-1,1)\setminus\{0\}$}, then the operator and the form domains are respectively
\begin{equation}\label{eq:domP_F}
\dom \mathcal P_{\mathfrak b,\Fb}=\big\{u\in H^2(\Om)~:~\nabla u\cdot\nu|_{\partial \Om}=0\big\}\quad \mathrm{and}\quad \dom Q_{\mathfrak b,\Fb}=H^1(\Om).
\end{equation}

 The bottom of the spectrum 
 \begin{equation}\label{eq:lmda_a}
 \lambda(\mathfrak b)=\inf_{\substack{u\in \dom Q_{\mathfrak b,\Fb}\\ u \neq 0 }}\frac{Q_{\mathfrak b,\Fb}(u)}{\|u\|^2_{L^2(\Om)}}.
 \end{equation} 
 is an eigenvalue.
\begin{rem}
	Compared to smooth magnetic fields cases, an extra argument is required to establish that the domains of $\mathcal P_{\mathfrak b,\Fb}$ and $Q_{\mathfrak b,\Fb}$ are independent of the parameter $\mathfrak b$, as in~\eqref{eq:domP_F}, in our case of a step magnetic field ($\curl \Fb=B_0$). This argument is given in Appendix~\ref{sec:regularity}. The independence of the domains from  $\mathfrak b$ will be crucial while applying the perturbation theory in Proposition~\ref{prop:monot} later.
\end{rem}
\subsection{Bounds of the ground-state energy}\label{sec:bound}
	\begin{theorem}\label{thm:lambda_lin}
	 Under Assumption~\ref{assump3}, there exist $\mathfrak b_0, C>0$  such that for all $\mathfrak b\geq \mathfrak b_0$, we have
		\[-Cb^{\frac 34}\leq \lambda(\mathfrak b)-\mathfrak b\min_{j\in\{1,...,n\}}\mu(\alpha_j,a)\leq C\mathfrak b^{\frac 35}.\]
	\end{theorem}
	Note that the error in the upper bound can be improved to be $\mathcal O(\mathfrak b^{1-\rho})$, for any $\rho \in (0,1/2)$ (see Remark~\ref{rem:up_opt}).
	
	The rest of this section is devoted to the proof of Theorem~\ref{thm:lambda_lin}. More precisely, the lower and upper bounds in this theorem are  established in Proposition~\ref{prop:Lower3} and~\ref{prop:Upper3} respectively.  The same techniques in~\cite{helffer2001magnetic} and~\cite{bonnaillie2005fundamental} are used here. We introduce a partition of unity to localize our analysis to different zones in $\Om$, then we compare our linear operator to an operator with a constant magnetic field  in $\R^2$ (if the zone is in $\Om\setminus \Gamma$), an operator with a constant magnetic field in $\R^2_+$ (if the zone meets the boundary away from $\Gamma$), an operator with a step magnetic field in $\R^2$, introduced in Section~\ref{sec:La} (if the zone meets $\Gamma$ away from $\partial \Om$), and finally an operator with a step magnetic field in $\R^2_+$,  introduced in Section~\ref{sec:new_model} (if the zone contains intersection points of  $\Gamma$ and $\partial \Om$).
	\subsubsection{Localization using a partition of unity}\label{sec:partition}

	Let $0<\rho<1$. For $R_0>0$, we can find a partition of unity, $\chi_j$, satisfying (when restricted to $\overline \Om$):
	\begin{multline}\label{eq:partition}
	\sum_j|\chi_j|^2=1, \quad \sum_j|\nabla \chi_j|^2 \leq CR_0^{-2}\mathfrak b^{2\rho}\quad
	\mathrm{and}\ \supp(\chi_j)\subset B(z_j,R_0\mathfrak b^{-\rho})\ \mbox{is such that}\\
	\begin{cases}
	\mathrm{either}& \supp(\chi_j)\cap(\partial \Om \cup \Gamma)=\emptyset,\\
	\mathrm{or}& z_j \in \partial \Om\setminus \Gamma\ \mathrm{and }\ \supp(\chi_j)\cap \Gamma=\emptyset,\\
	\mathrm{or}& z_j \in \Gamma\setminus \partial \Om \ \mathrm{and }\ \supp(\chi_j)\cap \partial \Om=\emptyset,\\
	\mathrm{or}& z_j =\mathsf p_j,
	\end{cases}
	\end{multline}
	where $C$ is independent of $R_0$ and $\mathfrak b$. \emph{The index $j$ is chosen such that $z_j=\mathsf p_j$, for $j \in \{1,...,n\}$}, where $\mathsf p_j \in \Gamma \cap \partial \Om$. For $u \in \dom Q_{\mathfrak b,\Fb}$, the IMS formula 
	asserts that
	\begin{multline}\label{eq:partition2}
	Q_{\mathfrak b,\Fb}(u)=\sum_{\blk}	Q_{\mathfrak b,\Fb}(\chi_j u)+\sum_{\bd}	Q_{\mathfrak b,\Fb}(\chi_j u)\\+\sum_{\br}	Q_{\mathfrak b,\Fb}(\chi_j u)+\sum_{\T}	Q_{\mathfrak b,\Fb}(\chi_j u)-\sum_j\big\||\nabla \chi_j|u\big\|^2_{L^2(\Om)},
	\end{multline}
	where
	\begin{align*}
	\blk:= &\{j~:~z_j \in \Om\setminus \Gamma\},&\qquad  \bd:= &\{j~:~z_j \in \partial \Om\setminus \Gamma\},\\
	\br:= &\{j~:~z_j \in \Gamma\setminus \partial \Om\},&\qquad
	\T:= &\{j~:~j=1,...,n \}.
	\end{align*}
	We will  optimize later the choice of $\rho$ and $R_0$ for our various problems.
	\subsubsection{Change of variables}\label{sec:Psi}
	In order to study the energy contribution near $\Gamma \cap \partial \Om$, we will carry out the computation in adapted coordinates in this zone.
	Recall that we are working under  Assumption~\ref{assump1} (see also Notation~\ref{not:alfa}). For $j  \in\{1,...,n\}$, there exist $r_j>0$ and a local diffeomorphism $\Psi=\Psi_j$ of $\R^2$ satisfying the following~(see Appendix~\ref{A:Psi}):  
	\begin{equation}\label{eq:Psi2}
	\Psi(\mathsf p_j)=(0,0)\,,\qquad |J_\Psi|(\mathsf p_j)=|J_{\Psi^{-1}}|(0,0)=1,
	\end{equation}
	and there exists a neighbourhood $\mathcal U_j$ of $(0,0)$ such that 
	
	\[\Psi\big(B(\mathsf p_j,r_j)\cap \Om_1\big) =\mathcal U_j \cap  D^{\alpha_j}_1\,,\quad  \Psi\big(B(\mathsf p_j,r_j)\cap \Om_2\big) =\mathcal U_j \cap  D^{\alpha_j}_2,\]
	and consequently, 
	\[\Psi\big(B(\mathsf p_j,r_j)\cap (\partial \Om_1 \setminus \Gamma)\big) =\mathcal U_j \cap \R_+\times\{0\},\]
	\[\Psi\big(B(\mathsf p_j,r_j)\cap (\partial \Om_2 \setminus \Gamma)\big) =\mathcal U_j \cap \R_-\times\{0\},\]
	\[\Psi\big(B(\mathsf p_j,r_j)\cap \Gamma \big) =\mathcal U_j \cap (\hat x_2=\hat x_1\tan \alpha_j).\]
	Here, $(\hat x_1,\hat x_2):=\Psi(x_1,x_2)$, and the sets $D^{\alpha_j}_1$ and $D^{\alpha_j}_2$ were defined in~\eqref{eq:A_alfa}.
	We assume further that the radii $r_j$ are sufficiently small so that $\big(B(\mathsf p_j,r_j)\big)_{j\in\{1,...,n\}}$ is a family of disjoint balls.	
	The smoothness of $\Psi$, the fact that $\{1,...,n\}$ is finite, the assumptions in~\eqref{eq:Psi2} and a Taylor expansion  prove the existence of $C>0$, \emph{independent of $j$}, such that the Jacobians $J_\Psi$ and $J_{\Psi^{-1}}$ satisfy
	\begin{equation}\label{eq:Psi3}
\big||J_\Psi(x)|-1 \big|\leq C \ell\qquad \mathrm{and}\qquad \big||J_{\Psi^{-1}}(\hat x)|-1 \big|\leq C \ell,
	\end{equation}
	for all $x \in B(\mathsf p_j,\ell) \subset B(\mathsf p_j,r_j)$ and $\hat x=\Psi(x)$. 
	Let $\Eb=(E_1,E_2)\in H^1(\Om;\R^2)$ be such that $\curl \Eb=B$, for $B \in L^2(\R^2)$, 
and let $u \in \dom Q_{\mathfrak b,\Eb}$ (see~\eqref{eq:Quad}) such that $\supp u \subset B(\mathsf p_j,r_j)$. Consider the magnetic potential $\hat {\Eb}=(\hat{E}_1,\hat{E}_2) \in H^1\big(\Psi\big(B(\mathsf p_j,r_j)\big)\cap \R^2_+,\R^2\big)$  satisfying $\hat{E}_1\,d\hat{x}_1+\hat{E}_2\,d\hat{x}_2=E_1\,d x_1+E_2\,dx_2$, and the function $\hat{u}$, defined in $\Psi(B(\mathsf p_j,r_j))\cap \R^2_+$ by $\hat{u}(\hat{x})=u\big(\Psi^{-1}(\hat{x})\big)$.  Furthermore, let 
	\[\hat B(\hat{x})=B\big(\Psi^{-1}(\hat{x})\big),\qquad \mbox{for all}\ \hat x \in \Psi(B(\mathsf p_j,r_j))\cap \R^2_+.\] 
	One can check that 
	\begin{equation}\label{eq:curl-hat}
	\curl \hat{\Eb}=\partial_{\hat{x}_1}\hat{E}_2-\partial_{\hat{x}_2}\hat{E}_1=\hat{B}J_{\Psi^{-1}},
	\end{equation}
	and
	\begin{equation}
	 Q_{\mathfrak b,\Eb}(u)
	=\int_D\sum_{1\leq k,m\leq 2} G_{k,m}(\hat x)\big(\partial_{\hat x_k}-i\mathfrak b\hat E_k\big)\hat {u}(\hat {x})\overline{\big(\partial_{\hat x_m}-i\mathfrak b\hat E_m\big)\hat {u}(\hat {x})}\,|J_{\Psi^{-1}}(\hat {x})|\,d\hat {x}.\label{eq:Q_trans}
	\end{equation}
	Here 
	$D=\Psi(B(\mathsf p_j,r_j))\cap \R^2_+$ and $G_{k,m}(\hat x)$ are the elements of the matrix $G(\hat x)=(d \Psi)(d \Psi)^t\,_{|\Psi^{-1}(\hat x)}$.

Note that $G(0,0)$ is the identity matrix. Then, for any $\ell<r_j$, one may apply Taylor's formula in $\Psi(B(\mathsf p_j,\ell)\big)$ to prove that
	\begin{equation}\label{eq:Bon1}|G_{k,m}(\hat x)-\delta_{k,m}|\leq C\ell,\end{equation}
	for some $C>0$ independent of $j$. 
	The following lemma presents a particular transformation, that will allow us to express a given vector field in a canonical manner.
	\begin{lemma}\label{lem:gauge1}
	Let $a\in [-1,1)\setminus \{0\}$, and $B(0,\ell)\subset \Psi\big(B(\mathsf p_j,r_j)\big)$ be a ball of radius $\ell$.  Consider the  vector potential $\Fb\in \Hd$ satisfying $\curl \Fb=\mathbbm 1_{\Om_1}+a\mathbbm 1_{\Om_2}$.  There exists a function $\varphi_\ell\in H^2\big(B(0,\ell)\cap \R^2_+\big)$ such that the vector potential $\hat \Fb_{\rm g}:=\hat \Fb-\nabla_{\hat{x}_1,\hat{x}_2} \varphi_\ell$, defined in $B(0,\ell)\cap \R^2_+$, satisfies
		\begin{equation*}
		\big(\hat{F}_{\rm g}\big)_1=0,\quad
		\big(\hat{F}_{\rm g}\big)_2= A_{\alpha,a}+f,
		\end{equation*}
		where $A_{\alpha,a}$ is the potential introduced in~\eqref{eq:P_alfa}, $f$ is a continuous function satisfying $|f(\hat{x}_1,\hat{x}_2)|\leq C (\hat{x}^2_1+|\hat x_1\hat x_2|)$, for some $C>0$ independent of $j$.
	\end{lemma}
	\begin{proof}
		Define 
		\[\varphi_\ell(\hat{x}_1,\hat{x}_2)=\int_0^{\hat{x}_1} \hat{F}_1(\hat{x}_1',\hat{x}_2)\,d\hat{x}_1'+\int_0^{\hat{x}_2} \hat{F}_2(0,\hat{x}_2')\,d\hat{x}_2',\]
		for $(\hat{x}_1,\hat{x}_2)\in B(0,\ell)\cap \R^2_+$. Obviously  $\big(\hat{F}_{\rm g}\big)_1=0$. Furthermore, a simple computation using~\eqref{eq:Psi3} and~\eqref{eq:curl-hat}  yields
		\begin{equation*}
		\big(\hat{F}_{\rm g}\big)_2(\hat{x}_1,\hat{x}_2)=\int_0^{\hat{x}_1}\big(1+\mathcal O(\hat{x}_1)\big)\hat B(\hat{x}_1',\hat{x}_2)\,d\hat{x}_1'.
		\end{equation*}
	 Recalling the definition of $\hat{B}$, we complete the proof. Note that the independence of the constant $C$ from $j$ follows from the fact that the points $\mathsf p_j$ are finite.
	\end{proof}
	\subsubsection{Lower bound of $\lambda(\mathfrak b)$}\label{sec:low}
		\emph{In this section we are working under Assumption~\ref{assump1} and we use Notation~\ref{not:alfa}. We do not require Assumption~\ref{assump3} to be fulfilled}. Let $u\in \dom Q_{\mathfrak b,\Fb}$ (see~\eqref{eq:Quad}). We use the relation~\eqref{eq:partition2} to localize the estimates. 
		The error term  $\sum_j\big\||\nabla \chi_j|u\big\|^2_{L^2(\Om)}$ is estimated using~\eqref{eq:partition}
		\begin{equation}\label{eq:part-err}
		\sum_j\big\||\nabla \chi_j|u\big\|^2_{L^2(\Om)}\leq CR_0^{-2}\mathfrak b^{2\rho}\|u\|^2_{L^2(\Om)}.
		\end{equation}
		Recall that the magnetic potential $\Fb \in \Hd$ satisfies $\curl \Fb=B_0={\mathbbm 1}_{\Om_1}+a{\mathbbm 1}_{\Om_2}$. 
		\begin{flushleft}{\itshape Estimating $\sum_{\blk} Q_{\mathfrak b,\Fb}(\chi_j u)$.}\end{flushleft}
		Let $j \in \blk$. Notice that  $\chi_j u \in H^1(\Om)$, and $\chi_j u$ is supported in $\Om$, then a well-known spectral property (see~\cite[Lemma 1.4.1]{fournais2010spectral}) assures that 
		\begin{equation}\label{eq:part-int}
		\sum_{\blk}Q_{\mathfrak b,\Fb}(\chi_j u)\geq \sum_{\blk} |B_0|\mathfrak b \int_{\Om} \mathfrak |\chi_j|^2|u|^2\,dx 
		\geq |a|\mathfrak b \sum_{\blk} \|\chi_j u\|^2_{L^2(\Om)}. 
		\end{equation}
		\begin{flushleft}{\itshape Estimating $\sum_{\bd} Q_{\mathfrak b,\Fb}(\chi_j u)$.}\end{flushleft}
		 Let $j \in \bd$. Notice that $\curl\Fb$ is constant in $B(z_j,\allowbreak R_0\mathfrak b^{-\rho})$. This allows us to use the local lower bound estimates in~\cite[Section 8.2.2]{fournais2010spectral}, in the case of a smooth magnetic field. Using our notation, we present here the result in~\cite{fournais2010spectral}: there exists a universal constant $C>0$ (independent of $j$) such that when $\mathfrak b$ is sufficiently large, 
		\begin{equation*}
		Q_{\mathfrak b,\Fb}(v)\geq \big( |a|\Theta_0\mathfrak b-C\mathfrak r_1(R_0,\mathfrak b)\big) \|v\|^2_{L^2(\Om)} ,\end{equation*}
		 where $v$ is any function such that $v \in \dom Q_{\mathfrak b,\Fb}$ and $\supp(v)\subset B(z_j,R_0\mathfrak b^{-\rho})$,  and
		 	\begin{equation}\label{eq:r1}
		 	\mathfrak r_1(R_0,\mathfrak b)=\mathfrak b^{\frac12}+\eta\mathfrak b+R_0^4\eta^{-1}\mathfrak b^{2-4\rho}+R^2_0\mathfrak b^{1-\rho},
		 	\end{equation}
		for arbitrary $\eta \in (0,1)$. Consequently for $v=\chi_j u$, we conclude that 
		\begin{equation}\label{eq:part-bd}
		\sum_{\bd} Q_{\mathfrak b,\Fb}(\chi_j u)\geq \big(|a|\Theta_0\mathfrak b -C\mathfrak r_1(R_0,\mathfrak b)\big) \sum_{\bd} \|\chi_j u\|^2_{L^2(\Om)}.\end{equation}

		\begin{flushleft}{\itshape Estimating $\sum_{\br} Q_{\mathfrak b,\Fb}(\chi_j u)$.}\end{flushleft}
		 Here, we will use the local transformation $\Phi$ introduced in Section~\ref{sec:bc}. In particular, a key-ingredient is the change of gauge in Lemma~\ref{lem:Anew2}, that will link (locally) the form $Q_{\mathfrak b,\Fb}$ to the spectral value $\beta_a$ defined in Section~\ref{sec:La}.

		Let $j \in \br$. Consider a function $v \in \dom Q_{\mathfrak b,\Fb}$
		  such that $\supp v \subset B(z_j,R_0\mathfrak b^{-\rho})$. After possibly performing a translation in the $s$ variable, we may assume that $\Phi^{-1}(z_j)=(0,0)$.
		Assume that $\mathfrak b$ is sufficiently large so that $B(z_j,R_0\mathfrak b^{-\rho}) \subset \Gamma(t_0)\cap \Om$ (see Appendix~\ref{sec:bc}). 
		The transformation $\Phi$  associates to  $\Fb$ the vector potential $\tilde{\Fb}$ in~\eqref{eq:A_tild1},  and to $v$ a function $\tilde{v}=v\circ\Phi$ defined in $\Phi^{-1} \big(B(z_j,R_0\mathfrak b^{-\rho})\big)$.   
		Using the estimates in~\eqref{eq:a2}, the change of variables formulae in~\eqref{eq:A_tild2}, and the support of $\tilde v$, one may deduce the existence of $C>0$, independent of $j$, such that 
		\begin{multline}\label{eq:J1}
		(1-CR_0\mathfrak b^{-\rho}) \int_{\Phi^{-1} \big(B(z_j,R_0\mathfrak b^{-\rho})\big)}\big|(\nabla-i\mathfrak b\tilde{\Fb})\tilde{v}\big|^2\,dx \leq  Q_{\mathfrak b,\Fb}(v) \\\leq (1+CR_0\mathfrak b^{-\rho})\int_{\Phi^{-1} \big(B(z_j,R_0\mathfrak b^{-\rho})\big)}\big|(\nabla-i\mathfrak b\tilde{\Fb})\tilde{v}\big|^2\,dx, 
		\end{multline}

		Next, we will make profit of the gauge result in Lemma~\ref{lem:Anew2}. Thanks to~\eqref{eq:a2} and the support of $v$, one may note the existence of $c_0>0$ such that  $\Phi^{-1} \big(B(z_j,R_0\mathfrak b^{-\rho})\big)\subset B(0,c_0 R_0\mathfrak b^{-\rho}) \subset \left(-|\Gamma|/2,|\Gamma|/2\right)\times(-t_0,t_0)$, for large $\mathfrak b$. 
		We define 
		\[
		\tilde v_{\rm g}(s,t)=
		\tilde v(s,t)e^{-i\mathfrak b\omega(s,t)},\] 
		for $(s,t) \in \Phi^{-1} \big(B(z_j,R_0\mathfrak b^{-\rho})\big)$, where $\omega=\omega_\ell$  is the function in Lemma~\ref{lem:Anew2} and $\ell=c_0R_0\mathfrak b^{-\rho}$.   One can easily check that
		\begin{equation}\label{eq:Q_tild}
		\int_{\Phi^{-1} \big(B(z_j,R_0\mathfrak b^{-\rho})\big)}\big|(\nabla-i\mathfrak b\tilde{\Fb})\tilde{v}\big|^2\,dx =\int_{\Phi^{-1} \big(B(z_j,R_0\mathfrak b^{-\rho})\big)}\big|(\nabla-i\mathfrak b\tilde{\Fb}_{\rm g})\tilde{v}_{\rm g}\big|^2\,dx.
		\end{equation}
		Consequently, it suffices  to estimate the  right hand side of~\eqref{eq:Q_tild}. We extend $\tilde{v}$ and $\tilde{v}_{\rm g}$ by zero in $\R^2$. Using Cauchy's inequality and the support of $\tilde{v}_{\rm g}$, we get for $\delta \in(0,1)$,
		\begin{multline}\label{eq:J2}
	\int_{\Phi^{-1} \big(B(z_j,R_0\mathfrak b^{-\rho})\big)}\big|(\nabla-i\mathfrak b\tilde{\Fb}_{\rm g})\tilde{v}_{\rm g}\big|^2\,dx
		\geq (1-\mathfrak b^{-\delta})\int_{\R^2}\Big(\big|\big(\partial _s+i\mathfrak b\sigma t\big) \tilde{v}_{\rm g}\big|^2+|\partial_t\tilde{v}_{\rm g}|^2 \Big)\,ds\,dt\\-CR_0^4\mathfrak b^{2-4\rho+\delta}\int_{\R^2}|\tilde{v}_{\rm g}|^2\,ds\,dt,
		\end{multline}
		where
		$\sigma=\sigma(s,t)=\mathbbm 1_{\R_+}(t)+a\mathbbm 1_{\R_-}(t)$.
		Performing a suitable change of gauge and a scaling, one can use the spectral properties of the operator $\mathcal L_a$, in Section~\ref{sec:La}, to conclude that
		\begin{equation}\label{eq:T1}
		\int_{\R^2}\Big(\big|\big(\partial _s+i\mathfrak b\sigma t\big) \tilde{v}_{\rm g}\big|^2+|\partial_t\tilde{v}_{\rm g}|^2 \Big)\,ds\,dt\geq \beta_a\mathfrak b\int_{\R^2}|\tilde{v}_{\rm g}|^2\,ds\,dt.
		\end{equation}
		Implementing~\eqref{eq:T1} in~\eqref{eq:J2} yields
		\begin{equation}\label{eq:J3}
		\int_{\Phi^{-1} \big(B(z_j,R_0\mathfrak b^{-\rho})\big)}\big|(\nabla-i\mathfrak b\tilde{\Fb}_{\rm g})\tilde{v}_{\rm g}\big|^2\,dx \geq \big( \beta_a\mathfrak b-C\mathfrak b^{1-\delta}-CR_0^4\mathfrak b^{2-4\rho+\delta}\big)\int_{\R^2}|\tilde{v}_{\rm g}|^2\,ds\,dt.
		\end{equation}
		
		Now, we estimate the $L^2$-norm of $\tilde{v}_{\rm g}$. We have
		\[\int_{\R^2}|\tilde{v}_{\rm g}|^2\,ds\,dt=\int_{\R^2}|\tilde{v}|^2\,ds\,dt=\int_{B(z_j,R_0\mathfrak b^{-\rho})}|v|^2\,J_{\Phi^{-1}}\,dx.\]
		Hence by~\eqref{eq:a2} and the support of $v$, there exists $C>0$ independent of $j$ such that
		\begin{equation}\label{eq:J4}
		(1-CR_0\mathfrak b^{-\rho}) \int_\Omega|v|^2\,dx\leq \int_{\R^2}|\tilde{v}_{\rm g}|^2\,ds\,dt \leq (1+CR_0\mathfrak b^{-\rho}) \int_\Omega|v|^2\,dx.
		\end{equation}
		Plug~\eqref{eq:Q_tild},~\eqref{eq:J3},~and~\eqref{eq:J4} into~\eqref{eq:J1}  to obtain 
		\begin{equation*}
		Q_{\mathfrak b,\Fb}(v) \geq \big(\beta_a\mathfrak b-C\mathfrak r_2(R_0,\mathfrak b)\big) \int_\Omega|v|^2\,dx,
		\end{equation*}	
		where 
			\begin{equation}\label{eq:r2}
			\mathfrak r_2(R_0,\mathfrak b)=R_0\mathfrak b^{1-\rho}+\mathfrak b^{1-\delta}+R_0^4\mathfrak b^{2-4\rho+\delta}.
			\end{equation}
		 We consider now the particular case where $v=\chi_j u$, and we conclude that
		\begin{equation}\label{eq:part-bar}
		\sum_{\br}	Q_{\mathfrak b,\Fb}(\chi_j u) \geq \big(\beta_a \mathfrak b-C\mathfrak r_2(R_0,\mathfrak b)\big) \sum_{\br} \|\chi_j u\|^2_{L^2(\Om)}.
		\end{equation}	
		
		\begin{flushleft}{\itshape Estimating  $\sum_{\T} Q_{\mathfrak b,\Fb}(\chi_j u)$.}\end{flushleft}
		The techniques we use below are quite similar to the ones used in estimating the $\sum_{\br}Q_{\mathfrak b,\Fb}(\chi_j u)$. We will make profit of the local transformation $\Psi$ introduced in Section~\ref{sec:Psi}, and particularly of the change of gauge in Lemma~\ref{lem:gauge1}, to link locally the form $Q_{\mathfrak b,\Fb}$ to  $\mu(\cdot,a)$ defined in~\eqref{eq:mu_alfa}.
		
		Let $j \in \T$. Consider the function $v \in \dom Q_{\mathfrak b,\Fb}$ such that $\supp v \subset B(z_j,R_0\mathfrak b^{-\rho})$. We use the change of variables introduced in Section~\ref{sec:Psi}, valid in a neighbourhood of $z_j$, to send locally the domain in $\Om$ onto $\R^2_+$.  $\mathfrak b$ is assumed large enough so that
	$B(z_j,R_0\mathfrak b^{-\rho}) \subset B(z_j,r_j)$.
 We associate to  $v$ the function $\hat{v}=v\circ\Psi^{-1}$, defined in $\Psi\big( B(z_j,R_0\mathfrak b^{-\rho})\big)$.
We may use the transformation formula in~\eqref{eq:Q_trans} and the  properties in~\eqref{eq:Psi3} and~\eqref{eq:Bon1} to conclude that
		\begin{multline}\label{eq:Q_hat}
		(1-CR_0\mathfrak b^{-\rho}) \int_{\Psi( B(z_j,R_0\mathfrak b^{-\rho}))\cap \R^2_+}\big|(\nabla-i\mathfrak b\hat{\Fb})\hat{v}\big|^2\,dx\leq Q_{\mathfrak b,\Fb}(v)\\\leq (1+CR_0\mathfrak b^{-\rho})\int_{\Psi( B(z_j,R_0\mathfrak b^{-\rho}))\cap \R^2_+}\big|(\nabla-i\mathfrak b\hat{\Fb})\hat{v}\big|^2\,dx,
		\end{multline}
	where $\hat \Fb$ is the transform of $\Fb$ by $\Psi$, and $C>0$ is a constant independent of $j$.

		In addition, due to the support of $v$ and~\eqref{eq:Psi3}, we note the existence of $c_1>0$  such that 
		$\Psi \big(B(z_j,R_0\mathfrak b^{-\rho})\big)\subset B(0,c_1 R_0\mathfrak b^{-\rho}) \subset \Psi \big(B(z_j,r_j)\big)$, for large $\mathfrak b$. 
		Consequently, the gauge transform in Lemma~\ref{lem:gauge1} allows us to write
		\begin{multline}\label{eq:gauge2}
		\int_{\Psi( B(z_j,R_0\mathfrak b^{-\rho}))\cap \R^2_+}\big|(\nabla-i\mathfrak b\hat{\Fb})\hat{v}\big|^2\,dx\\=\int_{\Psi( B(z_j,R_0\mathfrak b^{-\rho}))\cap \R^2_+} \big|(\nabla-i\mathfrak b\hat{\Fb}_{\rm g})\hat{v}_{\rm g}\big|^2\,d\hat{x},
		\end{multline}
	where
		$\hat{v}_{\rm g}(\hat x)=
		\hat{v}(\hat x)e^{-i\mathfrak b\varphi(\hat x)}$, for $\hat x\in \Psi\big( B(z_j,R_0\mathfrak b^{-\rho})\big) \cap \R^2_+$.
	 Here $\varphi=\varphi_\ell$, for $\ell=c_1R_0\mathfrak b^{- \rho}$, is the gauge function in Lemma~\ref{lem:gauge1}, and $\hat{\Fb}_{\rm g}$ is  the magnetic potential in the aforementioned lemma.  Let $\delta\in(0,1)$. Recall the potential $\Ab_{\alpha,a}$ introduced in~\eqref{eq:P_alfa}. Extending $\hat{v}$ and $\hat{v}_{\rm g}$ by zero in $\R^2_+$, the Cauchy's inequality applied in~\eqref{eq:gauge2}, and the support of the function $\hat{v}$ imply
		\begin{multline}\label{eq:gauge3}
		\int_{\Psi( B(z_j,R_0\mathfrak b^{-\rho}))\cap \R^2_+}\big|(\nabla-i\mathfrak b\hat{\Fb})\hat{v}\big|^2\,dx\geq (1-\mathfrak b^{-\delta})\int_{\R^2_+} \big|(\nabla-i\mathfrak b\Ab_{\alpha_j,a})\hat{v}_{\rm g}\big|^2\,d\hat{x}\\-CR_0^4\mathfrak b^{2-4\rho+\delta}\int_{\R^2_+}|\hat{v}_{\rm g}\big|^2\,d\hat{x},
		\end{multline}
	where  $\alpha_j$ is the corresponding angle to the point $z_j$, defined in Notation~\ref{not:alfa}. Hence, using a simple scaling argument we write
		\begin{equation}\label{eq:gauge4}
		\int_{\Psi( B(z_j,R_0\mathfrak b^{-\rho}))\cap \R^2_+}\big|(\nabla-i\mathfrak b\hat{\Fb})\hat{v}\big|^2\,dx\geq \Big(\mu(\alpha_j,a)\mathfrak b-C\mathfrak b^{1-\delta}-CR_0^4\mathfrak b^{2-4\rho+\delta}\Big)\int_{\R^2_+}|\hat{v}_{\rm g}\big|^2\,d\hat{x},
		\end{equation}
		where $\mu(\alpha_j,a)$ is the value in~\eqref{eq:mu_alfa} corresponding to the angle $\alpha_j$. 
		But 
		\[\int_{\R^2_+}|\hat{v}_{\rm g}\big|^2\,d\hat{x}=\int_{B(z_j,R_0\mathfrak b^{-\rho})\cap\Omega} |v|^2\,|J_\Psi|\,dx.\]
		Thus, using~\eqref{eq:Psi3} we get
		\begin{equation}\label{eq:v_v_hat}
		(1-CR_0\mathfrak b^{-\rho})\int_\Omega |v|^2\,dx\leq \int_{\R^2_+}|\hat{v}_{\rm g}\big|^2\,d\hat{x}\leq (1+CR_0\mathfrak b^{-\rho})\int_\Omega |v|^2\,dx.
		\end{equation}
		Plug~\eqref{eq:gauge4} and~\eqref{eq:v_v_hat} into~\eqref{eq:Q_hat} to obtain
		\begin{equation}\label{eq:gauge5}
		Q_{\mathfrak b,\Fb}(v)\geq \big(\mu(\alpha_j,a)\mathfrak b-C\mathfrak r_3(R_0,\mathfrak b)\big)\|v\|^2_{L^2(\Om)},
		\end{equation}	
		where 
		\begin{equation}\label{eq:r3}
			\mathfrak r_3(R_0,\mathfrak b)=R_0\mathfrak b^{1-\rho}+\mathfrak b^{1-\delta}+R_0^4\mathfrak b^{2-4\rho+\delta}.
			\end{equation}
		Taking the particular case $v=\chi_j u$, we infer from~\eqref{eq:gauge5} that
		\begin{equation}\label{eq:part-T}
		\sum_{\T}Q_{\mathfrak b,\Fb}(\chi_j u) \geq \big(\min_{j\in\T}\mu(\alpha_j,a)\mathfrak b -C\mathfrak r_3(R_0,\mathfrak b)\big) \sum_{\T} \|\chi_j u\|_{L^2(\Omega)}^2.
		\end{equation}
	
	Let 
	\begin{equation}\label{eq:r}
	\mathfrak r(R_0,\mathfrak b)=\max\big(\mathfrak r_1(R_0,\mathfrak b),\mathfrak r_2(R_0,\mathfrak b),\mathfrak r_3(R_0,\mathfrak b)\big),
	\end{equation} 
	where $\mathfrak r_1, \mathfrak r_2, \mathfrak r_3$   defined in~\eqref{eq:r1},~\eqref{eq:r2} and~\eqref{eq:r3} respectively. 
		The estimates in~\eqref{eq:part-err},~\eqref{eq:part-int},~\eqref{eq:part-bd}, \eqref{eq:part-bar}, and~\eqref{eq:part-T} give the following lower bound of $Q_{\mathfrak b,\Fb}(u)$: 
		\begin{align}\label{eq:lower_final}
		Q_{\mathfrak b,\Fb}(u) &\geq  |a|\mathfrak b\sum_{\blk} \|\chi_j u\|^2_{L^2(\Om)} + |a|\Theta_0\mathfrak b \sum_{\bd}\|\chi_j u\|^2_{L^2(\Om)}\nonumber\\&\quad+\beta_a \mathfrak b \sum_{\br} \|\chi_j u\|_{L^2(\Om)}^2+\min_{j\in\T}\mu(\alpha_j,a)\mathfrak b  \sum_{\T} \|\chi_j u\|_{L^2(\Om)}^2 \nonumber\\&\quad-C\big(\mathfrak r(R_0,\mathfrak b)+R_0^{-2}\mathfrak b^{2\rho}\big)\|u\|^2_{L^2(\Om)}.
		\end{align}
	 We may extract  particular results from the discussion done above, which we present  in the following two propositions: 
	
	\begin{proposition}\label{prop:Ub}
		There exists $C>0$, and for all $R_0>1$ there exists $\mathfrak b_0>0$ such that for $\mathfrak b\geq \mathfrak b_0$ and $u \in \dom Q_{\mathfrak b,\Fb}$, it holds
		\[Q_{\mathfrak b,\Fb}(u)\geq \int_{\Om} \big(U_{\mathfrak b}(x)-CR_0^{-2}\mathfrak b^{2\rho}\big)|u(x)|^2\,dx,\]
		where
	 	\begin{equation*}
			U_{\mathfrak b}(x)=
			\begin{cases}
			|a|\mathfrak b&\qquad \dist(x,\partial \Om \cup\Gamma)\geq R_0\mathfrak b^{-\rho},\\
			\beta_a\mathfrak b-C\mathfrak r(R_0,\mathfrak b)&\qquad \dist(x,\partial \Om)\geq R_0\mathfrak b^{-\rho}\, \&\, \dist(x,\Gamma)< R_0\mathfrak b^{-\rho},\\
			|a|\Theta_0\mathfrak b-C\mathfrak r(R_0,\mathfrak b)&\qquad \dist(x,\partial \Om)< R_0\mathfrak b^{-\rho}\, \&\, x \notin \bigcup\limits_{j=1}^{n} B(\mathsf p_j,R_0\mathfrak b^{-\rho}),\\
			\mu(\alpha_j,a)\mathfrak b-C\mathfrak r(R_0,\mathfrak b)&\qquad j \in\{1,...,n\}\,,\ x\in  B(\mathsf p_j,R_0\mathfrak b^{-\rho}),
			\end{cases}
				\end{equation*}
			where $\mathfrak r(R_0,\mathfrak b)$ is the term in~\eqref{eq:r},
			  $\mu(\alpha_j,a)$ and $\Theta_0$ are  introduced in~\eqref{eq:mu_alfa} and~\eqref{eq:theta01} respectively.
	\end{proposition}
	\begin{proof}
		Let $R_0>1$ and $\mathfrak b>0$ be large. Define the following partition of $\Om$:
		\[Z_1=\big\{x \in \Om~:~\dist(x,\partial \Om \cup\Gamma)\geq R_0\mathfrak b^{-\rho}\big\}.\]
		\[Z_2=\big\{x \in \Om~:~\dist(x,\partial \Om)\geq R_0\mathfrak b^{-\rho}\,,\,\dist(x,\Gamma)< R_0\mathfrak b^{-\rho}\big\}.\]
		\[Z_3=\big\{x \in \Om~:~\dist(x,\partial \Om)< R_0\mathfrak b^{-\rho}\,,\,x \notin \bigcup\limits_{j \in\T} B(\mathsf p_j,R_0\mathfrak b^{-\rho})\big\}.\]
		\[Z_4=\bigcup_{j \in\T}B(\mathsf p_j,R_0\mathfrak b^{-\rho})\cap \Om,\]
		and consider the partition of unity in Section~\ref{sec:partition}. Clearly, we have
	\[\bigcup\limits_{j \in\T} B(z_j,R_0\mathfrak b^{-\rho})\subset Z_4=(Z_1\cup Z_2\cup Z_3)^\complement,\ \bigcup\limits_{j \in \partial \Om \setminus\Gamma} B(z_j,R_0\mathfrak b^{-\rho})\subset (Z_1\cup Z_2)^\complement,\]
	\[\mathrm{and}\ \bigcup\limits_{j \in \Gamma \setminus\partial \Om} B(z_j,R_0\mathfrak b^{-\rho})\subset (Z_1)^\complement.\]
	Hence, using the lower bounds established in~\eqref{eq:part-int},~\eqref{eq:part-bd}, \eqref{eq:part-bar}, and~\eqref{eq:part-T}  and the ordering $\max_j\mu(\alpha_j,a)\leq|a|\Theta_0\leq\beta_a\leq|a|$ (see Theorem~\ref{thm:ess_theta} and Section~\ref{sec:La}), the IMS formula yields the proof.
	\end{proof}
	
	Note again that $\min_{j\in\{1,...,n\}}\mu(\alpha_j,a)\leq |a|\Theta_0\leq\beta_a\leq|a|$ (see~Section~\ref{sec:La} and Theorem~\ref{thm:ess_theta}). \emph{We choose  $R_0=1$, $\rho=3/8$,  $\delta=1/4$ and $\eta=\mathfrak b^{-1/4}$} in~
\eqref{eq:lower_final}. Consequently, the min-max principle implies the following:
	\begin{proposition}\label{prop:Lower3}
		Under Assumption~\ref{assump1}, there exist $\mathfrak b_0,C>0$  such that for all $\mathfrak b\geq \mathfrak b_0$,
		\[\lambda(\mathfrak b) \geq \min_{j\in\{1,...,n\}}\mu(\alpha_j,a)\mathfrak b-C\mathfrak b^\frac 34,\]
		where $\lambda(\mathfrak b)$ and $\mu(\alpha_j,a)$   are  the values in~\eqref{eq:lmda_a} and~\eqref{eq:mu_alfa} respectively.	
	\end{proposition}
	 The previous result is nothing but the lower bound in Theorem~\ref{thm:lambda_lin}, established  under the weaker Assumption~\ref{assump1}.
	 
	 In the non-linear Agmon estimates (see~Theorem~\ref{thm:agmon}), we need the localization zone to have the right surface scale, namely $\{\dist(x,S)< R_0 \mathfrak b^{-1/2}\}$ (for $\mathfrak b=\kappa H$). For this purpose, it is more convenient to choose the parameters in the above lower bound study as follows:
	 \emph{$\rho=\delta=1/2$, $\eta=\mathfrak b^{-1/2}$,  and $R_0$ large},
	 even though the lower bound estimate may appear weaker. With this choice of parameters,  Proposition~\ref{prop:Ub} becomes:
	 \begin{proposition}\label{prop:Ub2}
	 There exists $C>0$, and for all $R_0>1$ there exists $\mathfrak b_0>0$ such that for $\mathfrak b\geq \mathfrak b_0$ and  $u \in \dom Q_{\mathfrak b,\Fb}$, it holds
	 	\[Q_{\mathfrak b,\Fb}(u)\geq \int_{\Om} \Big(U^{(2)}_{\mathfrak b}(x)-C\frac{\mathfrak b}{R_0^2}\Big)|u(x)|^2\,dx,\]
	 	where
	 	\begin{equation*}
	 	U^{(2)}_{\mathfrak b}(x)=
	 	\begin{cases}
	 	|a|\mathfrak b&\qquad \dist(x,\partial \Om \cup\Gamma)\geq R_0\mathfrak b^{-\frac 12},\\
	 	\beta_a\mathfrak b-CR_0^4\mathfrak b^{\frac 12}&\qquad \dist(x,\partial \Om)\geq R_0\mathfrak b^{-\frac 12}\, \&\,\dist(x,\Gamma)< R_0\mathfrak b^{-\frac 12},\\
	 	|a|\Theta_0\mathfrak b-CR_0^4\mathfrak b^{\frac 12}&\qquad \dist(x,\partial \Om)< R_0\mathfrak b^{-\frac 12}\, \&\, x \notin \bigcup\limits_{j=1}^{n}  B(\mathsf p_j,R_0\mathfrak b^{-\frac 12}),\\
	 	\mu(\alpha_j,a)\mathfrak b-CR_0^4\mathfrak b^{\frac 12}&\qquad j \in\{1,...,n\}\,,\ x\in  B(\mathsf p_j,R_0\mathfrak b^{-\frac 12}).
	 	\end{cases}
	 	\end{equation*}
	 	Here $\mu(\alpha_j,a)$  and $\Theta_0$ are  the values in~\eqref{eq:mu_alfa} and~\eqref{eq:theta01} respectively.
	 \end{proposition}
	 
	\subsubsection{Upper bound of $\lambda(\mathfrak b)$}\label{sec:up}
	In the next proposition, we establish the upper bound in Theorem~\ref{thm:lambda_lin}.
	\begin{proposition}\label{prop:Upper3}
	 Under Assumption~\ref{assump3}, there exist $\mathfrak b_0, C>0$ such that for all $\mathfrak b\geq \mathfrak b_0$,
		\[\lambda(\mathfrak b) \leq \min_{j\in \{1,...,n\}}\mu(\alpha_j,a)\mathfrak b+C\mathfrak b^{\frac 35},\]
		where $\lambda(\mathfrak b)$ is the value in~\eqref{eq:lmda_a}.
	\end{proposition}
	\begin{proof}
		Let $k \in \{1,...,n\}$ be such that $\mu(\alpha_k,a)=\min_{j \in\{1,...,n\}}\mu(\alpha_j,a)$, and let $\mathsf p_k$ be the corresponding intersection point of $\Gamma$ and $\partial \Omega$ (see Notation~\ref{not:alfa}). We will establish the desired upper bound by defining a suitable test function, localized in a neighbourhood of $\mathsf p_k$. To this end, we consider a smooth cut-off function, $\chi$, satisfying 
		\begin{equation*}0\leq \chi\leq 1\ \mathrm{in}\ \R^2,\quad \chi=1\ \mathrm{in} \ B(0,1/2) \quad \mathrm{and}\quad \supp \chi\subset B(0,1).\end{equation*}
		Let $\mathfrak b>0$ be sufficiently large such that
		\begin{equation}\label{eq:b_rj1}
	 B(0,\mathfrak b^{- 2/5})\subset \Psi\big(B(\mathsf p_k,r_k)\big),
		\end{equation}
		where $r_k$ is the radius introduced in Section~\ref{sec:Psi}.
		We define the function $\hat{\chi}$ in $\R^2$ by
		$\hat{\chi}(\hat x)=\chi\big(\mathfrak b^\frac 25 \hat x\big)$.
		Consequently,
		\begin{equation}\label{eq:chi_hat}
		0\leq \hat\chi\leq 1\ \mathrm{in}\ \R^2, \ \hat\chi=1\ \mathrm{in} \ B\big(0,1/2\mathfrak b^{- 2/5}\big),\ \supp \hat{\chi} \subset B(0,\mathfrak b^{-2/5}),\ \mathrm{and}\ |\nabla_{\hat x} \hat\chi|\leq C\mathfrak b^{ 2/5}.\end{equation}
		We define the following test function in $\Om$:
		\begin{equation}\label{eq:u}
		u(x)=
		\begin{cases}
		\hat{u}\circ \Psi(x)&\mathrm{if~} x \in \Psi^{-1}\big(B(0,\mathfrak b^{- 2/5})\big) \cap \Om,\\
		0 &\mathrm{otherwise},
		\end{cases}
		\end{equation}
		where $\Psi$ is the diffeomorphism in Section~\ref{sec:Psi},
		\begin{equation*}
		\hat{u}(\hat{x})=
		\begin{cases}
		\hat{\chi}(\hat{x})u_0(\hat{x})e^{i\mathfrak b\varphi(\hat{x})}&\mathrm{if~} \hat x \in B(0,\mathfrak b^{- 2/5})\cap \R^2_+,\\
		0 &\mathrm{otherwise},
		\end{cases}
		\end{equation*}
		and $u_0(\hat{x})=\sqrt{\mathfrak b} v_0(\sqrt{\mathfrak b}\hat{x})$, for all $\hat{x} \in  \R^2_+$.
		Here $v_0$ is a normalized eigenfunction corresponding to $\mu(\alpha_k,a)$ (see~Remark~\ref{rem:v_0}), and $\varphi=\varphi_\ell$ is the gauge function in Lemma~\ref{lem:gauge1}, for $\ell=\mathfrak b^{- 2/5}$ ($\mathfrak b$ satisfies~\eqref{eq:b_rj1}). 
		 We will prove that 
		\begin{equation}\label{eq:up1}
		\frac{Q_{\mathfrak b, \Fb}(u)}{\|u\|_{L^2(\Omega)}^2 }\leq \mathfrak b \mu(\alpha_k,a)+C\mathfrak b^{\frac 35} .
		\end{equation}
		
		\emph{Upper bound of $Q_{\mathfrak b, \Fb}(u)$.} We establish the upper bound in several steps.
		\paragraph{\itshape Step 1.} (Change of variables). We use the properties of $\Psi$ in Section~\ref{sec:Psi} to get
		\begin{equation}\label{eq:up2}
		Q_{\mathfrak b, \Fb}(u) \leq (1+C\mathfrak b^{-\frac 25})\int_{B(0,\mathfrak b^{- 2/5})\cap \R^2_+}\big|(\nabla-i\mathfrak b\hat{\Fb})\hat{u}\big|^2\,dx,
		\end{equation}
		for some $C>0$ (see~\eqref{eq:Psi3},~\eqref{eq:Q_trans} and~\eqref{eq:Bon1}).
		\paragraph{\itshape Step 2.} (Change of gauge). We use the change of gauge in Lemma~\ref{lem:gauge1} to write
		\begin{equation}\label{eq:up3}
		\int_{B(0,\mathfrak b^{- 2/5})\cap \R^2_+}\big|(\nabla-i\mathfrak b\hat{\Fb})\hat{u}\big|^2\,dx = \int_{B(0,\mathfrak b^{- 2/5})\cap \R^2_+}\big|(\nabla-i\mathfrak b\hat{\Fb}_{\rm g})\hat{\chi}u_0\big|^2\,dx,
		\end{equation}
		where $\hat{\Fb}_{\rm g}$ is the vector potential in Lemma~\ref{lem:gauge1}.
	\paragraph{\itshape Step 3.} (Link to $\mu(\alpha_k,a)$). By Lemma~\ref{lem:gauge1}, we have 
		\begin{multline*}
		\int_{B(0,\mathfrak b^{- 2/5})\cap \R^2_+}\big|(\nabla-i\mathfrak b\hat{\Fb}_{\rm g})\hat{\chi}u_0\big|^2\,dx\\ \leq\int_{B(0,\mathfrak b^{- 2/5})\cap \R^2_+}\Big( \big|\partial _{\hat{x}_1}(\hat{\chi}u_0) \big|^2 + \Big|\Big(\partial _{\hat{x}_2}-i\mathfrak b(A_{\alpha_k,a}+f) \Big)\hat{\chi}u_0\Big|^2\Big)\,d\hat{x}.
		\end{multline*}
		Recall that the function  $f$ satisfies $|f(\hat{x}_1,\hat{x}_2)|\leq C (\hat{x}^2_1+|\hat x_1\hat x_2|)$, for some $C>0$.
	   Let $y=\sqrt{\mathfrak b}\hat{x}$, for $\hat{x} \in \R^2_+$. We define the function $\tilde{\chi}$ in $\R^2_+$ such that
		\[
		\tilde{\chi}(y)=\chi(\mathfrak b^{-\frac 1{10}}y)=\chi(\mathfrak b^\frac 25\hat{x})=\hat{\chi}(\hat{x}).\]
		Note that
		\begin{equation}\label{eq:up4} 
		0\leq \tilde\chi\leq 1\ \mathrm{in}\ \R^2, \ \tilde\chi=1\ \mathrm{in} \ B\big(0,1/2\mathfrak b^{1/10}\big),\ \supp \tilde{\chi} \subset B(0,\mathfrak b^{1/10}),\ |\nabla_y \tilde\chi|\leq C\mathfrak b^{-1/10},
		\end{equation}
		and $(\hat{\chi}u_0)(\hat x)=\sqrt{\mathfrak b}(v_0\tilde{\chi})(y)$. Hence, a simple computation yields that
		\begin{multline}\label{eq:up5}
		\int_{B(0,\mathfrak b^{- 2/5})\cap \R^2_+}\big|(\nabla-i\mathfrak b\hat{\Fb}_{\rm g})\hat{\chi}u_0\big|^2\,dx\\=\mathfrak b\int_{\R^2_+}\Big( \big|\partial _{y_1}(\tilde{\chi}v_0) \big|^2+ \Big|\Big(\partial _{y_2}-i\big(A_{\alpha_k,a}(y_1,y_2)+\mathfrak b^{-\frac 12}\mathcal O(y^2_1)+\mathfrak b^{-\frac 12}\mathcal O(y_1y_2)\big) \Big)\tilde{\chi}v_0\Big|^2\Big)\,dy. 
		\end{multline}	
		Below, we estimate each term of the right hand side of~\eqref{eq:up5} apart.
		We start by estimating the term  $\int \big|\partial _{y_1}(\tilde{\chi}v_0) \big|^2\,dy$. We use Cauchy's inequality and~\eqref{eq:up4}  to get
		\begin{align}
		\int_{\R^2_+}\big|\partial_{y_1}(\tilde{\chi}v_0) \big|^2\,dy &\leq (1+\mathfrak b^{-\frac 12}) 
		\int_{\R^2_+} \big|\tilde{\chi}\partial_{y_1}v_0 \big|^2\,dy
		+C\mathfrak b^\frac 12\int_{\R^2_+} \big|v_0\partial_{y_1}\tilde{\chi} \big|^2\,dy \nonumber\\
		& \leq (1+\mathfrak b^{-\frac 12}) 
		\int_{\R^2_+} \big|\partial_{y_1}v_0 \big|^2\,dy \nonumber\\
		&\qquad+C\mathfrak b^{\frac 3{10}}\int_{\big(B(0,\mathfrak b^{\frac 1{10}})\setminus B(0,\frac 12\mathfrak b^{\frac1{10}})\big) \cap \R^2_+} \big|v_0\big|^2\,dy. \label{eq:up6}
		\end{align}	
		To control the error term in~\eqref{eq:up6}, we use the following result derived from the decay of the eigenfunction $v_0$ established in Theorem~\ref{thm:v0-decay} (taking $\delta=\big(|a|\Theta_0-\mu(\alpha_k,a)\big)/2$ in the aforementioned theorem): 
		\begin{equation}\label{eq:up_expo}
		\int_{\big(B(0,\mathfrak b^{\frac 1{10}})\setminus B(0,\frac 12\mathfrak b^{\frac 1{10}})\big) \cap \R^2_+} |v_0|^2\,dy \leq 
		e^{-C_2\mathfrak b^{\frac 1{10}}}\int_{\R^2_+}e^{2\phi}|v_0|^2\,dy  
		\leq C_1 e^{-C_2\mathfrak b^{\frac 1{10}}}. 
		\end{equation}
		Here $C_1=C_{\delta,\alpha_k}$, $C_2=\sqrt{\big(|a|\Theta_0-\mu(\alpha_k,a)\big)/2}$, and $\phi$ is the function  introduced in Theorem~\ref{thm:v0-decay}. 
		Plugging~\eqref{eq:up_expo} in~\eqref{eq:up6}, we get for large values of $\mathfrak b$, and for some positive constants $\tilde {C_1}$ and $\tilde{C_2}$
		\begin{equation}\label{eq:up7}
		\int_{\R^2_+}\big|\partial_{y_1}(\tilde{\chi}v_0) \big|^2\,dy\leq(1+\mathfrak b^{-\frac 12})  \int_{\R^2_+} \big|\partial_{y_1}v_0 \big|^2\,dy+\tilde{C_1} e^{-\tilde{C_2}\mathfrak b^{\frac 1{10}}}.
	    \end{equation}
		Now we estimate the second term in the right hand side of~\eqref{eq:up5}:
			\begin{align}
			&\int_{\R^2_+}\Big|\Big(\partial _{y_2}-i\big(A_{\alpha_k,a}+\mathfrak b^{-\frac 12}\mathcal O(y^2_1)+\mathfrak b^{-\frac 12}\mathcal O(y_1y_2)\big) \Big)\tilde{\chi}v_0\Big|^2\,dy\nonumber\\&\leq (1+\mathfrak b^{-\frac 12})\int_{\R^2_+}\big|(\partial_{y_2}-iA_{\alpha_k,a}) v_0\big|^2\,dy
		+C\mathfrak b^{-\frac 12}\int_{\R^2_+} y_1^4|v_0|^2\,dy\nonumber\\
		&\qquad+C\mathfrak b^{-\frac 12}\int_{\R^2_+} y_1^2y_2^2|v_0|^2\,dy
		+C\mathfrak b^\frac 12 \int_{\R^2_+}|\partial_{y_2}\tilde{\chi}|^2|v_0|^2\,dy. \label{eq:up8}
			\end{align}
		In~\eqref{eq:up8}, we used Cauchy's inequality together with the properties of $\tilde{\chi}$ in~\eqref{eq:up4}. 
	In a similar fashion of establishing~\eqref{eq:up_expo}, we use~\eqref{eq:up4} together with the exponential decay in Theorem~\ref{thm:v0-decay} to estimate
		\begin{equation}\label{eq:expo2}
		\mathfrak b^\frac 12 \int_{\big(B(0,\mathfrak b^{\frac 1{10}})\setminus B(0,\frac 12\mathfrak b^{\frac 1{10}})\big) \cap \R^2_+}|\partial_{y_2}\tilde{\chi}|^2|v_0|^2\,dy\leq \tilde{C_1} e^{-\tilde{C_2}\mathfrak b^{\frac 1{10}}}.
			\end{equation}
		Moreover, the aforementioned exponential decay shows that $y\mapsto y_1^2v_0(y)$ and  $y\mapsto y_1y_2v_0(y)$ are square integrable in $\R^2_+$, that is there exists $C>0$ such that
		\begin{equation}\label{eq:expo3}
		\int_{\R^2_+} y_1^4|v_0|^2\,dy\leq C \quad {\rm and}\ 	\int_{\R^2_+} y_1^2y_2^2|v_0|^2\,dy\leq C.
		\end{equation}
		From~\eqref{eq:up8}--\eqref{eq:expo3}, we get
		\begin{multline}\label{eq:up9}
			\int_{\R^2_+}\Big|\Big(\partial _{y_2}-i\big(A_{\alpha_k,a}+\mathfrak b^{-\frac 12}\mathcal O(y^2_1)+\mathfrak b^{-\frac 12}\mathcal O(y_1y_2)\big) \Big)\tilde{\chi}v_0\Big|^2\,dy\\\leq(1+\mathfrak b^{-\frac 12})  \int_{\R^2_+}\big|(\partial_{y_2}-iA_{\alpha_k,a}) v_0\big|^2\,dy
	+C\mathfrak b^{-\frac 12}.
		\end{multline}
		Since $v_0$ is a normalized eigenfunction of the operator $\mathcal H_{\alpha_k,a}$ (in~\eqref{eq:P_alfa}), corresponding to $\mu(\alpha_k,a)$, we have
		\begin{equation}\label{eq:upp}
		\int_{\R^2_+}\Big(\big|\partial_{y_1}v_0 \big|^2+\big|(\partial_{y_2}-iA_{\alpha_k,a}) v_0\big|^2\Big)\,dy=q_{\alpha_k,a}(v_0)=\mu(\alpha_k,a).\end{equation}
	Gathering pieces in~\eqref{eq:up5},~\eqref{eq:up7},~\eqref{eq:up9}, and~\eqref{eq:upp} implies
		\begin{align}
		\int_{B(0,\mathfrak b^{- 2/5})\cap \R^2_+}\big|(\nabla-i\mathfrak b\hat{\Fb}_{\rm g})\hat{\chi}u_0\big|^2\,dx
		& \leq (1+\mathfrak b^{-\frac 38})\mathfrak b q_{\alpha_k,a}(v_0)+C\mathfrak b^{\frac 12} \nonumber\\
		&\leq \mathfrak b \mu(\alpha_k,a)+C\mathfrak b^{\frac 12}. \label{eq:up10}
		\end{align}
		
		Finally, the estimates established  in~\eqref{eq:up2},~\eqref{eq:up3}, and~\eqref{eq:up10} yield
		\begin{equation}\label{eq:up11}
		Q_{\mathfrak b,\Fb}(u)\leq (1+C\mathfrak b^{-\frac 25})\Big( \mathfrak b \mu(\alpha_k,a)+C\mathfrak b^\frac 12\Big) 
		\leq  \mathfrak b \mu(\alpha_k,a)+C\mathfrak b^\frac 35  .
		\end{equation}
		
		\emph{Lower bound of $\|u\|^2_{L^2(\Omega)}$.} The definition of $u$ in~\eqref{eq:u} and the property in~\eqref{eq:Psi3} yield
		\begin{align}
		\int_{\Om}|u|^2\,dx&\geq (1-C\mathfrak b^{-\frac 25})\int_{B(0,\mathfrak b^{- \frac 25})\cap\R^2_+} |\hat u|^2\,d\hat x\nonumber\\
		&= (1-C\mathfrak b^{-\frac 25})\int_{B(0,\mathfrak b^{- \frac 25})\cap \R^2_+} |\hat\chi u_0|^2\,d\hat x\nonumber\\
		&= (1-C\mathfrak b^{-\frac 25})\int_{B(0,\mathfrak b^{\frac 1{10}})\cap \R^2_+} |\tilde{\chi} v_0|^2\,dy\nonumber\\
		&\geq (1-C\mathfrak b^{-\frac 25})\int_{B(0,\frac 12\mathfrak b^{\frac 1{10}})\cap \R^2_+} |v_0|^2\,dy\nonumber\\
		&=(1-C\mathfrak b^{-\frac 25})\Big(1-\int_{B(0,\frac 12\mathfrak b^{\frac 1{10}})^\complement \cap \R^2_+} |v_0|^2\,dy\Big).
		\end{align}
		Similarly to~\eqref{eq:up_expo}, we have
		\[\int_{B(0,\frac 12\mathfrak b^{\frac 1{10}})^\complement \cap \R^2_+} |v_0|^2\,dy\leq C_1 e^{-C_2\mathfrak b^{\frac 1{10}}}.\]
		Hence, 
		\begin{equation}\label{eq:up12}
		\int_{\Om}|u|^2\,dx\geq 1-C\mathfrak b^{-\frac 25}.
		\end{equation}
		We gather the results in~\eqref{eq:up11} and~\eqref{eq:up12} to establish the claim in~\eqref{eq:up1}. Consequently the min-max principle completes the proof of Proposition~\ref{prop:Upper3}.
	\end{proof}
	\begin{rem}\label{rem:up_opt}
	The error established in Proposition~\ref{prop:Upper3} is not optimal. More generally, for any $\rho \in (0,1/2)$, one may set $B(0,\mathfrak b^{-\rho})$ to be the support of $\hat\chi$ in~\eqref{eq:chi_hat}. Then, by adjusting the choice of the parameters in  the upper bound proof, one can get
	\[\lambda(\mathfrak b) \leq \min_{j\in \{1,...,n\}}\mu(\alpha_j,a)\mathfrak b+C\mathfrak b^{1-\rho},\]
	for all $\mathfrak b\geq \mathfrak b_0$.
	\end{rem}
	
	\section{Breakdown of superconductivity}\label{sec:giorgi}
	Below, we prove that when the magnetic field is sufficiently large, the only solution of~\eqref{eq:Euler} is the normal state $(0,\Fb)$, where $\Fb \in \Hd$ is the vector potential in~\eqref{A_1} (see~Theorem~\ref{thm:priori}).
	\subsection{A priori estimates}\label{sec:a-p-est}
	We present certain known estimates needed in the sequel to control the  errors arising in our various approximations.
	\begin{proposition}\label{prop:psi}
		If $(\psi,\Ab) \in H^1(\Om;\C)\times H^1(\Om;\R^2)$ is a weak solution of~\eqref{eq:Euler}, then
		\[\|\psi\|_{L^\infty(\Omega)}\leq 1.\]
	\end{proposition}
	We omit the proof of Proposition~\ref{prop:psi}, and we refer to the similar proof
	in~\cite[Proposition~10.3.1]{fournais2010spectral}. 
	
	Recall the magnetic field $B_0={\mathbbm 1}_{\Om_1}+a{\mathbbm 1}_{\Om_2}$ with $a\in[-1,1)\setminus\{0\}$, introduced in Assumption~\ref{assump1}.  There exists a  unique vector potential $\Fb
	\in \Hd$ such that  (see~\cite[Lemma A.1]{Assaad})
	\begin{equation}\label{A_1}
	\curl\,\Fb=B_0.\end{equation}	

	\begin{thm}\label{thm:priori}
		Let  $\beta\in(0,1)$. Suppose that the
		conditions in Assumption~\ref{assump1} hold.
		There exists $C>0$ such that for all $\kp>0$, if $(\psi,\Ab)\in H^1(\Om;\C)\times \Hd$ is a solution of~\eqref{eq:Euler}, then  
		\begin{enumerate}
			\item $\|(\nb-i \kp H\Ab)\psi\|_{L^2(\Om)}\leq \kp\|\psi\|_{L^2(\Om)}$.
			\item $\displaystyle\|{\curl(\Ab-\Fb)}\|_{L^2(\Om)}\leq \frac C
			H\|\psi\|_{L^2(\Om)}$.
			\item $\Ab-\Fb\in H^2(\Omega)$ and $\displaystyle\|\Ab-\Fb\|_{H^2(\Om)}\leq \frac C H\|\psi\|_{L^2(\Om)}$.
			\item $\Ab-\Fb\in \mathcal C^{0,\beta}(\overline{\Om})$ and $\displaystyle\|\Ab-\Fb\|_{\mathcal C^{0,\beta}(\overline \Om)}\leq \frac C H \|\psi\|_{L^2(\Om)}$.
		\end{enumerate}
	\end{thm}
The proof of the previous theorem is given in~\cite[Lemma~10.3.2]{fournais2010spectral} and~\cite[Theorem~4.2]{Assaad}.
	\subsection{Trivial minimizers}
We adapt a result of Giorgi--Phillips~\cite{giorgi2002breakdown} to our case of the step magnetic field $B_0$. Let ${\mathbf
		F} \in \Hd$ be the magnetic potential in~\eqref{A_1}, satisfying $\curl\,\Fb=B_0$. Observe that $(0,\Fb)$ is a  critical point of the functional in~\eqref{eq:GL}, i.e. it is a weak solution of~\eqref{eq:Euler}.  In Theorem~\ref{thm:giorgi} below, we show that this solution is the \emph{unique  minimizer of the functional in~\eqref{eq:GL}}, for sufficiently large values of $H$.
	\begin{theorem}\label{thm:giorgi}
	Under Assumption~\ref{assump1}, there exist positive constants $\kappa_1$ and $C_1$ such that if $\kappa\geq \kappa_1$, 
		\[H> C_1\kappa,\]
		then $(0,\Fb)$ is the unique solution of~\eqref{eq:Euler} in $H^1(\Om)\times \Hd$.
	\end{theorem}
	\begin{proof}
		Let $\kappa>0$ and $H>0$.
		 Assume that the corresponding GL system~\eqref{eq:Euler} admits a non-trivial solution $(\psi,\Ab) \in H^1(\Om)\times \Hd$. We mean by \emph{non-trivial}  that
		\begin{equation}\label{eq:non_triv}
		\|\psi\|_{L^2(\Om)}>0.
		\end{equation}
	 We compare $\|(\nb-i \kp H\Fb)\psi\|_{L^2(\Om)}$ and $\|(\nb-i \kp H\Ab)\psi\|_{L^2(\Om)}$ using  Cauchy's inequality 
		\begin{equation}\label{eq:gio2}
		\|(\nb-i \kp H\Fb)\psi\|^2_{L^2(\Om)} \leq 2\|(\nb-i \kp H\Ab)\psi\|^2_{L^2(\Om)}+2(\kappa H)^2\|(\Ab-\Fb)\psi\|^2_{L^2(\Om)}.
		\end{equation}
			The estimates in Theorem~\ref{thm:priori} ensure that
			\begin{equation}\label{eq:gio1}
			\|(\nb-i \kp H\Ab)\psi\|^2_{L^2(\Om)}+(\kappa H)^2\|\Ab-\Fb\|^2_{L^2(\Om)} \leq C\kappa^2\|\psi\|^2_{L^2(\Om)}.
			\end{equation}
			This inequality, together with $|\psi|\leq 1$, allow us to control the right hand side of~\eqref{eq:gio2} and get
		\[\|(\nb-i \kp H\Fb)\psi\|^2_{L^2(\Om)} \leq C\kappa^2\|\psi\|^2_{L^2(\Om)}.\]
	Since $(\psi,\Ab)$ is non-trivial, we get 
		\begin{equation}\label{eq:gio3}
		\lambda(\kappa H)\leq C\kappa^2,
		\end{equation}
		where $\lambda(\kappa H)$ is the value in~\eqref{eq:lmda_a} . 
		
		On the other hand,  let $\kappa_0$ be such that $\kappa_0\geq\mathfrak b_0$, where $\mathfrak b_0$ is the constant in Proposition~\ref{prop:Lower3}. Applying this Proposition, we get the existence of $\tilde{C}>0$ such that for all $\kappa\geq \kappa_0$ and $H\geq 1$,
		\begin{equation}\label{eq:gio4}
		\lambda(\kappa H) \geq \tilde{C}\min\big(|a|\Theta_0,\min_j\mu(\alpha_j,a)\big)\kappa H.
		\end{equation}
		We combine~\eqref{eq:gio3} and~\eqref{eq:gio4} to obtain the following: for all $\kappa\geq \kappa_0$ and $H\geq 1$, if 
		the corresponding GL system~\eqref{eq:Euler} admits a non-trivial solution, then 
		 \[\tilde{C}\min\big(|a|\Theta_0,\min_j\mu(\alpha_j,a)\big)\kappa H \leq\lambda(\kappa H)\leq C\kappa^2,\]
		which in this case implies that
		\[H \leq C_1\kappa,\]
		for $C_1=C/\big(\tilde{C}\min\big(|a|\Theta_0,\min_j\mu(\alpha_j,a)\big)\big)$. 
		This result can be reformulated as follows:
		 For all $\kappa\geq\kappa_0$, if $H>\max(C_1\kappa,1)$ then $\mathcal E_{\kappa,H}
$ admits only  trivial minimizers.	Take $\kappa_1\geq \max(\kappa_0, 1/C_1)$ so that for all $\kappa \geq \kappa_1$, $C_1\kappa\geq 1$. We have then proved Theorem~\ref{thm:giorgi}.	  
	\end{proof}
	\section{Monotonicity of $\lambda(\mathfrak b)$}\label{sec:monot}
	We consider $\lambda(\mathfrak b)$---the lowest eigenvalue of the operator $\mathcal P_{\mathfrak b,\Fb}$ defined in Section~\ref{sec:intro2}. We will establish the so-called strong diamagnetic property (\cite{fournais2007strong}); $\mathfrak b\mapsto\lambda(\mathfrak b)$ is strictly increasing for large values of $\mathfrak b$ (Proposition~\ref{prop:monot}). This property will  enable us to prove the first statement of Theorem~\ref{thm:Hc3} (Proposition~\ref{prop:H_unique}). Moreover, we  will provide the asymptotics of $H_{C_3}(\kappa)$ stated in Theorem~\ref{thm:Hc3} (Proposition~\ref{prop:Hc3}).
	
Information about the localization of a ground-state of  $\mathcal P_{\mathfrak b,\Fb}$ is needed while establishing the monotonicity result in Proposition~\ref{prop:monot}. Theorem~\ref{thm:lin_bound} below provides such  localization (Agmon) estimates. Our argument is quite similar to that in~\cite[Section~15]{bonnaillie2003analyse}. Still, we give the proof of this theorem for completeness.
	
Recall the set $\Gamma \cap \partial \Om=\big\{ \mathsf p_j~:~j \in\{1,...,n\}\big\}$.  In this section, we assume that Assumption~\ref{assump3} holds. We denote by
\begin{equation}\mu^*=\min_{j\in\{1,...,n\}}\mu(\alpha_j,a).\end{equation}
Let $S^*$ be the set of points  $\mathsf p_k$ corresponding to  the minimal energy $\mu(\alpha_k,a)$
	\begin{equation}\label{eq:N*}
	S^*=\big\{\mathsf p_k \in \Gamma \cap \partial \Om~:~\mu(\alpha_k,a)=\mu^*\big\},\end{equation}
	
	As shown in the next theorem, a ground-state is localized near the points of $S^*$.
	\begin{theorem}\label{thm:lin_bound}
		Under Assumption~\ref{assump3}, there exist positive constants $\mathfrak b_0$, $C$, and $\zeta$ such that if $\mathfrak b\geq \mathfrak b_0$ and $\psi$ is a ground-state of the operator $\mathcal P_{\mathfrak b,\Fb}$ then
		\begin{equation}\label{eq:agm0}
		\int_\Om e^{2\zeta\sqrt \mathfrak b\dist(x,S^*)}\left(|\psi|^2+\mathfrak b^{-1} \big|\big(\nabla-i\mathfrak b\Fb\big)\psi\big|^2\right)\,dx \leq C\|\psi\|^2_{L^2(\Om)}.\end{equation}
		Consequently,  for all $N>0$,
		\[\int_\Om \dist(x,S^*)^N|\psi|^2\,dx=\mathcal O\big(\mathfrak b^{-\frac N2}\big).\]	
	\end{theorem}
	\begin{proof}
	Let $R_0>1$. We define the real Lipschitz function 
	\begin{equation}\label{eq:g}
	g(x)=\zeta \max\big(\dist(x,S^*),R_0\mathfrak b^{-\frac 12}\big),\qquad x\in\Om,\end{equation}
	where $\zeta>0$ is to be chosen later. An integration by parts yields
	\begin{equation}\label{eq:agm1}
	\re\big\langle  \mathcal P_{\mathfrak b,\Fb}\psi, e^{2\sqrt \mathfrak b g} \psi\big\rangle=Q_{\mathfrak b,\Fb}\big(e^{\sqrt \mathfrak b g} \psi\big)-\mathfrak b\big\||\nabla g|e^{\sqrt \mathfrak b g} \psi\big\|^2_{L^2(\Om)},\end{equation}
	where $Q_{\mathfrak b,\Fb}$ is the quadratic form in~\eqref{eq:Quad}.  Hence, using~\eqref{eq:agm1}  and the definition of $\psi$, we get
	\begin{equation}\label{eq:agm2}
	\lambda(\mathfrak b)\|e^{\sqrt \mathfrak b g} \psi\|^2=Q_{\mathfrak b,\Fb}\big(e^{\sqrt \mathfrak b g} \psi\big)-\mathfrak b\big\||\nabla g|e^{\sqrt \mathfrak b g} \psi\big\|^2.
	\end{equation}
	By Propositions~\ref{prop:Ub2} and~\ref{prop:Upper3}, we have 
	\begin{equation}\label{eq:agm3}
	Q_{\mathfrak b,\Fb}\big(e^{\sqrt \mathfrak b g} \psi\big)\geq \int_{\Om} \big(U^{(2)}_{\mathfrak b}(x)-C\mathfrak b R_0^{-2}\big)\big|e^{\sqrt \mathfrak b g(x)} \psi(x)\big|^2\,dx,
	\end{equation}
	and
	\begin{equation}\label{eq:agm4}
	\lambda(\mathfrak b)\leq \mu^*\mathfrak b+o(\mathfrak b).
	\end{equation}
	Implementing~\eqref{eq:agm3} and~\eqref{eq:agm4} in~\eqref{eq:agm2}, dividing by $\mathfrak b$ and using the properties of the function $U^{(2)}_{\mathfrak b}$ in Proposition~\ref{prop:Ub2} yield
	\begin{multline}\label{eq:agm5}
	\int_{\big\{t(x)\geq R_0\mathfrak b^{-\frac 12}\big\}} \big(\mu^{**}-CR_0^4\mathfrak b^{-\frac 12}-CR_0^{-2}\big)\big|e^{\sqrt \mathfrak b g} \psi\big|^2\,dx\\ + \int_{\big\{t(x)\leq R_0\mathfrak b^{-\frac 12}\big\}} \big(\mu^*-CR_0^4\mathfrak b^{-\frac 12}-CR_0^{-2}\big)\big|e^{\sqrt \mathfrak b g} \psi\big|^2\,dx\\ \leq
	\big(\mu^*+o(1)\big)\|e^{\sqrt \mathfrak b g} \psi\|^2+\big\||\nabla g|e^{\sqrt \mathfrak b g} \psi\big\|^2.
	\end{multline}
	Here $t(x)=\dist(x,S^*)$, and
	$\mu^{**}$ is the minimum of all the $\mu(\alpha_j,a)$ that are strictly greater that $\mu^*$ (if such a $\mu(\alpha_j,a)$ does not exist, we take $\mu^{**}=|a|\Theta_0$). 
	By~\eqref{eq:g}, we have $\supp (\nabla g) \subset \big\{t(x)\geq R_0\mathfrak b^{-\frac 12}\big\}$ and $|\nabla g|\leq \zeta$. Consequently, 
	\begin{equation}\label{eq:agm6}
	\big\||\nabla g|e^{\sqrt \mathfrak b g} \psi\big\|^2\leq \zeta^2\int_{\big\{t(x)\geq R_0\mathfrak b^{-\frac 12}\big\}}e^{2\sqrt \mathfrak b g} |\psi|^2\,dx.
	\end{equation}
	Hence,~\eqref{eq:agm5} yields
	\begin{multline}\label{eq:agm7}
		\int_{\big\{t(x)\geq R_0\mathfrak b^{-\frac 12}\big\}} \big(\mu^{**}-\mu^*-o(1)-CR_0^4\mathfrak b^{-\frac 12}-CR_0^{-2}-\zeta^2\big)\big|e^{\sqrt \mathfrak b g} \psi\big|^2\,dx
		\\ \leq \int_{\big\{t(x)\leq R_0\mathfrak b^{-\frac 12}\big\}} \big(CR_0^4\mathfrak b^{-\frac 12}+CR_0^{-2}+o(1)\big)\big|e^{\sqrt \mathfrak b g} \psi\big|^2\,dx .
	\end{multline}
 We may choose $\zeta<\sqrt{\mu^{**}-\mu^*}$. Then using~\eqref{eq:agm7} and the definition of $g$ in~\eqref{eq:g}, there exist large positive constants $R_0$ and $\mathfrak b_0$  such that for all $\mathfrak b\geq \mathfrak b_0$
	\begin{equation}\label{eq:agm8}
	\int_\Om e^{2\zeta\sqrt \mathfrak b\dist(x,\partial \Om)} |\psi|^2\,dx \leq \tilde C (R_0,\zeta)\|\psi\|^2_{L^2(\Om)}.
	\end{equation}
	
	One can deduce the other part of~\eqref{eq:agm0} by gathering the estimates in~\eqref{eq:agm2},~\eqref{eq:agm6}, \eqref{eq:agm8} and the upper bound in Proposition~\ref{prop:Upper3}.
	\end{proof}
\begin{rem}
	In similar situations in the literature, when the applied magnetic field is uniform, certain normal Agmon estimates were established showing the decay of the ground-state away from the boundary.  Such  decays were usually used in the proofs of the monotonicity of the ground-state energy (see~\cite[Section~2]{fournais2007strong}). In the present work, one can similarly establish such a normal decay of the ground-state away from the boundary of $\Omega_1\cup\Om_2$. However as it will be explained later in this section, the localization result in Theorem~\ref{thm:lin_bound} is sufficient  while deriving the monotonicity of the ground-state in our step magnetic field case. Therefore, we opt  not to state the normal estimates here.
\end{rem}
	Having the domain of $\mathcal P_{\mathfrak b,\Fb}$  independent of $\mathfrak b$ (see~\eqref{eq:domP_F}),  the existence of the left and right derivatives of $\lambda(\mathfrak b)$ is guaranteed by the analytic perturbation theory (see~\cite{Kato}):
	\[\lambda'_\pm(\mathfrak b)=\lim_{\epsilon \rightarrow 0^\pm}\frac {\lambda(\mathfrak b+\epsilon)-\lambda(\mathfrak b)}{\epsilon}.\] 

	\begin{proposition}\label{prop:monot}
		Under Assumption~\ref{assump3}, the limits of $\lambda'_-(\mathfrak b)$ and $\lambda'_+(\mathfrak b)$ as $\mathfrak b \rightarrow +\infty$ exist, and we have
		\[\lim_{\mathfrak b \rightarrow +\infty}\lambda'_+(\mathfrak b,a)=\lim_{\mathfrak b \rightarrow +\infty}\lambda'_-(\mathfrak b,a)=\mu^*,\]
		where $\mu^*=\min_{j \in\{1,...,n\}} \mu(\alpha_j,a)>0$.
		
		Consequently, $\mathfrak b \mapsto \lambda(\mathfrak b)$ is strictly increasing, for large $\mathfrak b$.
		\end{proposition}
	The proof of Proposition~\ref{prop:monot} is inspired by that of~\cite[Theorem 1.1]{fournais2007strong}, although the two proofs  differ slightly at the technical level in a way that we will describe below.
	
	The argument in~\cite{fournais2007strong} avoids the use of a complete expansion of the ground-state energy. Such expansions have been used in other works such as~\cite{fournais2006third,bonnaillie2007superconductivity}, and  are usually difficult to  establish. Fournais and Helffer  succeeded to prove the monotonicity of the ground-state  by only using  its leading order asymptotics. Their proof  mainly rely on the control of a certain (error) term, $\|\hat \Ab \psi\|_{L^2(\Om)}$, appearing in the differentiation of the energy, where $\psi$ is a ground-state of the linear operator and $\hat \Ab$ is a vector potential that we introduce below. In~\cite{fournais2007strong}, they use the fact that their vector potential, denoted by $\Fb$, generates  a  constant magnetic field ($\curl \Fb=1$). In the case where the sample is not a disc,  this implies the  existence of a part of the boundary (away from the points with maximal curvature) where $\psi$ is negligible. The remaining part, $\Om_0$, of the boundary (containing the points with maximal curvature) is a simply connected domain. Hence, a gauge transform is used to construct from the potential $\Fb$ another potential $\hat \Ab \in H^1(\Om,\R^2)$ such that $|\hat \Ab|\leq C\dist (x,\partial\Om)$ in $\Om_0$. This upper bound of $|\hat \Ab|$ compensates the fact that $\psi$ is \emph{big} in $\Om_0$, and, together with the normal and boundary Agmon estimates, allow to control $\|\hat \Ab \psi\|_{L^2(\Om)}$. 
	
	We adopt a parallel strategy where we use the leading order asymptotics of $\lambda(\mathfrak b)$ established in Theorem~\ref{thm:lambda_lin}. The intersection points, $\mathsf p_j$, of the magnetic edge $\Gamma$ and the boundary $\partial \Om$ play the role of the points with maximum curvature in~\cite{fournais2007strong}. However, the discontinuity of our magnetic field  makes us take into consideration the way $\Gamma$ intersects  $\partial \Om$, while constructing the gauge vector potential $\Fb_{\mathfrak g}$ (playing the role of  $\hat \Ab$ in~\cite{fournais2007strong}). This generates a more complicated definition of  $\Fb_{\mathfrak g}$ related to the geometry of the problem (Lemma~\ref{lem:gauge_field}). This definition guarantees that $\Fb_{\mathfrak g}$ is in $H^1(\Om,\R^2)$ and satisfies $|\Fb_{\mathfrak g}|\leq C\dist(x,\mathsf p_j)$ in the vicinity  of any  point $\mathsf p_j$. Consequently, the localization estimates in Theorem~\ref{thm:lin_bound} are sufficient to control the (error) term $\|\Fb_{\mathfrak g} \psi\|_{L^2(\Om)}$. Here, $\psi$ is a ground-state corresponding to the energy $\lambda(\mathfrak b)$.

	Now, we present the approach in details. It is convenient to work in the so-called Frenet coordinates.
For $t_0>0$, we define
		\[\breve{\Phi}~:~\frac {|\partial \Om|}{2\pi}\mathbb S^1\times(0,t_0) \ni (s,t)\longmapsto \gamma(s)+t\nu(s) \in \R^2 .\]
		where $\big(|\partial \Om|/2\pi\big)\mathbb S^1\ni s\mapsto \gamma(s)\in\partial \Om$ is the arc length parametrization of $\partial \Om$, oriented counterclockwise and $\nu(s)$ is the inward unit normal vector of $\partial \Om$ at the point $\gamma(s)$. We assume that $t_0$ is sufficiently small so that $\breve{\Phi}$ is a diffeomorphism, and we denote its image by  $\Om(t_0)$.
		
		 Notice that $t=\dist(\breve{\Phi}(s,t),\partial \Om)$. The Jacobian of $\breve{\Phi}$ satisfies $J_{\breve{\Phi}}=1-tk(s)$. Here, $k(s)$ is the curvature of $\partial \Om$ at the point $\gamma(s)$, which is bounded according to the assumptions on the domain. For more details about Frenet coordinates, see~\cite[Appendix F]{fournais2010spectral}.
		
	For $j \in\{1,...,n\}$, let $s_j$ be the abscissa of the point $\mathsf p_j\in\Gamma\cap\partial\Om$ in the Frenet coordinates, that is  $(\breve{\Phi})^{-1}(\mathsf p_j)=(s_j,0)$. We denote by $l=\min_{p,m}|s_p-s_m|$. For any positive $\epsilon$ such that $\epsilon<\min(t_0,l/2)$, we define the set:
		\begin{equation}\label{eq:N_j}
		\mathcal N(\mathsf p_j,\epsilon) =\big\{x=\breve{\Phi}(s,t)~:~0<t< \epsilon,|s-s_j|< \epsilon\big\}.\end{equation}
		When the above conditions on $\epsilon$ hold, we choose an $\epsilon_0>0$ to get  a family of pairwise disjoint sets $\left(\mathcal N(\mathsf p_j,2\epsilon_0)\right)_{j=1}^n$ of $\Om(t_0)$.

		\begin{lemma}\label{lem:gauge_field}
		Let $\Gamma \cap \partial \Om=\big\{ \mathsf p_j~:~j \in\{1,...,n\}\big\}$. There exist  $C>0$ and a function $\varphi\in H^2(\Om)$ such that $\Fb_{\mathfrak g}=\Fb+\nabla \varphi$ satisfies for any $j \in\{1,...,n\}$ 
		\begin{equation*}
			|\Fb_{\mathfrak g}(x)|\leq C\dist(x,\mathsf p_j),\  x\in \mathcal N(\mathsf p_j,\epsilon_0).
		\end{equation*}
		\end{lemma}
	\begin{proof}
		Let $\breve{\Fb}=(\breve{F_1},\breve{F_2})$ be the vector potential defined so that
		\[F_1 dx_1+F_2 dx_2=\breve{F_1}ds+\breve{F_2}dt.\]
		We have
		 \begin{equation*}
		\curl_{s,t}\breve{\Fb}=\partial_s\breve{F_2}-\partial_t\breve{F_1}=
         \begin{cases}
          1-tk(s),&\mathrm{if}~\breve{\Phi}(s,t)\in\Om_1,\\
          a(1-tk(s)),&\mathrm{if}~\breve{\Phi}(s,t)\in\Om_2.
          \end{cases}
		 \end{equation*}

		We fix a point $\mathsf p_j \in\Gamma\cap\partial\Om$ and we work locally in the set $\mathcal N(\mathsf p_j,2\epsilon_0)$. After performing a translation, we assume that the Frenet coordinates of $\mathsf p_j$ are $(0,0)$, but for simplicity we still denote by $\breve{\Phi}$ the obtained diffeomorphism $\breve{\Phi}^j$. Furthermore, let $\mathcal N_m=\mathcal N(\mathsf p_j,2\epsilon_0)\cap \Om_m$, $m=1,2$. To fix computation, we assume w.l.o.g that $\breve{\Phi}^{-1}(\partial\mathcal N_1\cap \partial \Om)$ (respectively  $\breve{\Phi}^{-1}(\partial\mathcal N_2\cap \partial \Om)$ is a subset of $\{(s,t)~:~s\geq 0,t=0\}$ (respectively $\{(s,t)~:~s\leq 0,t=0\}$). The curve $\Gamma\cap \mathcal N(\mathsf p_j,2\epsilon_0)$ is transformed to the curve $\breve{\Gamma}$  in the $(s,t)$-plane.  Under Assumption~\ref{assump1} (particularly Item 6) and due to the nature of the diffeomorphism $\breve{\Phi}$, the curve $\breve{\Gamma}$ is not tangent to the $s$-axis.
Hence for sufficiently small $\epsilon_0$, one may distinguish between three cases:
\begin{itemize}
	\item [\emph{Case 1.}] $\breve{\Gamma} \subset \{(s,t)~:~s>0\}$.
	\item [\emph{Case 2.}] $\breve{\Gamma} \subset \{(s,t)~:~s<0\}$.
	\item [\emph{Case 3.}] $\breve{\Gamma} \subset \{(s,t)~:~s=0\}$.
\end{itemize}
In each of the first two cases, we assume that $\epsilon_0$ is small enough so that the curve $\breve{\Gamma}$ corresponds to a strictly monotonous function $s\mapsto f(s)$.
 We consider the vector potential  $\breve{\Fb}_\mathfrak g^j=\big(0,\breve{F}_\mathfrak g^j\big)\in H^1\big(\breve{\Phi}^{-1}\big(\mathcal N(\mathsf p_j,2\epsilon_0)\big)\big)$, where $\breve{F}_\mathfrak g^j$ is defined in $\breve{\Phi}^{-1}\big(\mathcal N(\mathsf p_j,2\epsilon_0)\big)$, in each of the three cases above,  as follows:
	\paragraph{\itshape Case 1.} $\breve{F}_\mathfrak g^j(s,t)$ is given by
	\begin{equation*}
\begin{cases}
s+(a-1)f^{-1}(t)-at\int_0^{f^{-1}(t)}k(s')\,ds'-t\int_{f^{-1}(t)}^sk(s')\,ds',&~s>0\ \mathrm{and}\ s\geq f^{-1}(t), \\
as-at\int_0^{s}k(s')\,ds' ,&\mathrm{elsewhere}.
\end{cases}
	\end{equation*}
\paragraph{\itshape Case 2.} $\breve{F}_\mathfrak g^j(s,t)$ is given by
\begin{equation*}
\begin{cases}
as+(1-a)f^{-1}(t)-t\int_0^{f^{-1}(t)}k(s')\,ds'-at\int_{f^{-1}(t)}^sk(s')\,ds',&~s<0\ \mathrm{and}\ s\leq f^{-1}(t), \\
s-t\int_0^{s}k(s')\,ds' ,&\mathrm{elsewhere}.
\end{cases}
\end{equation*}
\paragraph{\itshape Case 3.}
\begin{equation*}
\breve{F}_\mathfrak g^j(s,t)=
\begin{cases}
s-t\int_0^s k(s')\,ds',&~s>0, \\
as-at\int_0^s k(s')\,ds',&~s<0.
\end{cases}
\end{equation*}

Note that $\breve{\Phi}^{-1}\big(\mathcal N(\mathsf p_j,2\epsilon_0)\big)$ is simply connected and that $\curl_{s,t}\breve{\Fb}=\curl_{s,t}\breve{\Fb}_\mathfrak g^j$ in each of the aforementioned cases. Consequently, there exists  a function $\breve{\varphi}^j\in H^2\big(\Phi^{-1}\big(\mathcal N(\mathsf p_j,2\epsilon_0)\big)\big)$ such that
\[\breve{\Fb}+\nabla_{s,t}\breve{\varphi}^j=\breve{\Fb}_\mathfrak g^j.\]	
Having $k(s)$ bounded and $\epsilon_0$ small, and using the properties of the diffeomorphism $\Phi$, one can see that $|\breve{\Fb}_\mathfrak g^j|\leq C_1|s|\leq C\dist(x,\mathsf p_j)$ for some $C_1,C>0$. 
	
Now, we consider $\chi \in \mathcal C^\infty(\overline{\Om})$ such that 
\[\supp \chi \subset \bigcup_{j=1}^n\mathcal N(\mathsf p_j,2\epsilon_0),\quad 0\leq \chi\leq 1 \quad  \mathrm{and}\quad \chi=1\ \mathrm{in}\ \bigcup_{j=1}^n\mathcal N\big(\mathsf p_j,\epsilon_0\big).\]
Hence, defining  $\varphi(x)=\breve{\varphi}^j\big((\breve{\Phi})^{-1}(x)\big)\chi(x)$ completes the proof. 
	\end{proof}		
		
		\begin{proof}[Proof of Proposition~\ref{prop:monot}]
			Let $\mathfrak b \geq 0$ and $M$ be the multiplicity of $\lambda(\mathfrak b)$. Recall that the domains of the corresponding operator and quadratic form are independent of $\mathfrak b$ (see~\eqref{eq:domP_F}). The perturbation theory 
			 asserts the existence of $\epsilon>0$, and  analytic functions 
			 \[(\mathfrak b-\epsilon,\mathfrak b+\epsilon)\ni \beta\mapsto \psi_m\in H^2(\Om)\setminus \{0\},\]
			 \[(\mathfrak b-\epsilon,\mathfrak b+\epsilon)\ni \beta\mapsto E_m\in \R,\]
			 for $m=1,...,M$, such that the functions $\{\psi_m(\mathfrak b)\}$ are linearly independent and normalized in $L^2(\Om)$, and
			\[\mathcal P_{\beta,\Fb}\psi_m(\beta)=E_m(\beta)\psi_m(\beta),\quad E_m(\mathfrak b)=\lambda(\mathfrak b).\]
				For small $\epsilon$, there exist $m_+$ and $m_-$ in $\{1,...,M\}$ such that
				\[\mathrm{for}\ \beta \in (\mathfrak b,\mathfrak b+\epsilon),\qquad  E_{m_+}(\beta)=\min_{\{1,...,M\}}E_m(\beta),\]
				\[\mathrm{for}\ \beta \in (\mathfrak b-\epsilon,\mathfrak b),\qquad  E_{m_-}(\beta)=\min_{\{1,...,M\}}E_m(\beta).\]
				Let $\Fb_\mathfrak g$ be the field introduced in Lemma~\ref{lem:gauge_field},
				and $\mathcal P_{\mathfrak b, \Fb_{\mathfrak g}}$, $Q_{\mathfrak b,\Fb_{\mathfrak g}}$ be  the  operator and  the quadratic form defined  in~\eqref{eq:P0} and~\eqref{eq:Quad} respectively.
				The operators $\mathcal P_{\mathfrak b,\Fb_{\mathfrak g}}$ and $\mathcal P_{\mathfrak b, \Fb}$ are unitarily equivalent. Indeed,  $\mathcal P_{\mathfrak b,\Fb_{\mathfrak g}}=e^{i\mathfrak b\varphi}\mathcal P_{\mathfrak b, \Fb}e^{-i\mathfrak b\varphi}$, where $\varphi$ is the gauge function in Lemma~\ref{lem:gauge_field}. Let $\psi_{\mathfrak g,m\pm}(\mathfrak b)=e^{i\mathfrak b\varphi}\psi_{m\pm}(\mathfrak b)$ be normalized eigenfunctions of $\mathcal P_{\mathfrak b,\Fb_{\mathfrak g}}$, associated with the lowest ground-state energy $\lambda(\mathfrak b)$.
				By the first order perturbation theory, the derivatives $\lambda'_\pm(\mathfrak b)$ can be written as
				\[
				\lambda'_\pm(\mathfrak b)
				=\frac{d}{d\beta}Q_{\beta,\Fb_{\mathfrak g}}\big(\psi_{\mathfrak g,m\pm}(\beta)\big)|_{\beta=\mathfrak b}
				=2 \im\big\langle\Fb_\mathfrak g \psi_{\mathfrak g,m\pm}(\mathfrak b),(\nabla-i\mathfrak b\Fb_\mathfrak g)\psi_{\mathfrak g,m\pm}(\mathfrak b)\big\rangle.
				\]
				This implies for any $B>0$ 
				\begin{align*}
			\lambda'_+(\mathfrak b)&= 
				\frac {Q_{\mathfrak b+B,\Fb_{\mathfrak g}}\big(\psi_{\mathfrak g,m+}(\mathfrak b)\big)-Q_{\mathfrak b,\Fb_{\mathfrak g}}\big(\psi_{\mathfrak g,m+}(\mathfrak b)\big)}{B}-B\int_{\Om}|\Fb_\mathfrak g|^2 |\psi_{\mathfrak g,m+}(\mathfrak b)|^2\,dx ,\\
				&\geq \frac{\lambda(\mathfrak b+B)-\lambda(\mathfrak b)}{B}
-B\int_\Om|\Fb_\mathfrak g|^2 |\psi_{\mathfrak g,m+}(\mathfrak b)|^2\,dx.			
\end{align*}
We decompose the integral in the right hand side of the previous inequality into two, one over $\bigcup_{j=1}^n\mathcal N(\mathsf p_j,\epsilon_0)$ and the other over its complement. By Theorem~\ref{thm:lin_bound} and Lemma~\ref{lem:gauge_field}, the first integral is bounded from above by $C\mathfrak b^{-1}$ (assuming $\mathfrak b$ large). The second integral is bounded  by  $C\|\Fb_\mathfrak g\|^2_\infty\mathfrak b^{-1}$, due to the exponential decay in Theorem~\ref{thm:lin_bound}\footnote{The fact that $\Fb_\mathfrak g \in L^{\infty}(\Om)$ can be deduced from the explicit definition of this field in Lemma~\ref{lem:gauge_field} together with the boundedness of the potential  $\Fb$ established in~\ref{thm:regul}.}. 
These bounds imply that $\int_\Om|\Fb_\mathfrak g|^2 |\psi_{\mathfrak g,m+}(\mathfrak b)|^2$ is bounded by $C\mathfrak b^{-1}$.
Hence, choosing $B=\eta\mathfrak b$ for any $\eta>0$ and using Propositions~\ref{prop:Lower3} and~\ref{prop:Upper3}, we get
\[\liminf_{\mathfrak b\rightarrow +\infty}\lambda'_+(\mathfrak b)\geq \min_{j\in\{1,...,n\}}\mu(\alpha_j,a)-C\eta.\]
Since $\eta$ is arbitrary, then
\begin{equation}\label{eq:lmda'1}
\liminf_{\mathfrak b\rightarrow +\infty}\lambda'_+(\mathfrak b)\geq \min_{j\in\{1,...,n\}}\mu(\alpha_j,a).\end{equation}
For $B<0$, we use a similar argument to get
\begin{equation}\label{eq:lmda'2}\limsup_{\mathfrak b\rightarrow +\infty}\lambda'_-(\mathfrak b)\leq \min_{j\in\{1,...,n\}}\mu(\alpha_j,a).\end{equation}
By the perturbation theory $\lambda'_+(\mathfrak b)\leq \lambda'_-(\mathfrak b)$. This together with~\eqref{eq:lmda'1} and~\eqref{eq:lmda'2} complete the proof.
	\end{proof}
\begin{proposition}	\label{prop:H_unique}
	Under Assumption~\ref{assump3}, there exists $\kappa_0>0$ such that for all $\kappa\geq \kappa_0$, the equation in $H$
	\[\lambda(\kappa H)=\kappa^2\]
	has a unique solution, which we denote by $H_{C_3}(\kappa)$.
\end{proposition}	
\begin{proof}
	Proposition~\ref{prop:monot} and the perturbation theory ensure the existence of $\mathfrak b_0$ such that $\mathfrak b\mapsto\lambda(\mathfrak b)$  is a strictly increasing continuous function from $[\mathfrak b_0,+\infty)$ onto $[\lambda(\mathfrak b_0),+\infty)$. We may choose $\mathfrak b_0$ sufficiently large so that for any $0<\mathfrak b<\mathfrak b_0$, $\lambda(\mathfrak b)<\lambda(\mathfrak b_0)$. 
	Let $\kappa_0=\sqrt{\lambda(\mathfrak b_0)}$, then for all $\kappa\geq\kappa_0$, the equation
	\[\lambda(\kappa H)=\kappa^2\]
	admits a unique solution $H_{C_3}(\kappa)=\lambda^{-1}(\kappa^2)/\kappa$, where $\lambda^{-1}(\cdot)$ is the inverse function of $\lambda(\cdot)$ defined on $[\lambda(\mathfrak b_0),+\infty)$.
\end{proof}
\begin{rem}\label{rem:equal_loc}
	For $\kappa>0$, recall  the local critical fields, 
	$\overline{H}^\mathrm{loc}_{C_3}(\kappa)$ and $\underline{H}^\mathrm{loc}_{C_3}(\kappa)$,  defined in~\eqref{eq:Hc3_loc_over} and~\eqref{eq:Hc3_loc_under} respectively. For sufficiently large values of $\kappa$, the equality of these two critical fields follows easily from the result established in Proposition~\ref{prop:H_unique}.
\end{rem}
\begin{proposition}\label{prop:Hc3}
	 Under Assumption~\ref{assump3}, there exists $\kappa_0>0$ such that for all $\kappa\geq\kappa_0$, the unique solution, $H=H_{C_3}(\kappa)$, to the equation
	\[ \lambda(\kappa H)=\kappa^2\]
	satisfies the following. There exist positive constants $\eta_1$ and $\eta_2$ such that
	\[-\eta_1\kappa^\frac 12\leq H_{C_3}(\kappa) -\frac \kappa {\min\limits_{j \in\{1,...,n\}}\mu(\alpha_j,a)}\leq\eta_2\kappa^\frac 12. \]	
	\end{proposition}	
\begin{proof}
 We assume that $\kappa$ is sufficiently large so that the results of Proposition~\ref{prop:H_unique} hold. We will suitably define two  fields $H_1=H_1(\kappa)$ and $H_2=H_2(\kappa)$ satisfying 
	\[\lambda(\kappa H_1)<\kappa^2\qquad \mathrm{and}\qquad \lambda(\kappa H_2)>\kappa^2,\]
then the desired result follows by using the continuity of $\mathfrak b \mapsto \lambda(\mathfrak b)$.

	Set $H_1=\kappa/\mu^*-\eta_1\kappa^{\delta_1}$, where $\eta_1>0$ and $\delta_1 \in (0,1)$ are two constants to be  chosen  soon.  For any fixed choice of $\eta_1$ and $\delta_1$, we assume that $\kappa$ is sufficiently large so that $H_1>1$. Hence, Theorem~\ref{thm:lambda_lin} asserts the existence of $\kappa_0>0$ and  $C>0$ such that for all $\kappa \geq \kappa_0$,
	\begin{align*}
	\lambda(\kappa H_1)\leq \mu^* \kappa H_1 +C(\kappa H_1)^\frac 34
	&\leq \kappa^2-\eta_1\mu^* \kappa^{1+\delta_1}+C\kappa^\frac32\big( (\mu^*)^{-1}-\eta_1\kappa^{-1+\delta_1}\big)^\frac 34\\
	&\leq \kappa^2-\eta_1\mu^* \kappa^{1+\delta_1}+ C {(\mu^*)}^{-3/4}\kappa^\frac32.
	\end{align*}
	Choose $\delta_1=1/2$ and $\eta_1>C{(\mu^*)}^{-7/4}$ (so that $-\eta_1\mu^* + C {(\mu^*)}^{-3/4}<0$). This choice of parameters yields 
		\begin{equation*}\label{eq:H1}
		\lambda(\kappa H_1)<\kappa^2\,,\qquad \mbox{for all}\ \kappa\geq \kappa_0.
		\end{equation*}
	Similarly, set $H_2=\kappa/\mu^*+\eta_2\kappa^{\delta_2}$, where $\eta_2>0$ and $\delta_2 \in (0,1)$ are constants to be  chosen. By Theorem~\ref{thm:lambda_lin},  there exists $\kappa_0>0$ and  $C>0$ such that for all $\kappa \geq \kappa_0$,
	\begin{equation*}
	\lambda(\kappa H_2)\geq \mu^* \kappa H_2 -C(\kappa H_2)^\frac 34
	\geq\kappa^2+\eta_2\mu^* \kappa^{1+\delta_2}-C(\mu^*)^{-\frac 34}\kappa^\frac32.	
	\end{equation*}
	Choose $\delta_2=1/2$ and $\eta_2$ such that $\eta_2>C{(\mu^*)}^{-7/4}$ to obtain 
	\begin{equation*}\label{eq:H2}
	\lambda(\kappa H_2)>\kappa^2\,,\qquad \mbox{for all}\ \kappa\geq \kappa_0.
	\end{equation*}
\end{proof}
\section{Proof of Theorem~\ref{thm:agmon}}\label{sec:agmon}
The aim of this section is to establish Theorem~\ref{thm:agmon}. This theorem displays how, with an increasing field, the order parameter  (in~\eqref{eq:Euler}) and the corresponding GL energy decay successively away from  the intersection points of $\Gamma$ and $\partial \Om$,  $\{\mathsf p_j\}_j$, according to the ordering of the eigenvalues $\{\mu(\alpha_j,a)\}_j$. Moreover, it asserts the eventual localization of the order parameter near the point(s) ${\mathsf p_k}$ with the smallest corresponding eigenvalue $\mu^*$.

The following lower bound is crucial in establishing Theorem~\ref{thm:agmon}.
\begin{lemma}\label{lem:spectral}
	Suppose that $\Om$ satisfies Assumption~\ref{assump3}. Let $\T=\{1,...,n\}$ and $\mu>0$ satisfy $\mu^*\leq\mu<|a|\Theta_0$. Define 
	\[\Sigma=\left\{ j\in\T~:~\mu(\alpha_j,a)\leq\mu\right\},\ S=\left\{\mathsf p_j\in\Gamma\cap\partial\Om,~j\in\Sigma\right\}\ \mathrm{and}\ d=\min_{j\in \T\setminus\Sigma}\mu(\alpha_j,a)-\mu\]
	(in the case $\Sigma=T$, we set $d=|a|\Theta_0-\mu$).
	There exists $C>0$, and for all $R_0>1$ there exists $\tilde{\kappa}_0>0$ such that for $\kappa\geq\tilde{\kappa}_0$, if $(\psi,\Ab) \in H^1(\Om;\C)\times \Hd$ is a critical point of~\eqref{eq:Euler}, $H$ satisfies $H\geq \kappa/\mu$	and $Q_{\kp H,\Ab}$ is the form in~\eqref{eq:Quad}, then for $\varphi \in \dom Q_{\kp H,\Ab}$  such that $\dist(\supp\varphi,S)\geq R_0(\kappa H)^{-1/2}$ we have
	\[Q_{\kp H,\Ab}(\varphi)\geq \kappa H\Big(\mu+\frac d2-\frac C{R_0^2}\Big)\|\varphi\|_{L^2(\Om)}^2.\]
\end{lemma}
\begin{proof}
	Let $\varphi \in \dom Q_{\kp H,\Ab}$  be such that $\dist(\supp\varphi,S)\geq R_0 (\kappa H)^{-1/2}$, $\Fb$ be the vector potential defined in~\eqref{A_1}, and $\beta \in (0,1)$. We consider the family of cut-off functions $(\chi_j)_{j\in \mathcal P}$ introduced in Section~\ref{sec:partition} for $\mathfrak b=\kappa H$ and  $\rho=1/2$.
	For all $j\in \mathcal P$, we define on $\overline \Om$ the function
	$\phi_j(x)=\big(\Ab(z_j)-\Fb(z_j)\big)\cdot x$. As a consequence of the last item in Theorem~\ref{thm:priori}, we may approximate the vector potential $\Ab$ as follows:
	\begin{equation}\label{eq:AF_loc}
	|\Ab(x)-\nabla\phi_j(x)-\Fb(x)|\leq C\frac{R_0^\beta(\kappa H)^{-\frac 12\beta}}{H}\,,\ \mbox{for all}\ x\in B(z_j,R_0(\kappa H)^{-\frac 12})\cap \overline{\Om}.
	\end{equation}
	We choose $\beta=3/4$ and we define $h=e^{-i\kappa H\phi_j}\varphi$. Using~\eqref{eq:AF_loc} and $H\geq \kappa/\mu$, Cauchy's inequality yields
	\begin{equation}\label{eq:Ag1}
	\|(\nabla-i\kappa H\Ab)\chi_j\varphi\|^2_{L^2(\Om)}\geq (1-\kappa^{-\frac 12})\|(\nabla-i\kappa H\Fb)\chi_j h\|^2_{L^2(\Om)}-CR_0^{\frac 32}\kappa \|\chi_j\varphi\|^2_{L^2(\Om)}. \end{equation}
	Notice that $\supp h=\supp \varphi$. Hence~\eqref{eq:Ag1}, $H\geq \kappa/\mu$, the support of $\varphi$ and Proposition~\ref{prop:Ub2} assert that  
	\begin{align*}
	\|(\nabla-i\kappa H\Ab)\chi_j\varphi\|^2_{L^2(\Om)}&\geq (1-\kappa^{-\frac 12})\kappa H\Big(\min_{j \in \T\setminus \Sigma}\mu(\alpha_j,a)-\frac C{R_0^2}-CR_0^4\kappa^{-1}\Big)\|\chi_j\varphi\|^2_{L^2(\Om)}\nonumber \\&\qquad-CR_0^{\frac 32}\kappa\|\chi_j\varphi\|^2_{L^2(\Om)}\\
	&\geq\kappa H\Big(\min_{j \in \T\setminus \Sigma}\mu(\alpha_j,a)-\frac C{R_0^2}-CR_0^4\kappa^{-1}\Big)\|\chi_j\varphi\|^2_{L^2(\Om)}. 
	\end{align*}
	Hence, the IMS formula gives
	\begin{equation}\label{eq:spect2}
	\|(\nabla-i\kappa H\Ab)\varphi\|^2_{L^2(\Om)}
	\geq\kappa H\Big(\min_{j \in \T\setminus \Sigma}\mu(\alpha_j,a)-\frac C{R_0^2}-CR_0^4\kappa^{-1}\Big)\|\varphi\|^2_{L^2(\Om)}. 
	\end{equation}
	Choose $\tilde{\kappa}_0$ sufficiently large so that  $CR_0^4\tilde{\kappa}_0^{-1}<d/2$, for $d=\min_{j\in \T\setminus\Sigma}\mu(\alpha_j,a)-\mu$. Consequently,~\eqref{eq:spect2} yields for all $\kp\geq \tilde{\kappa}_0$,
	\begin{align*}
	\|(\nabla-i\kappa H\Ab)\varphi\|^2_{L^2(\Om)}&\geq \kappa H\Big(\min_{j \in \T\setminus \Sigma}\mu(\alpha_j,a)-\frac d2-\frac C{R_0^2}\Big)\|\varphi\|^2_{L^2(\Om)}\\
	& \geq\kappa H\Big(\mu+\frac d2-\frac C{R_0^2}\Big)\|\varphi\|^2_{L^2(\Om)}. \qedhere
	\end{align*}
\end{proof}
\begin{proof}[Proof of Theorem~\ref{thm:agmon}] 
	Let $R_0>1$. Take $\kappa_0$ in Theorem~\ref{thm:agmon} to be $\kappa_0=\max(\tilde{\kappa}_0,\kappa_1)$, where $\tilde{k}_0$ and $\kappa_1$ are  the constants in Lemma~\ref{lem:spectral} and Theorem~\ref{thm:giorgi} respectively. Assume that $\mu^{-1}\leq C_1$, where $C_1$ is the constant in Theorem~\ref{thm:giorgi}, else Equation~\eqref{eq:agmon} is evidently true for any positive constants $C$ and $\beta$, in light of Theorem~\ref{thm:giorgi}. 
	So, the case examined below is
	\begin{equation}\label{eq:A5}
	\kappa \geq \kappa_0,\quad \mathrm{and}\ \mu^{-1}\leq \frac H\kappa\leq C_1.
	\end{equation}
	Let $S$ be the set  appearing in Lemma~\ref{lem:spectral},
	$t(x)=\dist(x,S)$, and $\tilde{\chi} \in \mathcal C^\infty(\R)$ be a function satisfying
	\begin{equation*}
	\tilde{\chi}=0\ \mathrm{on}\ (-\infty,1/2]\qquad \mathrm{and}\qquad \tilde{\chi}=1\ \mathrm{on}\ [1,+\infty).
	\end{equation*}
	We define the two functions $\chi$ and $f$ as follows:
	\begin{equation}\label{eq:expo1}
	\chi(x)=\tilde{\chi}\big(R_0^{-1}{(\kappa H)^{\frac 12 }} t(x)\big) \ \mathrm{and} \ f(x)=\chi(x)\exp\big(\beta (\kappa H)^{\frac 12}t(x)\big),
	\end{equation}
	where $\beta$ is a positive constant whose value will be fixed soon. Integrating in the first equation of~\eqref{eq:Euler}, we get
	\begin{equation}\label{eq:decay1}
	\int_{\Om} |\nb f|^2|\psi|^2 \,dx\geq \int_{\Om} \big|(\nb-i\kp H {\Ab})f\psi\big|^2\,dx-\kappa^2\int_{\Om}|\psi|^2f^2\,dx.
	\end{equation}
	Notice that the conditions in Lemma~\ref{lem:spectral} are satisfied for $\varphi=f\psi$, hence we may apply this lemma to obtain
	\[\int_{\Om} |\nb f|^2|\psi|^2 \,dx\geq \left(\kappa H
	\Big(\mu+\frac d2-\frac C{R_0^2}\Big)-\kappa^2\right)\|f\psi\|^2_{L^2(\Om)}.\]
	Since  $H\geq\kappa/\mu$, we get further
	\begin{equation}\label{eq:decay2}
	\int_{\Om} |\nb f|^2|\psi|^2 \,dx\geq \Big(\frac d2-\frac C{R_0^2}\Big)\mu^{-1} \kappa^2\|f\psi\|^2_{L^2(\Om)}.
	\end{equation}
	On the other hand, using~\eqref{eq:expo1}, we estimate the term $\int_{\Om} |\nb f|^2|\psi|^2 \,dx$ as follows:
	\begin{equation}\label{eq:decay3}
	\int_{\Om} |\nb f|^2|\psi|^2 \,dx\leq 2\beta^2\kp H \|f \psi\|^2_{L^2(\Om)}+C(R_0) \kp H\int_{ \big\{\sqrt{\kp H}t(x)< R_0\big\}}|\psi|^2\,dx,
	\end{equation}
	where $C(R_0)$ is a constant only dependent on $R_0$. Recall that we are working  under the assumption in~\eqref{eq:A5}. Hence, we combine~\eqref{eq:decay2} and~\eqref{eq:decay3}, and we divide by $\kappa^2$ to get
	\[\Big(\frac{\mu^{-1}d}2-\frac{C\mu^{-1}}{R_0^2}-2C_1\beta^2\Big)\|f\psi\|^2_{L^2(\Om)}\leq \tilde{C}(R_0)\int_{ \{\sqrt{\kp H} t(x)< R_0\}}|\psi|^2\,dx,\]
	where $C_1$ is the value in~\eqref{eq:A5}.
	We choose $\beta$ small so that $\mu^{-1}d-4C_1\beta^2>0$ (that is $\beta<1/2\sqrt{\mu^{-1}d/C_1}$). Consequently, for $R_0$ sufficiently large, we get the existence of $\hat C=C(R_0,\beta)>0$ such that 
	\begin{equation}\label{eq:decay4}
	\|f\psi\|^2_{L^2(\Om)}\leq \hat C\int_{ \big\{\sqrt{\kp H} t(x)< R_0\big\}}|\psi|^2\,dx.
	\end{equation}
	Plug~\eqref{eq:decay3} and~\eqref{eq:decay4} in~\eqref{eq:decay1} to complete the proof.
\end{proof}	
	\section{Equality of global and local fields}\label{sec:main thm}
		 We consider the global and local critical fields $\overline{H}_{C_3}(\kappa)$, $\underline{H}_{C_3}(\kappa)$, $\overline{H}^\mathrm{loc}_{C_3}(\kappa)$ and $\underline{H}^\mathrm{loc}_{C_3}(\kappa)$ defined  in~\eqref{eq:Hc3_over},~\eqref{eq:Hc3_under},~\eqref{eq:Hc3_loc_over}, and~\eqref{eq:Hc3_loc_under} respectively. 
		\begin{theorem}\label{thm:Hc_order}
			Let $\kappa>0$. Under Assumption~\ref{assump1}, the following relations hold:
			\begin{equation}\label{eq:Hc_order1}
			\overline{H}_{C_3}(\kappa)\geq \overline{H}^\mathrm{loc}_{C_3}(\kappa),\qquad \underline{H}_{C_3}(\kappa)\geq \underline{H}^\mathrm{loc}_{C_3}(\kappa).
			\end{equation}
		\end{theorem}
		\begin{proof}
			First, we prove the left inequality in~\eqref{eq:Hc_order1}. 
			Let $H<\overline H^{\mathrm {loc}}_{C_3}(\kappa)$, hence there exists $H_0>H$ such that 
			\begin{equation}\label{eq:hypo1}
			\lambda(\kappa H_0)-\kappa^2<0,
			\end{equation}
			where $\lambda(\kappa H_0)$ is the value in~\eqref{eq:lmda_a}.
			It suffices to prove that $H<\overline H_{C_3}(\kappa)$. 
			Let $\psi_0$ be a normalized ground-state of $\mathcal P_{\kappa H_0,\Fb}$ in~\eqref{eq:P0}.
			Let $t>0$, we have
			\[\mathcal E_{\kappa,H_0}(t\psi_0,\Fb)=t^2(\lambda(\kappa H_0)-\kappa^2)+\frac {\kappa^2}2t^4\|\psi_0\|^4_{L^4(\Om)}.\]
			Choose $t$ such that $t^2<2\big(\kappa^2-\lambda(\kappa H_0)\big)/{\kappa^2\|\psi_0\|^4_{L^4(\Om)}}$, and use~\eqref{eq:hypo1} to get
			\[\mathcal E_{\kappa,H_0}(t\psi_0,\Fb)<0.\]
			This reveals the existence of a non-trivial minimizer of $\mathcal E_{\kappa,H_0}$. Recalling the definition of $\overline H_{C_3}(\kappa)$, we get that $H<\overline H_{C_3}(\kappa)$ which yields the claim. 
			
			Secondly, to derive the right inequality in~\eqref{eq:Hc_order1}, we proceed as in the argument above to get that  $\mathcal E_{\kappa,H}$ has a non-trivial minimizer, for all  $H<\underline H^{\mathrm {loc}}_{C_3}(\kappa)$. Consequently, assuming that $\underline H^{\mathrm {loc}}_{C_3}(\kappa)> \underline H_{C_3}(\kappa)$ contradicts the definition of $\underline H_{C_3}(\kappa)$.
		\end{proof}
With Theorem~\ref{thm:Hc_order} and the equality of the local critical fields in hand (see Remark~\ref{rem:equal_loc}), it  remains to prove the equality of the local and global upper fields in order to establish the equality of the global and local fields. This together with Proposition~\ref{prop:H_unique} and Proposition~\ref{prop:Hc3} will complete the proof of Theorem~\ref{thm:Hc3}. To this end, we follow similar steps as in~\cite[Theorem 1.7]{bonnaillie2007superconductivity} and use the following additional result:
	\begin{thm}\label{thm:agmon1}
		Given $a \in[-1,1)\setminus\{0\}$, there exist positive constants $\kappa_0$, $C$ and $\delta$ such that if
		$\kappa \geq \kappa_0$ and $(\psi,\Ab)$ is a solution of~\eqref{eq:Euler} for $H>1/|a|\kappa$
		then
		\begin{multline*}
		\int_{\Omega}
		\Big(|\psi|^2+{\frac{1}{\kp H}}|(\nabla-i\kappa H\Ab)\psi|^2\Big)\,\exp\Big(2\delta\sqrt{\kappa H}\,{\rm dist}(x,\partial\Omega\cup \Gamma)\Big)dx\\
		\leq
		C \int_{\Omega\cap\{{\rm dist}(x,\partial\Omega \cup \Gamma)< \frac{1}{\sqrt{\kappa H}}\}}
		|\psi|^2\,dx.\end{multline*}
	\end{thm}
		Theorem~\ref{thm:agmon1} above displays certain Agmon-type estimates established in~\cite[Theorems~1.5~\&~7.3]{Assaad}. These estimates reveal the exponential decay of the order parameter and the GL energy in the bulk of $\Om_1$ and $\Om_2$, in a certain regime of the intensity of the applied magnetic field.
	\begin{proof}[Proof of Theorem~\ref{thm:Hc3}]
		Let $\kappa>0$.  $\overline H_{C_3}(\kappa)\geq\overline H^{\mathrm {loc}}_{C_3}(\kappa)$ was proved in Theorem~\ref{thm:Hc_order}. 
		Next, we prove that $\overline H_{C_3}(\kappa)\leq\overline H^{\mathrm {loc}}_{C_3}(\kappa)$.
		Assume that $\overline H_{C_3}(\kappa)>\overline H^{\mathrm {loc}}_{C_3}(\kappa)$, then the definitions of $\overline H_{C_3}(\kappa)$ and $\overline H^{\mathrm {loc}}_{C_3}(\kappa)$ ensure the existence of $H>0$ satisfying: 
		\begin{enumerate}
			\item $\overline H^{\mathrm {loc}}_{C_3}(\kappa)<H\leq\overline H_{C_3}(\kappa)$.
			\item $\lambda(\kappa H)\geq\kappa^2$.
			\item  The GL functional $\GL$ in~\eqref{eq:GL} admits  a non-trivial minimizer $(\psi,\Ab)$.
 		\end{enumerate}
	In particular, $(\psi,\Ab)$ satisfies
		\[\kappa^2\|\psi\|^2_{L^2(\Om)}>Q_{\kappa H ,\Ab}(\psi),\]
		where $Q_{\kappa H ,\Ab}$ is the quadratic form in~\eqref{eq:Quad}. We define $\Delta=\kappa^2\|\psi\|^2_{L^2(\Om)}-Q_{\kappa H ,\Ab}(\psi)$. An integration in the first GL equation of $\eqref{eq:Euler}$ gives 
		\begin{equation}\label{eq:psi_4_delta}
		\|\psi\|_{L^4(\Om)}^4=\frac {\Delta}{\kappa^2}.
		\end{equation} 
		Furthermore, the assumption that $H>\overline H_{C_3}^{\mathrm {loc}}(\kappa)$ and the asymptotics of $\overline H_{C_3}^{\mathrm {loc}}(\kappa)$ in Proposition~\ref{prop:Hc3} assert that we are working under the conditions of Theorem~\ref{thm:agmon1} and allow us to write
		\begin{align}
		\|\psi\|^2_{L^2(\Om)}&\leq 	C \int_{\Omega\cap\{{\rm dist}(x,\partial\Omega\cup \Gamma)< \frac{1}{\sqrt{\kappa H}}\}}
		|\psi|^2\,dx\nonumber\\
		                   &\leq C\|\psi\|_{L^4(\Om)}^2\Big(\int_{\Omega\cap\{{\rm dist}(x,\partial\Omega \cup \Gamma)< \frac{1}{\sqrt{\kappa H}}\}}
		                   \,dx\Big)^\frac 12  \leq C\kappa^{-\frac 32}\Delta^{\frac 12}.\label{eq:psi_2_delta}               
		\end{align} 
		The last inequality follows from~\eqref{eq:psi_4_delta}. Since $\psi\neq 0$ then, using  the min-max principle and Cauchy-Schwarz inequality, we can estimate
		\begin{equation}\label{eq:delta1}
		0<\Delta
		\leq \big(\kappa^2-(1-\delta)\lambda(\kappa H)\big)\|\psi\|^2_{L^2(\Om)}+C\delta^{-1}(\kappa H)^2\|\Ab-\Fb\|^2_{L^4(\Om)}\|\psi\|^2_{L^4(\Om)},
		\end{equation} 
		for any $\delta \in(0,1)$.
By the Sobolev estimates in $\R^2$ and the curl-div estimates (see~\cite[Proposition~D.2.1]{fournais2010spectral}), we have
		\[\|\Ab-\Fb\|_{L^4(\Om)}\leq C\|\Ab-\Fb\|_{H^1(\Om)}\leq C\|\curl(\Ab-\Fb)\|_{L^2(\Om)}.\]
		Consequently, since $\GL(\psi,\Ab)\leq0$ we conclude that
		\begin{equation}\label{eq:delta2}
	 (\kappa H)^2\|\Ab-\Fb\|^2_{L^4(\Om)}\leq C(\kappa H)^2\|\curl(\Ab-\Fb)\|^2_{L^2(\Om)}\leq C\Delta.
		\end{equation}
		Choose $\delta=\Delta^{1/2}\kappa^{-3/4}$. The hypothesis on $H$ and the definition of $\overline H_{C_3}(\kappa)$ together with Theorem~\ref{thm:giorgi} ensure that $H\leq C_1\kp$, where $C_1$ is the constant in the aforementioned theorem. We use this upper bound of $H$ and  Proposition~\ref{prop:Upper3}, and we insert~\eqref{eq:psi_4_delta},~\eqref{eq:psi_2_delta},  and~\eqref{eq:delta2} in~\eqref{eq:delta1}  to get
		\[0<\Delta\leq \big(\kappa^2-\lambda(\kappa H)\big)\|\psi\|^2_{L^2(\Om)}+C \Delta \kappa^{-\frac 14}.\]
		When $\kappa$ is  big, $1-C\kappa^{-1/4}>0$. Therefore, since $\lambda(\kappa H)\geq\kappa^2$ we get
		\[0<(1-C\kappa^{-\frac 14})\Delta\leq \big(\kappa^2-\lambda(\kappa H)\big)\|\psi\|_{L^2(\Om)}^2\leq 0,\]
	which is absurd. This means that  $\overline H_{C_3}(\kappa)\leq\overline H^{\mathrm {loc}}_{C_3}(\kappa)$. 
	\end{proof}
	 
\section*{Acknowledgements} I would like to thank Ayman Kachmar and Mikael Persson-Sundqvist for their encouragement and comments which  significantly improved this article. I am grateful to Jacob Christiansen for the insightful suggestions about various sections of the manuscript. I would also like to acknowledge the useful discussions with  Erik Wahl\'en and Anders Holst about certain regularity problems.

	\appendix
	\section{Some spectral properties of the model operator $\mathcal H_{\alpha,a}$}\label{sec:persson}
	Let $\alpha \in (0,\pi)$ and $a \in [-1,1)\setminus \{0\}$. Recall the operator $\mathcal H_{\alpha,a}$  defined on $\R^2_+$ in Section~\ref{sec:new_model}.
	This appendix is devoted to the establishment some spectral properties of this operator, presented in the aforementioned section. In particular, we prove the claim in Theorem~\ref{thm:ess_theta} that the bottom of the essential spectrum of $\mathcal H_{\alpha,a}$ is equal to $|a|\Theta_0$.
	
	Recall the set $\mathcal M_r$ defined in~\eqref{eq:M_r}. A central step in proving Theorem~\ref{thm:ess_theta} is to establish Theorem~\ref{thm:ess_spect} below.
	\begin{theorem}\label{thm:ess_spect}
		The essential spectrum of the Neumann realization of the operator $\mathcal H_{\alpha,a}$ defined in~\eqref{eq:P_alfa} satisfies
		\[\inf \spc_\mathrm{ess}\mathcal H_{\alpha,a}=\Sigma\mathcal H_{\alpha,a},\]
		where 
		\[\Sigma \mathcal H_{\alpha,a}=\lim_{r\rightarrow+\infty}\Sigma (\mathcal H_{\alpha,a},r)\] 
		and 
		\[\Sigma (\mathcal H_{\alpha,a},r)=\inf_{\substack{u\in  \mathcal M_r\\ u \neq 0}} \frac{\|(\nabla-i\Ab_{\alpha,a})u\|^2_{L^2(\R^2_+)}}{\|u\|^2_{L^2(\R^2_+)}}. \]
	\end{theorem}
	\begin{rem}	\label{persson}
	The function $r\mapsto \Sigma (\mathcal H_{\alpha,a},r)$ is increasing on $\R_+$. Indeed, if a function $u \in \mathcal M_r $ then $u \in \mathcal M_\rho$ for $\rho<r$. Consequently, the limit $\Sigma \mathcal H_{\alpha,a}$ exists and belongs to $(0,+\infty]$, having $\Sigma(\mathcal H_{\alpha,a},r)$  positive.
	\end{rem}
	The following lemma is needed in the proof of Theorem~\ref{thm:ess_spect}.
	\begin{lemma}\label{lem:weyl}
		Let $(u_n)$ be a Weyl sequence of the operator $\mathcal H_{\alpha,a}$. For all $r>0$, $(u_n)$ converges to zero in $L^2(B_r^+)$. 
	\end{lemma}
	\begin{proof}
		A Weyl sequence $(u_n)$ is included in $\dom\mathcal H_{\alpha,a}$ and satisfies:
		\begin{equation}\label{eq:weyl} \|u_n\|_{L^2(\R^2_+)}=1,\quad u_n\rightharpoonup0\quad \mathrm{and}\quad \|\mathcal H_{\alpha,a} u_n-\lambda u_n\|_{L^2(\R^2_+)}\rightarrow 0,
		\end{equation}
		where $\lambda$ is the scalar associated to $(u_n)$. First, we prove the boundedness of  $(u_n)$ in $H^1(B_r^+)$. Using Cauchy-Schwarz inequality, we have
		\begin{equation}\label{eq:P4}
		\langle \mathcal H_{\alpha,a} u_n,u_n\rangle-\lambda =\langle \mathcal H_{\alpha,a} u_n-\lambda u_n,u_n\rangle     
		\leq\|\mathcal H_{\alpha,a} u_n-\lambda u_n\|_{L^2(\R^2_+)}.
		\end{equation}
		The third property satisfied by $(u_n)$ in~\eqref{eq:weyl} assures the existence of $n_0\in \N$ such that for all $n \geq n_0$, $\|\mathcal H_{\alpha,a} u_n-\lambda u_n\|_{L^2(\R^2_+)}\leq 1$. Implementing this inequality in~\eqref{eq:P4}, we get for $n\geq n_0$
		\begin{equation}\label{eq:P1}
		\langle \mathcal H_{\alpha,a} u_n, u_n \rangle-\lambda \leq 1.
		\end{equation}
		Having $u_n \in \dom\mathcal H_{\alpha,a}$, we integrate by parts in~\eqref{eq:P1} to get
		\begin{equation*}	
		\|(\nabla-i\Ab_{\alpha,a})u_n\|^2_{L^2(\R^2_+)}+\|u_n\|^2_{L^2(\R^2_+)}\leq \lambda+2.
		\end{equation*}
		Particularly,
		\begin{equation}	\label{eq:P2}
		\|(\nabla-i\Ab_{\alpha,a})u_n\|^2_{L^2(B_r^+)}+\|u_n\|^2_{L^2(B_r^+)}\leq \lambda+2.
		\end{equation}
		Thus, there exists  $C>0$ dependent on $r$ such that
		\[\|\nabla u_n\|^2_{ L^2(B_r^+)}+\|u_n\|^2_{L^2(B_r^+)}\leq \lambda+C,\]
		having $\Ab_{\alpha,a}$ bounded in  $B_r^+$. Hence $(u_n)$ is bounded in $H^1(B_r^+)$. 
		
		Next, we prove that the sequence $(u_n)$ converges to zero in $L^2(B_r^+)$. Suppose not, then there exist $\epsilon>0$ and a subsequence $(u_{n_j})$ of $(u_n)$ such that 
		\begin{equation}	\label{eq:P3}
		\|u_{n_j}\|_{L^2(B_r^+)}> \epsilon.
		\end{equation} 
		The boundedness of $(u_n)$ in $H^1(B_r^+)$ and the compact injection of $H^1(B_r^+)$ into $L^2(B_r^+)$ imply that $(u_{n_j})$ is convergent in $L^2(B_r^+)$, along a subsequence. The second property in~\eqref{eq:weyl} assures that the limit of this subsequence is zero, which contradicts~\eqref{eq:P3}.
	\end{proof}
	\begin{proof}[Proof of Theorem~\ref{thm:ess_spect}]
		First we prove 
		\begin{equation}	\label{eq:first}
		\Sigma\mathcal H_{\alpha,a} \leq \inf \spc_\mathrm{ess}(\mathcal H_{\alpha,a}).
		\end{equation}
		Let $\lambda \in \spc_\mathrm{ess}(\mathcal H_{\alpha,a})$. Recalling the definition of $\Sigma\mathcal H_{\alpha,a}$, it suffices to prove that $\Sigma (\mathcal H_{\alpha,a},r)\leq \lambda$ for all $r>0$. 
		We consider the Weyl sequence $(u_n)$ associated to $\lambda$, and localize this sequence outside $B_r^+$ by using a truncation function $\chi \in \mathcal C^\infty(\R^2_+,[0,1])$ 
		satisfying for $\rho>r$
		\begin{equation}\label{eq:chi}
		\chi(x)=1\ \mathrm{in}\ B_\rho^\complement\cap \R^2_+,\qquad \mathrm{and}\ \chi(x)=0\ \mathrm{in}\ B_r^+.
		\end{equation}
		Note that $\chi u_n \in \mathcal M_r$.
		The triangle inequality gives 
		\begin{equation}\label{eq:chi0}
		\|(\nabla-i\Ab_{\alpha,a})\chi u_n\|_{L^2(\R^2_+)}\leq  \|\chi(\nabla-i\Ab_{\alpha,a})u_n\|_{L^2(\R^2_+)}+ \|u_n|\nabla \chi| \|_{L^2(\R^2_+)}.                                            
		\end{equation}
		Using  the properties of $\chi$ in~\eqref{eq:chi}, we have
		\begin{equation*}
		\|u_n|\nabla \chi| \|^2_{L^2(\R^2_+)}
		= \int_{B_\rho\cap \R^2_+}|u_n|^2|\nabla \chi |^2\,dx
		 \leq  C^2\int_{B_\rho\cap \R^2_+}|u_n|^2\,dx.
		\end{equation*}
		But $(u_n)$ converges to zero in $L^2\big(B_\rho \cap \R^2_+\big)$ by Lemma~\ref{lem:weyl}, then for any $\epsilon>0$ there exists $n_0 \in \N$ such that for $n\geq n_0$, 
		\begin{equation*}
		\int_{B_\rho\cap \R^2_+}|u_n|^2\,dx\leq \frac{\epsilon^2}{C^2}.
		\end{equation*}
		Hence,
		\begin{equation}\label{eq:chi1}
		\|u_n|\nabla \chi| \|_{L^2(\R^2_+)}\leq \epsilon .
		\end{equation}  
		On the other hand, the properties of $(u_n)$ and $\chi$ in~\eqref{eq:weyl} and~\eqref{eq:chi} respectively, together with an integration by parts, ensure the existence of $n_1\geq n_0$ such that for all $n \geq n_1$,
		\begin{equation*}
		\|\chi(\nabla-i\Ab_{\alpha,a})u_n\|^2_{L^2(\R^2_+)}\leq \|(\nabla-i\Ab_{\alpha,a})u_n\|^2_{L^2(\R^2_+)}\leq\lambda+\epsilon.
		\end{equation*}
		Put the above inequality together with~\eqref{eq:chi1} into~\eqref{eq:chi0} to get
		\begin{equation}\label{eq:chi2}
		\|(\nabla-i\Ab_{\alpha,a})\chi u_n\|^2_{L^2(\R^2_+)}\leq  \lambda+ C\epsilon.
		\end{equation}
		Next, we prove that for  $n$ sufficiently large
		\begin{equation}\label{eq:chi3}
		\frac 1{\|\chi u_n\|^2_{L^2(\R^2_+)}}\leq 1+\epsilon.
		\end{equation} 
		We have
		\begin{align}\label{eq:chi4}
		1=\|u_n\|^2_{L^2(\R^2_+)}\geq \|\chi u_n\|^2_{L^2(\R^2_+)}
		&=\int_{B_\rho\cap \R^2_+}|\chi u_n|^2\,dx+\int_{B_\rho^\complement\cap \R^2_+}|u_n|^2\,dx\nonumber\\
		&=\int_{\R^2_+}|u_n|^2\,dx+\int_{B_\rho\cap \R^2_+}(\chi^2-1)| u_n|^2\,dx\nonumber\\
		&\geq 1-\int_{B_\rho\cap \R^2_+}|u_n|^2\,dx. 
		\end{align}
		In light of Lemma~\ref{lem:weyl}, we introduce $\lim_{n\rightarrow +\infty}$ on~\eqref{eq:chi4} to get the convergence of $\|\chi u_n\|^2_{L^2(\R^2_+)}$ to $1$ as $n$ tends to $+\infty$, which proves~\eqref{eq:chi3}.
		The inequalities in~\eqref{eq:chi2} and~\eqref{eq:chi3} imply the existence of $n_2 \in \N$ and a positive constant $C$, independent of $\epsilon$, such that for all $n\geq n_2$,
		\[\frac {\|(\nabla-i\Ab_{\alpha,a})\chi u_n\|^2_{L^2(\R^2_+)}}{\|\chi u_n\|^2_{L^2(\R^2_+)}}\leq \lambda+C\epsilon.\] 
		Then by the definition of $\Sigma (\mathcal H_{\alpha,a},r)$, we get for any $\lambda \in \spc_\mathrm{ess}(\mathcal H_{\alpha,a})$
		\[\Sigma (\mathcal H_{\alpha,a},r)\leq \lambda+C\epsilon.\]
		Taking $\epsilon$ to zero  establishes~\eqref{eq:first}.

		Now we prove that
		\begin{equation}\label{eq:second}
		\Sigma\mathcal H_{\alpha,a} \geq \inf \spc_\mathrm{ess}(\mathcal H_{\alpha,a}).
		\end{equation}
		Let $\mu<\inf \spc_\mathrm{ess}(\mathcal H_{\alpha,a})$ and $\epsilon>0$. By Remark~\ref{persson} , it is sufficient to establish the existence of $r_\epsilon>0$ such that 
		\begin{equation}\label{eq:second1}
		\Sigma(\mathcal H_{\alpha,a},r_\epsilon) \geq \mu-\mathcal O(\epsilon).
		\end{equation}
		By the min-max principle, the previous inequality trivially holds if $\mu<\inf \spc(\mathcal H_{\alpha,a})$. Assume now that $\inf \spc(\mathcal H_{\alpha,a})\leq\mu<\inf \spc_\mathrm{ess}(\mathcal H_{\alpha,a})$. Let $q_{\alpha,a}$ be the quadratic form associated to $\mathcal H_{\alpha,a}$, and $\mathbbm{1}_{(-\infty,\mu]}(\mathcal H_{\alpha,a})$ be the spectral projection operator corresponding to this operator, that has finite rank (since we are below the essential spectrum). 
		There exists a finite orthonormal system of  normalized eigenfunctions $(v_i) \in L^2(\R^2_+)$ such that 
		\[\mathbbm{1}_{(-\infty,\mu]}(\mathcal H_{\alpha,a})=\sum_i\langle.,v_i\rangle v_i.\]
	For all $x \in \R^2_+$ and  $\varphi\in L^2(\R^2_+)$, we have
		\[|\mathbbm{1}_{(-\infty,\mu]}(\mathcal H_{\alpha,a})\varphi|^2(x)=\sum_i|\langle \varphi,v_i\rangle|^2|v_i(x)|^2 \leq  \|\varphi\|^2_{L^2(\R^2_+)}\sum_i |v_i(x)|^2.\]
		Since the sum is over a finite set and $(v_i)$ are in $L^2(\R^2_+)$, then  the dominated convergence theorem asserts that, for all $\epsilon>0$, there exists $r_\epsilon$ such that 
		\[\mbox{for all}\  \varphi\in L^2(\R^2_+),\int_{|x|\geq r_\epsilon} |\mathbbm{1}_{(-\infty,\mu]}(\mathcal H_{\alpha,a})\varphi|^2(x)\,dx\leq \epsilon \|\varphi\|^2_{L^2(\R^2_+)}.\]
		Hence  for all $\varphi \in \mathcal M_{r_\epsilon}$, it holds
		\begin{equation}\label{eq:chi5}
		\|\mathbbm{1}_{(-\infty,\mu]}(\mathcal H_{\alpha,a})\varphi\|^2_{L^2(\R^2_+)} \leq \epsilon \|\varphi\|^2_{L^2(\R^2_+)}.
		\end{equation} 
		Using the properties of the spectral projections, we have for all $\varphi \in  \mathcal M_{r_\epsilon}$,
		\begin{equation*}q_{\alpha,a}(\varphi)=q_{\alpha,a}\big(\mathbbm{1}_{(-\infty,\mu]}(\mathcal H_{\alpha,a})\varphi\big)+q_{\alpha,a}\Big(\big(I-\mathbbm{1}_{(-\infty,\mu]}(\mathcal H_{\alpha,a})\big)\varphi\Big).\end{equation*}
		The min-max principle and the definition of $\mathbbm{1}_{(-\infty,\mu]}(\mathcal H_{\alpha,a})$ ensure the boundedness of  $q_{\alpha,a}\Big(\big(I-\mathbbm{1}_{(-\infty,\mu]}(\mathcal H_{\alpha,a})\big)\varphi\Big)$ from below by $\mu\|\big(I-\mathbbm{1}_{(-\infty,\mu]}(\mathcal H_{\alpha,a})\big)\varphi\|^2$ (for $\varphi \neq 0$). In addition, $q_{\alpha,a}\big(\mathbbm{1}_{(-\infty,\mu]}(\mathcal H_{\alpha,a})\varphi\big)$ is  non negative, then for all $\varphi \in  \mathcal M_{r_\epsilon}$,
		\begin{equation}\label{eq:chi6}
		q_{\alpha,a}(\varphi)\geq \mu\|\big(I-\mathbbm{1}_{(-\infty,\mu]}(\mathcal H_{\alpha,a})\big)\varphi\|^2_{L^2(\R^2_+)}.
		\end{equation}
		On the other hand, we have
		\[\|\big(I-\mathbbm{1}_{(-\infty,\mu]}(\mathcal H_{\alpha,a})\big)\varphi\|^2_{L^2(\R^2_+)}=\|\varphi\|^2_{L^2(\R^2_+)}-\|\mathbbm{1}_{(-\infty,\mu]}(\mathcal H_{\alpha,a})\varphi\|^2_{L^2(\R^2_+)}.\]
		Hence, by~\eqref{eq:chi5} we get 
		\[\|\big(I-\mathbbm{1}_{(-\infty,\mu]}(\mathcal H_{\alpha,a})\big)\varphi\|^2_{L^2(\R^2_+)}\geq (1-\epsilon)\|\varphi\|^2_{L^2(\R^2_+)}.\]
		We use the above inequality together with~\eqref{eq:chi6} to obtain
		\[\frac {q_{\alpha,a}(\varphi)}{\|\varphi\|^2_{L^2(\R^2_+)}}\geq \mu(1-\epsilon).\]
		Since $\varphi \in \mathcal M_{r_\epsilon}$ is arbitrary, then
		\[\Sigma(\mathcal H_{\alpha,a}, r_\epsilon)\geq \mu (1-\epsilon).\]
		This establishes~\eqref{eq:second1} and consequently~\eqref{eq:second}.
	\end{proof}
	Next, we give the proof of Lemma~\ref{lem:out_B_r} which will also be used in the proof of Theorem~\ref{thm:ess_theta} below.
		\begin{proof}[Proof of Lemma~\ref{lem:out_B_r}]
		The main tool is a partition of unity that divides $\R^2_+$ into three sectors, which allows us to use spectral properties of some explored operators in Sections~\ref{sec:corner} and~\ref{sec:La}.
		One can find a partition of unity $(\hat{\chi}_j)$ for the interval $[0,\pi]$ satisfying
		\begin{multline*}
		\supp \hat{\chi}_1 \subset  \Big[0,\frac 23\alpha\Big],\ \supp \hat{\chi}_2 \subset \Big[\frac 13 \alpha,\frac 12\alpha+\frac \pi 2\Big],\ \supp \hat{\chi}_3 \subset  \Big[\frac 34 \alpha+\frac \pi 4,\pi\Big],\\
		\sum_{j=1}^3 \hat{\chi}_j^2(\theta)=1,\quad \sum_{j=1}^3 |\hat{\chi}_j^{'2}(\theta)|\leq C,\quad \forall \theta \in [0,\pi] ,		\end{multline*}
		where $C$ is a constant dependent on $\alpha$,
		 but independent of $a$. Let $r>0$. We define the truncation functions in polar coordinates
		\[\forall (\rho,\theta) \in \R_+\times(0,\pi), \qquad \chi_j^{r,\mathrm{pol}}(\rho,\theta)=\hat{\chi}_j(\theta),\]
		for $j\in\{1,2,3\}$. The associated functions in the Cartesian coordinates are defined by:
		\[\chi_j^r(x_1,x_2)=\chi_j^{r,\mathrm{pol}}(\rho,\theta),\qquad (x_1,x_2)\in \R^2_+.\]
		Consider a non-zero function $\varphi \in  \mathcal M_r$.
		The IMS localization formula ensures that
		\begin{equation}\label{eq:ims}
		\|(\nabla-i\Ab_{\alpha,a}) \varphi\|^2_{L^2(\R^2_+)} =\sum_{j=1}^3\|(\nabla-i\Ab_{\alpha,a})(\chi_j^r \varphi)\|^2_{L^2(\R^2_+)}-
		\sum_{j=1}^3\|\varphi |\nabla \chi_j^r| \|^2_{L^2(\R^2_+)}.
		\end{equation}
		We first evaluate the term $\sum_{j=1}^3\|\varphi |\nabla \chi_j^r| \|^2_{L^2(\R^2_+)}$. 
		For $(x_1,x_2) \in \R^2_+$, we have
		\begin{equation*}
		|\nabla \chi_j^r(x_1,x_2)|^2 =\big|\partial_\rho \chi_j^{r,\mathrm{pol}}(\rho,\theta)\big|^2+\frac 1{\rho^2}\big|\partial_\theta \chi_j^{r,\mathrm{pol}}(\rho,\theta)\big|^2=\frac 1{\rho^2}\big|\partial_\theta \chi_j^{r,\mathrm{pol}}(\rho,\theta)\big|^2.
		\end{equation*}
		By the construction of $ \chi_j^r$ and due to the support of $\varphi$, we get
		\begin{equation}\label{eq:chi13}\sum_{j=1}^3\|\varphi |\nabla \chi_j^r| \|^2_{L^2(\R^2_+)}\leq \frac C{r^2}\|\varphi\|^2_{L^2(\R^2_+)},
		\end{equation}
	 for some  $C=C(\alpha)$.
	Next, we bound $\sum_{j=1}^3\|(\nabla-i\Ab_{\alpha,a}) (\chi_j^r \varphi)\|^2_{L^2(\R^2_+)}$. The idea is to extend the functions $\chi_j^r \varphi$ by zero, to refer to the operators  introduced in the sections~\ref{sec:corner} and~\ref{sec:La}.
	Notice that $\curl \Ab_{\alpha,a}=\curl \Ab_0=1$ in the support of $\chi_1^{r} \varphi$, where $\Ab_0$ is the vector potential defined in~\eqref{canon}. Hence, extending $\chi_1^{r} \varphi$ by zero in the half-plane $\R^2_+$, and performing a suitable change of gauge, we get by the min-max principle 
	\begin{equation}\label{eq:chi9}
	\frac {\|(\nabla-i\Ab_{\alpha,a}) (\chi_1^{r} \varphi)\|^2_{L^2(\R^2_+)}}{\|\chi_1^{r} \varphi\|^2_{L^2(\R^2_+)}}
	\geq \inf_{\substack{u\in \dom \q_{\mathfrak b=1,\R^2_+}\\u \neq 0}} \frac {\| (\nabla-i\Ab_0)u\|^2_{L^2(\R^2_+)}}{\|u\|^2_{L^2(\R^2_+)}}
	=\Theta_0.
	\end{equation}
	(see Section~\ref{sec:corner}).
	Proceeding similarly and using a simple scaling, we get
	\begin{equation}\label{eq:chi10}
	\frac {\|(\nabla-i\Ab_{\alpha,a})(\chi_3^{r} \varphi)\|^2_{L^2(\R^2_+)}}{\|\chi_3^{r} \varphi\|^2_{L^2(\R^2_+)}}\geq |a|\Theta_0.
	\end{equation}
	Finally, we extend $\chi_2^{r} \varphi$ by zero in $\R^2$, and we perform a rotation of domain (by angle $\pi/2-\alpha$) and a suitable change of gauge to get
		\begin{equation}\label{eq:chi11}\frac {\|(\nabla-i\Ab_{\alpha,a}) (\chi_2^{r} \varphi)\|^2_{L^2(\R^2_+)}}{\|\chi_2^{r} \varphi\|^2_{L^2(\R^2_+)}}
		\geq \beta_a,
		\end{equation}
		where $\beta_a$ is the ground-state energy of the operator $\mathcal L_a$  in~\eqref{eq:ham_operator}.
		Gathering results in~\eqref{eq:chi9},~\eqref{eq:chi10} and~\eqref{eq:chi11} yields
		\begin{equation}\label{eq:chi12}
		\sum_{j=1}^3\|(\nabla-i\Ab_{\alpha,a}) (\chi_j^r \varphi)\|^2_{L^2(\R^2_+)}\geq|a|\Theta_0 \|\varphi\|^2_{L^2(\R^2_+)}.
		\end{equation}
		The last inequality follows from the fact that  $a \in [-1,1)\setminus \{0\}$, $\beta_a\geq|a|\Theta_0$ (see Section~\ref{sec:La}) and  $\sum_{j=1}^3 |\chi_j^r|^2=1$ in $\R^2_+$. Implementing~\eqref{eq:chi13} and~\eqref{eq:chi12} in~\eqref{eq:ims} completes the proof.	
	\end{proof}
		\begin{proof}[Proof of Theorem~\ref{thm:ess_theta}]
		We can equivalently prove that $\Sigma \mathcal H_{\alpha,a}=|a|\Theta_0$, now that we have Theorem~\ref{thm:ess_spect} in hand. This is  done in two steps:
		\paragraph{\itshape Step 1.} We prove $\Sigma \mathcal H_{\alpha,a}\geq |a|\Theta_0$. Let $r>0$, recall the definition of $\Sigma (\mathcal H_{\alpha,a},r)$.
		In light of Lemma~\ref{lem:out_B_r}, we get the following lower bound:
		\[\Sigma (\mathcal H_{\alpha,a},r)\geq |a|\Theta_0-\frac C {r^2}.\]	
		Taking $r \rightarrow +\infty$ in the inequality above establishes Step 1.
		\paragraph{\itshape Step 2.}  We prove $\Sigma \mathcal H_{\alpha,a}\leq |a|\Theta_0$. Let $\epsilon>0$ and $r>0$. The Neumann realization of the operator $-(\nabla-ia\Ab_0)^2$ in the half-plane $\R^2_+$ admits $|a|\Theta_0$ as a ground-state energy. Hence, the min-max principle together with a standard limiting argument ensure the existence of a constant $r>0$ and a function $f$, belonging to the form domain of $ (\nabla-ia\Ab_0)^2$ and vanishing outside $B(0,r)$,
		such that
		\[|a|\Theta_0\leq \frac {\| (\nabla-ia\Ab_0)f\|^2_{L^2(\R^2_+)}}{\|f\|^2_{L^2(\R^2_+)}}\leq |a|\Theta_0+\epsilon.\]
		Notice that $\curl \Ab_{\alpha,a}= \curl a\Ab_0=a$ in the set $D_\alpha^2$ defined in~\eqref{eq:A_alfa}. Hence, one may perform a translation and a change of gauge to obtain from $f$ a function $v$, supported in $ B_r^\complement \cap D_\alpha^2$ and satisfying  
		\[\frac {\| (\nabla-ia\Ab_0)f\|^2_{L^2(\R^2_+)}}{\|f\|^2_{L^2(\R^2_+)}}=\frac {\| (\nabla-i\Ab_{\alpha,a})v\|^2_{L^2(\R^2_+)}}{\|v\|^2_{L^2(\R^2_+)}}.\] 
		Consequently,
		\[\Sigma (\mathcal H_{\alpha,a},r)\leq \frac {\| (\nabla-i\Ab_{\alpha,a})v\|^2_{L^2(\R^2_+)}}{\|v\|^2_{L^2(\R^2_+)}} \leq |a|\Theta_0+\epsilon.\]
		Take successively $\epsilon$ to zero and $r$ to $+\infty$ to complete the proof of Step 2.
	\end{proof} 
	\begin{proof}[Proof of Lemma~\ref{lem:cont2}]
		Let $h \in \R$ such that $a+h \in [-1,1)\setminus\{0\}$.
		We prove that $\lim_{h \rightarrow 0} \mu(\alpha,a+h,r)=\mu(\alpha,a,r)$.
		Let $u \in \mathcal D_r$ such that $\|u\|_{L^2(\R^2_+)}=1$. We extend $u$ by zero outside the ball $B_r$, and we use the min-max principle together with Cauchy's inequality to write, 
		\begin{align*}
		\mu(\alpha,a,r)\leq q_{\alpha,a}(u)
		&\leq (1+|h|)q_{\alpha,a+h}(u)+C|h|^{-1}\int_{B_r}h^2(x_1^2+x_2^2)|u|^2\,dx\\
		&\leq (1+|h|)q_{\alpha,a+h}(u)+C(r)|h|,
		\end{align*}
		where $C(r)$ is a constant dependent solely on $r$. Again the min-max principle gives
		\[\mu(\alpha,a,r)\leq (1+|h|)\mu(\alpha,a+h,r)+C(r)|h|.\]
		Taking $h$ to zero, we get $\mu(\alpha,a,r) \leq \liminf_{h \rightarrow 0} \mu(\alpha,a+h,r)$. 
		
		In a similar fashion, we establish that $\mu(\alpha,a,r) \geq \limsup_{h \rightarrow 0} \mu(\alpha,a+h,r)$.	
	\end{proof}
	\section{Change of variables}
	\subsection{Frenet coordinates}\label{sec:bc}
	
	In this section we assume that the set $\Gamma$ consists of a simple smooth curve that intersects the boundary of $\Omega$ transversely in two points. In the general case, $\Gamma$ consists of a finite number of (disjoint) such curves. We may reduce to the simple case above by working on each component separately.
	We introduce some \emph{Frenet coordinates} which are valid in a tubular neighbourhood of $\Gamma$. These coordinates are known in the literature. We list below some of their basic properties. For more details, see~\cite[Appendix~F]{fournais2010spectral} and~\cite{Assaad2019}. 
	
	Let $\left[-|\Gamma|/2,|\Gamma|/2\right] \ni s\longmapsto M(s)\in\Gamma$ be  the arc length parametrization of $\Gamma$. Let $T(s)$ be a unit tangent vector to $\Gamma$ at the point $M(s)$, and $\nu(s)$ be the unit normal of $\Gamma$ at the point $M(s)$, pointed toward $\Om_1$.
	The orientation of the parametrization $M$ is fixed as follows:
	\[\mathrm{det}\big(T(s),\nu(s) \big)=1.\]
	The  curvature $k_r$ of $\Gamma$  is  defined by
	$T'(s)=k_r(s)\nu(s)$.
For $t_0>0$, we define the transformation
	\begin{equation}\label{Frenet}
	\Phi~:~\left(-\frac{|\Gamma|}2,\frac{|\Gamma|}2 \right)\times(-t_0,t_0)\, \ni (s,t)\longmapsto M(s)+t \nu(s) \in \R^2.
	\end{equation}
	For a  sufficiently small  $t_0$, $\Phi$ is a diffeomorphism from $\big(-|\Gamma|/2,|\Gamma|/2 \big)\times(-t_0,t_0)$ to $\Gamma(t_0)$,
	where $\Gamma(t_0):=\ImP$.
	The Jacobian of $\Phi$ is 
	\begin{equation}\label{eq:a1}
	\mathfrak a(s,t)=J_\Phi(s,t)=1-tk_r(s).
	\end{equation}
	The inverse, $\Phi^{-1}$, of $\Phi$ defines a system of coordinates  for the tubular neighbourhood $ \Gamma(t_0)$ of  $\Gamma$,
	\[\Phi^{-1}(x)=\big( s(x),t(x)\big).\]
	Note that since the curvature is bounded, then~\eqref{eq:a1} implies the  existence of $C>0$ such that
	\begin{equation}\label{eq:a2}
	\big|J_{\Phi^{-1}}(x)-1 \big|\leq C \ell\qquad \mathrm{and}\qquad \big|J_\Phi(s,t)-1 \big|\leq C \ell,
	\end{equation}
	where $x \in B(\ell)\subset \Gamma(t_0)$, $B(\ell)$ is a ball of radius $\ell$, and $(s,t)=\big(s(x),t(x)\big)$.  
	
	To each function $u \in  H_0^1\big(\Gamma(t_0) \big)$, we associate the function $\tilde{u} \in  H^1\big(\big(-|\Gamma|/2,|\Gamma|/2\big)\times(-t_0,t_0)\big)$ as follows:
	\begin{equation*}
	\tilde{u}(s,t)=u\big(\Phi(s,t)\big).
	\end{equation*}
	We also associate to any vector potential $\Eb=(E_1,E_2) \in H^1_\mathrm{loc}(\R^2,\R^2)$, the vector field
	$\tilde{\Eb}=(\tilde{E_1},\tilde{E_2}) \in H^1\big(\big(-|\Gamma|/2,|\Gamma|/2\big)\times(-t_0,t_0),\R^2\big)$,
	where 
	\begin{equation}\label{eq:A_tild1}
	\tilde{E_1}(s,t)=\mathfrak a(s,t)E\big(\Phi(s,t) \big)\cdot T(s)\quad\mathrm{and}\quad \tilde{E_2}(s,t)=E\big(\Phi(s,t) \big)\cdot\nu(s).
	\end{equation}
	We have the following change of variable formulae:
	\begin{equation}\label{eq:A_tild2}
	\int_{\Gamma(t_0)}\big|\big(\nabla-i \Eb \big)u \big|^2\,dx=\int_{-\frac{|\Gamma|}2}^{\frac{|\Gamma|}2}\int_{-t_0}^{t_0}\left(\mathfrak a^{-2}\big|(\partial_s-i\tilde{E_1})\tilde{u}\big|^2+\big|(\partial_t-i\tilde{E_2})\tilde{u}\big|^2 \right)\, \mathfrak a\,ds\,dt.
	\end{equation}	
	Finally, we present the following gauge transformation lemma: 
	
	\begin{lemma}\label{lem:Anew2}
		Let $a \in [-1,1)\setminus\{0\}$ and  $\overline{B_\ell}\subset (-|\Gamma|/2,|\Gamma|/2)\times(-t_0,t_0)$ be a ball of radius $\ell$ such that $\Phi(B_\ell)\subset \Om$.  If $\Eb$ is a vector potential in $H^1(\Omega,\R^2)$ with   
		$\curl \Eb=\mathbbm{1}_{\Omega_1}+a\mathbbm{1}_{\Omega_2}$, then there exists a function $\omega_\ell
		\in H^2(B_\ell)$  such that the vector potential $\tilde{\Eb}_{\rm g}:=\tilde{\Eb}-\nabla_{s,t}\omega_\ell$, defined in $B_\ell$, satisfies
		\begin{equation}\label{eq:Anew3}
		\big(\tilde{E}_{\rm g}\big)_1(s,t)= 
		\begin{cases}
		- \big(t-\frac {t^2}2 k_r(s)\big),&\mathrm{if}~t>0\\
		-a \big(t-\frac {t^2}2 k_r(s)\big),&\mathrm{if}~t<0
		\end{cases}
		;\qquad \big(\tilde{E}_{\rm g}\big)_2(s,t)=0.
		\end{equation}
	\end{lemma}
	\subsection{Coordinates near $\Gamma \cap \partial \Om$}\label{A:Psi}
In this section we will explicitly define the diffeomorphism $\Psi$ introduced in Section~\ref{sec:Psi}. The construction of $\Psi$ below is inspired by~\cite[Lemma 14.3]{bonnaillie2003analyse}. 

For $j\in\{1,...,n\}$, consider $\mathsf p_j \in \Gamma\cap\partial \Om$ and $\alpha_j$  the corresponding angle introduced in Notation~\ref{not:alfa}. We choose a system of coordinates such that $\mathsf p_j$ is the origin, there exists a neighbourhood of $\mathsf p_j$ where  $\partial \Om$ and $\Gamma$ coincide respectively with the representative curves of two smooth monotonous functions $f_1$ and $f_2$, defined in an interval $(-r_j,r_j)$ for a small $r_j>0$ and the following is satisfied:
\[f_1(0)=0,\ f_2(0)=0,\ f'_1(0)=-\tan\frac {\alpha_j}2,\ f'_2(0)=\tan\frac {\alpha_j}2,\] 
\[\Om \cap B(0,r_j)=E\cap B(0,r_j),\]
\begin{multline*} E:=\{(x_1,x_2)~:~x_1 \geq 0\ \mathrm{and}\ f_1(x_1)<x_2\leq f_2(x_1)\} \\\cup \{(x_1,x_2)~:~x_2 \geq 0\ \mathrm{and}\ f^{-1}_1(x_2)<x_1\leq f^{-1}_2(x_2)\}.\end{multline*}
We define the diffeomorphism $\breve{\Psi}$ in $B(0,r_j)$ by (see Figure~\ref{fig3}) 
\[\breve \Psi(x_1,x_2)=\Big(\frac {f_2(x_1)-f_1(x_1)}{2\tan\frac{\alpha_j}{2}},x_2-\frac {f_2(x_1)+f_1(x_1)}{2}\Big):=(\breve x_1,\breve x_2).\]
By performing a rotation of axes of an angle $-\alpha_j/2$, we can define out of $\breve \Psi$ a diffeomorphism $\Psi$ satisfying the desired conditions in Section~\ref{sec:Psi}.
\begin{figure}
	\begin{subfigure}{0.45\linewidth}
		\centering
		\includegraphics[scale=1]{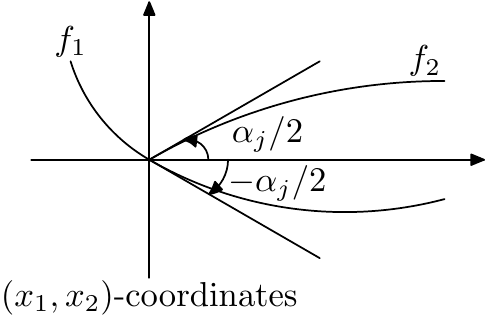}
	\end{subfigure}%
	\begin{subfigure}{0.45\linewidth}
		\centering
		\includegraphics[scale=1]{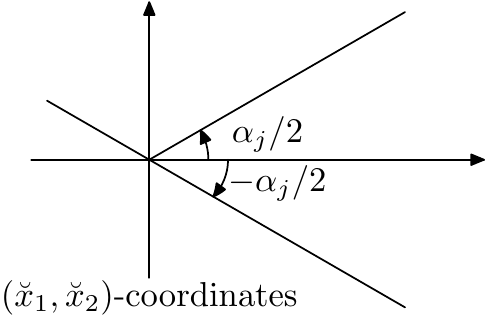}
	\end{subfigure}%
	\caption{Change of coordinates from $(x_1,x_2)$ to $(\breve x_1,\breve x_2)$.} 
	\label{fig3}
\end{figure}	
\section {Regularity properties}\label{sec:regularity}
Let $\mathfrak b>0$. Recall the operator $\mathcal P_{\mathfrak b,\Fb}$ and the associated quadratic form $Q_{\mathfrak b,\Fb}$, introduced  in~\eqref{eq:P0} and~\eqref{eq:Quad} respectively, where $\Fb$ is the vector potential in $\Hd$ satisfying $\curl \Fb=B_0={\mathbbm 1}_{\Om_1}+a{\mathbbm 1}_{\Om_2}$, $a\in[-1,1)\setminus\{0\}$.  In this section we prove the claim in~\eqref{eq:domP_F}  that the corresponding domains of $\mathcal P_{\mathfrak b,\Fb}$ and $Q_{\mathfrak b,\Fb}$ are independent of the parameter $\mathfrak b$.

A key-ingredient of the argument is the boundedness of the field $\Fb$. This boundedness is known for smooth fields, but it should be ensured for our potential with the piecewise-constant field $B_0$. As will be seen below, the fact that $\Fb \in \Hd$ and  $B_0\in L^p(\Om)$, for  $p\in [1,\infty]$, is sufficient for our needs.

\begin{theorem}\label{thm:regul}
	Let $a\in[-1,1)\setminus\{0\}$ and $\Fb \in \Hd$ be such that $\curl \Fb={\mathbbm 1}_{\Om_1}+a{\mathbbm 1}_{\Om_2}$, then $\Fb \in L^\infty(\Om)$.
\end{theorem}
\begin{proof}
	Since $\Fb \in \Hd$ and $\curl \Fb=B_0\in L^2(\Om)$ then $\Fb=(-\partial_{x_2}u,\partial_{x_1}u)$, where $u$ is the unique solution in $H^1_0(\Om)\cap H^2(\Om)$ of the Dirichlet problem for the Laplacian $-\Delta u=B_0$ (see~\cite[Propositions D.2.1 \& D.2.5]{fournais2010spectral} and~\cite[Theorem~9.15]{gilbarg2015elliptic}). 
	
	Now, notice that $B_0\in L^p(\Om)$, for all $p\in[1,+\infty]$. Consequently, for a fixed $p\in[2,+\infty)$ there exists a unique $v\in W_0^{1,p}(\Om)\cap W^{2,p}(\Om)$ satisfying $-\Delta v=B_0$ (\cite[Theorem~9.15]{gilbarg2015elliptic}). But $W_0^{1,p}(\Om)\cap W^{2,p}(\Om)\subset H_0^1(\Om)\cap H^2(\Om)$, thus $v=u$ and $\Fb=(-\partial_{x_2}v,\partial_{x_1}v)$.  Pick $p=4$,~\cite[(7.30)]{gilbarg2015elliptic} asserts that $v\in \mathcal C^1(\Om)$ and $\partial_{x_1}v$, $\partial_{x_2}v \in L^\infty(\Om)$. This completes the proof.
\end{proof}

\begin{proof}[Proof of~\eqref{eq:domP_F}] 
With $\Fb\in L^\infty(\Om)$ in hand, the proof is easy to establish. We will only derive the operator domain result in~\eqref{eq:domP_F}. Let $u\in \dom \mathcal P_{\mathfrak b,\Fb}$. We have
\[\Delta u=(\nabla-i\mathfrak b \Fb)^2u+2i\mathfrak b\Fb\cdot\nabla u+|\mathfrak b|^2|\Fb|^2u.\] 
Since $\Fb \in \Hd \cap L^\infty(\Om)$, we get that $\Delta u \in L^2(\Om)$ and $\nabla u\cdot \nu_{|\partial \Om}=0$. This ensures that $u\in H^2(\Om)$ (see~\cite[Theorem~E.4.7]{fournais2010spectral}). One can similarly establish the opposite inclusion; $\{u\in H^2(\Om)~:~\nabla u\cdot\nu|_{\partial \Om}=0\}\subset \dom \mathcal P_{\mathfrak b,\Fb}$. 
\end{proof}
	\bibliographystyle{alpha}
	\bibliography{WAbib}
	\end{document}